\documentclass{article}
\usepackage{graphicx} 
\usepackage{hyperref}
\usepackage{soul}
\usepackage{amsmath,amsfonts,amssymb,amsthm,mathrsfs,tikz-cd, enumerate,mathtools,authblk}
\usepackage{tikz}
\usepackage{subcaption}
\usetikzlibrary{calc}
\usepackage{cite}

\DeclarePairedDelimiterX\set[1]\lbrace\rbrace{#1}

\newcommand{\ra}{\rightarrow}

\newcommand{\RR}{{\mathbb R}}
\newcommand{\NN}{{\mathbb N}}
\newcommand{\ZZ}{{\mathbb Z}}
\newcommand{\CC}{{\mathbb C}}
\newcommand{\BB}{{\mathbb B}}

\newcommand{\CA}{{\mathcal A}}
\newcommand{\CB}{{\mathcal B}}
\newcommand{\cS}{{\mathcal S}}
\newcommand{\cC}{{\mathcal C}}

\newcommand{\CF}{{\mathcal F}}

\newcommand{\CP}{{\mathcal P}}
\newcommand{\CS}{{\mathcal C \mathcal S}}
\newcommand{\SCS}{{\mathcal S \mathcal C \mathcal S}}

\newcommand{\SA}{{\mathscr A}}
\newcommand{\SAl}{{\mathscr A}_{\ell}}
\newcommand{\SAal}{{\mathscr A}_{a\ell}}

\newcommand{\mfkDal}{{\mathfrak{D}}_{al}}

\newcommand{\mfkdl}{{{\mathfrak d}_{l}}}

\newcommand{\mfkDpal}{{\mathfrak D}^{\psi}_{al}}

\newcommand{\mfkDpQal}{{\mathfrak D}^{\psi,\derQ}_{al}}

\newcommand{\mfkdpal}{{{\mathfrak d}^{\psi}_{al}}}

\newcommand{\mfkDpalG}{{\mathfrak D}^{\psi,\bf G}_{al}}

\newcommand{\mfkg}{{\mathfrak g}}

\newcommand{\mfkF}{{\mathfrak F}}

\newcommand{\mfkM}{{\mathfrak M}}
\newcommand{\mfkG}{{\mathfrak G}}

\newcommand{\fj}{{\mathfrak j}}
\newcommand{\mfkgal}{{{\mathfrak g}_{al}}}

\newcommand{\mfkgalG}{{{\mathfrak g}^{\bf G}_{al}}}

\newcommand{\derB}{{\mathsf B}}

\newcommand{\derF}{{\mathsf F}}
\newcommand{\derG}{{\mathsf G}}
\newcommand{\derH}{{\mathsf H}}
\newcommand{\derQ}{{\mathsf Q}}

\newcommand{\derf}{{\mathsf f}}
\newcommand{\derg}{{\mathsf g}}

\newcommand{\derq}{{\mathsf q}}
\newcommand{\derp}{{\mathsf p}}
\newcommand{\dera}{{\mathsf a}}
\newcommand{\derb}{{\mathsf b}}

\DeclareMathOperator{\diam}{diam}

\newcommand{\Cech}{{\v{C}ech}}
\newcommand{\Hom}{\mathrm{Hom}}
\newcommand{\ad}{{\rm ad}}

\newtheorem{theorem}{Theorem}
\newtheorem{lemma}[theorem]{Lemma}
\newtheorem{definition}[theorem]{Definition}
\newtheorem{corollary}[theorem]{Corollary}

\newtheorem{prop}[theorem]{Proposition}
\newtheorem{remark}[theorem]{Remark}
\newtheorem{example}[theorem]{Example}

\newcommand{\fU}{{\mathfrak U}}
\newcommand{\fV}{{\mathfrak V}}

\newcommand{\frg}{{\mathfrak g}}
\newcommand{\Frechet}{{Fr\'{e}chet}}
\newcommand{\MC}{{\rm MC}}
\newcommand{\Sym}{{\rm Sym}}
\newcommand{\otr}{{\overline{\rm tr}}}
\newcommand{\pVectZ}{{{\mathsf {pVect}}_\ZZ}}
\newcommand{\com}{{\mathbf {com}}}
\newcommand{\CSnW}{{{\CS_n}/W}}
\newcommand{\clcomm}{{\overline{[\mfkM,\mfkM]}}}
\newcommand{\hI}{{\hat I}}
\newcommand{\hJ}{{\hat J}}
\newcommand{\rec}{{\rm rec}}

\DeclareMathOperator{\tr}{tr}
\DeclareMathOperator{\id}{id}

\title{A mathematical theory of topological invariants of quantum lattice systems}
\author{Adam Artymowicz}
\author{Anton Kapustin}
\affil{\normalsize California Institute of Technology, Pasadena, CA 91125}
\author{Bowen Yang}
\affil{\normalsize Center of Mathematical Sciences and Applications, Harvard University, Cambridge, MA 02138}

\begin{document}

\maketitle

\begin{abstract}
\noindent
    We show that Hall conductance and its non-abelian and higher-dimensional analogs are obstructions to promoting a symmetry of a state to a gauge symmetry. To do this, we define a local Lie system over a Grothendieck site as a pre-cosheaf of Lie algebras with additional properties and propose that a gauge symmetry should be described by such an object. We show that infinitesimal symmetries of a gapped state of a quantum lattice system form a local Lie system over a site of semilinear sets and use it to construct topological invariants of the state. Our construction applies to lattice systems on arbitrary asymptotically conical subsets of a Euclidean space including those which cannot be studied using field theory. 
\end{abstract}
\pagebreak
\tableofcontents

\section{Introduction}
\subsection{Overview}

A quantum lattice system is a mathematical model in quantum many-body physics in which degrees of freedom are arranged at points of a countable set $X\subset \RR^n$ (``the lattice'') and interact locally. Such systems are naturally described in the language of operator algebras.

We begin by recalling three basic structures associated with a quantum lattice system \cite{BR1,BR2}. First of all, we are given a $C^*$-algebra $\mathscr A$ of \emph{quasi-local observables}. This algebra admits local subalgebras $\mathscr A(Y) \subseteq \mathscr A$ for each subset $Y \subseteq X$. In the case of fermionic systems, $\mathscr A$ is $\ZZ_2$-graded. A \emph{state} on $\mathscr A$ is a positive normalized linear functional encoding expectation values. A state is called \emph{gapped} if it satisfies an appropriate energy stability condition. In this paper, we restrict attention to gapped states. Finally, one considers a Lie algebra of $*$-derivations of a dense $*$-subalgebra\footnote{For infinite-dimensional $\mathscr A$, important derivations often do not extend to the entire algebra.} of $\mathscr A$, whose elements are infinitesimal generators of dynamics.

For illustration, consider a bosonic system on a lattice consisting of a single point. Then $\mathscr A$ is a finite-dimensional simple $C^*$-algebra, and hence isomorphic to a matrix algebra $M_n(\mathbb C)$. States correspond to density matrices, and gapped states are precisely vector states, which form the projective space $\mathbb{CP}^{n-1}$. 

All derivations are inner and defined on the entire algebra, given by
\begin{equation}
    A \mapsto [F, A],
\end{equation}
for some $F \in M_n(\mathbb C)$. Moreover, each derivation admits a unique traceless representative, so that the Lie algebra of derivations is isomorphic to $\mathfrak{sl}_n$.

Since the $*$-structure on $M_n(\mathbb C)$ is given by the adjoint, a derivation is a $*$-derivation if and only if $F$ is skew-adjoint. Thus, the $*$-derivations correspond to the real Lie algebra $\mathfrak{su}_n$.

For a general lattice $\Lambda$, $*$-derivations need not be inner and may only be defined on a dense $*$-subalgebra of $\mathscr A$. Nevertheless, under suitable conditions, a $*$-derivation exponentiates to a one-parameter group of $*$-automorphisms of $\mathscr A$. More generally, a continuous path of $*$-derivations gives rise to a continuous path of $*$-automorphisms, which act on states by pullback and thus generate continuous families of states. This induces an equivalence relation: two states are considered equivalent if they are connected by such a path \cite{Moon2020, Bachmann_2011}. A central problem is to construct invariants of gapped states under this relation. At present, there are satisfactory constructions only in one spatial dimension \cite{BSY, Bourne2021}. Conjecturally, invariants of gapped states of lattice systems in higher dimensions are classified by Topological Quantum Field Theory (TQFT), but current tools seem inadequate to tackle this problem. 

More can be said for the equivariant variant, which deals with gapped states invariant under a locality-preserving symmetry. A symmetry of a state by a compact connected Lie group $G$ can be described infinitesimally by a Lie algebra homomorphism into the $*$-derivations of $\mathscr A$ that preserve the state. Equivalently, the corresponding $*$-automorphisms fix the state. Such states are called $G$-invariant. This leads to a $G$-equivariant version of the problem: to construct invariants of $G$-invariant gapped states under continuous paths through $G$-invariant gapped states. When $\Lambda$ is a single point, such invariants correspond to one-dimensional representations of $G$. When $\Lambda$ is a countable set which uniformly fills $\RR^n$ with $n>0$, invariants of $G$-invariant gapped states been recently constructed in \cite{LocalNoether} (see also \cite{ThoulessHall,Berry_Thouless}). They take values in $G$-invariant polynomials on the Lie algebra of $G$. In the special case $G=U(1)$ and $n=2$ one obtains the famous Hall conductance \cite{TKNN}. 

\subsection{Main Results and Outline}
In this paper, we develop a mathematical framework for defining invariants of gapped quantum lattice systems with Lie symmetries. In particular, we place the constructions of \cite{ThoulessHall,LocalNoether,Berry_Thouless} on a more conceptual footing.
A secondary goal is to generalize the construction in various directions. For example, we show how to define topological invariants of quantum lattice systems confined to well-behaved subsets of $\RR^n$. This generalization makes explicit that the invariants take values in a vector space which is determined by the large-scale geometry of the subset.  

The central idea of this paper is a new formulation of locality on a lattice, leading to the notion of a \textit{local Lie system}. Building on the ideas introduced in \cite{ThoulessHall,LocalNoether,Berry_Thouless}, we associate to any (possibly unbounded) region of the lattice a space of derivations that are approximately localized on that region, and show that for sufficiently regular regions these spaces behave well under natural operations such as the commutator. These spaces assemble into a cosheaf on a suitable \textit{site}, i.e.\ a category equipped with a Grothendieck topology, thereby producing a global geometric object encoding locality.

This structure provides a natural algebraic framework for symmetries. An on-site action of a Lie group gives rise to a representation by derivations, and we show that the invariants of \cite{LocalNoether} admit a purely algebraic description in this setting. Specifically, they arise as obstructions to promoting the symmetry to one acting by state-preserving derivations, with values in the homology of a differential graded Lie algebra (DGLA) naturally associated to the system. These obstruction classes pair with the \v{C}ech cohomology of the sphere at infinity to recover the invariants of \cite{LocalNoether}. This interpretation both clarifies their mathematical meaning and shows that they are independent of certain choices in their construction. It also highlights an analogy between topological invariants of gapped states and 't Hooft anomalies in Quantum Field Theory which are obstructions for gauging a symmetry action, see e.g. \cite{tHooftanomalies} for a review. 

The content of the paper is as follows. 
In Section \ref{sec:symmetries} we review quantum lattice systems and their symmetries. We introduce Lie algebras of derivations approximately localized on  regions in $\RR^n$ and prove some of their properties. In the next three sections we formalize these properties. 
In Section \ref{sec:locality} we axiomatize the notion of an infinitesimal local symmetry by defining local Lie systems abstractly in terms of cosheaves on a site. We also introduce the \Cech\ functor which assigns a DGLA to a local Lie system, and will serve as the main link between lattice geometry and algebraic topology.  In Section \ref{sec:semilinearsets} we identify a suitable category of regions in $\RR^n$, the category $\CS_n$ of \textit{fuzzy semilinear sets} which comes with a natural Grothendieck topology. In Section \ref{sec:LLA} we show that infinitesimal symmetries of any gapped state $\psi$ of a quantum lattice system on $\RR^n$ can be described by a local Lie system $\mfkDpal$ over $\CS_n$. We also study gapped states invariant under an action of a compact Lie group $G$ and define a $G$-equivariant analog of $\mfkDpal$. In Section \ref{sec:invariants} we construct invariants of $G$-invariant gapped states. The construction is along the lines of \cite{LocalNoether,Berry_Thouless} and uses an inhomogeneous Maurer-Cartan equation. We show that these invariants are obstructions for promoting the symmetry $G$ to a local symmetry of the state $\psi$. 
We also explain how to generalize the construction of invariants to lattice systems defined on sufficiently nice (asymptotically conical) subsets of $\RR^n$ and show that their invariants take values in a space which depends on the asymptotic geometry of the subset. This goes beyond what one can access using the TQFT heuristics.
In Appendix \ref{appendix:LGA} we isolate some proofs necessary for Section \ref{sec:symmetries}, and Appendix \ref{sec:deformation} develops some properties of the inhomogeneous Maurer-Cartan equation.
\subsection{Related work}
Recent progress has led to mathematically rigorous approaches to classifying gapped phases of quantum lattice systems, particularly in one and two dimensions \cite{Ogata1,Ogata2,Bourne2021,BSY,Ogata4,NikSop,ThoulessHall,Berry_Thouless,Bachmannetal,tensorHall}, and more generally in arbitrary dimension \cite{LocalNoether}. These works rigorously construct invariants of infinite-volume gapped states that are stable under suitable continuous deformations, and include well-known examples like the Hall conductance.

\subsection*{Acknowledgments}
We are grateful to Owen Gwilliam, Ezra Getzler, and Bas Janssens for discussions. The work of A. A. and A. K. was supported by the Simons Investigator Award. B. Y. would like to thank Yu-An Chen and Peking University for hospitality during the final stages of this work. A. K. would like to thank Yau Mathematical Sciences Center, Tsinghua University, for the same.

\section{Quantum lattice systems}\label{sec:symmetries}
This section is in part a review of quantum lattice systems on $\RR^n$  following  \cite{BR1,BR2} and \cite{LocalNoether}. We define the quasilocal algebra, its various subalgebras, and spaces of derivations. We then introduce spaces of derivations approximately localized on subsets of $\RR^n$ and derive some of their properties. We also study spaces of derivations preserving a gapped state. 
\subsection{Observables and derivations}\label{sec:3properties}

We use the $\ell^\infty$ metric on $\RR^n$, i. e. $d(x,y) := \max_{i=1,\hdots, n}|x_i-y_i|$. For any $U,V\subseteq \RR^n$ we write $\diam(U):= \sup_{x,y\in U}d(x,y)$ and $d(U,V) := \inf_{x\in U, y\in V}d(x,y)$. They take values in extended non-negative reals $[0,\infty]$. Thus $\diam(\varnothing)= 0$ and $d(U,\varnothing)=\infty$ for any $U\subseteq \RR^n$. For a nonempty set $U$ and $r\geq 0$ we define $U^r := \{x\in \RR^n: d(x,U)\le r\}$ while we set $\varnothing^r = \varnothing$.

\begin{definition}
    A quantum lattice system consists of a countable subset $\Lambda\subseteq \RR^n$ (``the lattice'') and a finite-dimensional $\ZZ_2$-graded complex Hilbert space $V_x$ for every $x\in\Lambda$
    We assume that the number of points of $\Lambda$ in a hypercube of diameter $d$ is upper bounded by $C(1+d)^n$ for some $C>0$.
    \footnote{In \cite{Berry_Thouless} $\Lambda$ was taken to be $\ZZ^n$ for simplicity. In \cite{LocalNoether} $\Lambda$ was assumed to be a Delone set, i.e. it was required to be uniformly filling and uniformly discrete. These assumptions were imposed on physical grounds. All the results proved in \cite{LocalNoether} hold under weaker assumptions adopted in this paper}.
\end{definition}
\begin{example}
    Models describing fermionic particles have $V_x$ which is an exterior algebra of another vector space: $V_x=\wedge H_x$, where $H_x$ is a  finite-dimensional Hilbert space (the space of one-particle states on site $x\in\Lambda$). The dimension of $H_x$ is referred to as the number of orbitals on site $x$. For translationally-invariant models, one typically takes all $H_x$ to be isomorphic to the same Hilbert space $H$, in which case $\dim H$ is called the number of bands.
\end{example}

For any $x\in\Lambda$ the algebra $\SA_x=\Hom_{\CC}(V_x,V_x)$ is a finite-dimensional $\ZZ_2$-graded $C^*$-algebra. For any nonempty bounded $X \subseteq \RR^n$ let $\SA(X) := \widehat\bigotimes_{x\in X\cap \Lambda}\SA_x$, where $\widehat\bigotimes$ denotes super-tensor product of $\ZZ_2$-graded algebras. For any $X\subseteq Y$ there is an inclusion $\SA(X)\hookrightarrow \SA(Y)$ and the algebras $\SA(X)$ form a direct system with respect to these inclusions. 
We extend this direct system to include the empty set by setting $\SA(\varnothing) = \CC$ and letting the inclusion $\SA(\varnothing)\hookrightarrow \SA(X)$ take $\alpha \mapsto \alpha\boldsymbol{1}$. 
Each $\SA(X)$ is a finite-dimensional $\ZZ_2$-graded $C^*$-algebra with the operator norm, and the inclusions $\SA(X)\hookrightarrow \SA(Y)$ preserve this norm. The normed $\ZZ_2$-graded *-algebra of {\it local observables} is 
\begin{align}
    \SAl=\SAl^{even}\oplus\SAl^{odd}=\varinjlim_X \SA(X).
\end{align}
The algebra of {\it quasi-local observables} $\SA$ is the norm-completion of $\SAl$; it is a $\ZZ_2$-graded $C^*$-algebra. 
\begin{example}\label{ex:algebraferm}
    For systems of fermionic particles on a lattice, $\SA_x$ as well as $\SA(X)$ for any bounded $X\subseteq \RR^n$ are complex Clifford algebras. For simplicity, assume that all $H_x$, $x\in\Lambda$,  are isomorphic to the same Hilbert space $H$ of dimension $N$. Choose an orthonormal basis in $H$. Then $\SA(X)$ is generated as a $\ZZ_2$-graded $*$-algebra by odd local observables $c_{\alpha,x}$, $\alpha\in\{1,\ldots,N\}$, $x\in X\cap\Lambda$. The relations are  
    \begin{align}\label{eq:CARrelations}
    c_{\alpha,x} c_{\beta,y}+c_{\beta,y} c_{\alpha,x}&=0,& \forall \alpha,\beta,\forall x,y\in X\cap\Lambda,\\
    c_{\alpha,x} c^*_{\beta,y}+c^*_{\beta,y} c_{\alpha,x}&=\delta_{\alpha\beta}\delta_{xy}, &\forall \alpha,\beta,\forall x,y\in X\cap\Lambda.
    \end{align}
    The algebra of quasi-local observables $\SA$ for systems of this type is known as a CAR algebra (the algebra of canonical anti-commutation relations) \cite{BR2,ArakiMoriya}.
\end{example}

For any bounded $X\subseteq \RR^n$, define the \textit{normalized even trace} $\otr: \SA(X)\to \CC$ as $\otr(\CA)=\tr(\CA)/{\sqrt {\dim(\SA(X))}}$. For any bounded $X\subseteq Y$ the \textit{partial trace} $\otr_{X^c}:\SA(Y)\to \SA(X)$ is uniquely specified by the condition 
\begin{align}
    \otr_{X^c}(\CA\otimes \CB) = \otr(\CA)\CB \label{eqn:partial-trace-defn}
\end{align} for any $\CA\in \SA(Y\backslash X)$ and $\CB\in \SA(X)$. In addition to forming a direct system with respect to inclusions, the spaces $\SA(X)$ are also an \textit{inverse} system with respect to the partial trace. $\otr$ extends to a normalized positive linear functional on $\SA$, i.e. a state. All these results are well known, for a proof see \cite{ArakiMoriya,NielsenChuang}.

We say that $\CA\in\SA$ is traceless if $\otr(\CA)=0$. The space of traceless anti-hermitian even elements of $\SA(X)$ will be denoted by $\mfkdl(X)$. $\mfkdl(X)$ is a real Lie algebra with respect to the commutator. Note that $\mfkdl(X)=0$ if $X\cap\Lambda=\emptyset$. The Lie algebras $\mfkdl(X)$ form a direct system over the directed set of bounded subsets of $\RR^n$, and its limit will be denoted $\mfkdl$. Equivalently, $\mfkdl$ is the Lie algebra of traceless anti-hermitian even elements of $\SAl$. Note that $\SAl^{even}=\CC{\mathbf 1}\oplus(\mfkdl\otimes\CC)$. 

So far, our discussion has been completely standard and mostly follows the monograph \cite{BR1, BR2}. A point of departure is the way we treat the derivations of $\SA$. In the $C^*$-algebraic approach, derivations represent generators of continuous  symmetries of a quantum lattice system. For example, a Hamiltonian is a generator of the time-translation symmetry and is described by a derivation. For quantum lattice systems, derivations of physical interest have the following form \cite{BR2}:
\begin{definition}\label{def:interaction}
    Let $\mathcal{P}(\Lambda)$ denote the set of finite subsets of $\Lambda$. An interaction is a function $\Phi:\mathcal{P}(\Lambda) \to \mfkdl$ such that $\Phi(X)\in\mfkdl(X)$ for each $X\in \mathcal{P}(\Lambda)$
    and such that for every $x\in \Lambda$ we have
        \begin{align}\label{eqn:locally-finite}
            \sum_{\substack{X\in\mathcal{P}(\Lambda)\\X\ni x}}\|\Phi(X)\| < \infty.
        \end{align}
    An interaction determines a derivation $\derF_\Phi$ defined by
    \begin{align}\label{eq:derint}
        \derF_\Phi: \CA\mapsto \sum_{X\in {\mathcal P}(\Lambda)} [\Phi(X),\CA],\quad\CA\in \SA.
\end{align}
    This is a (possibly unbounded) operator on $\SA$ whose domain is taken to be $\SAl$.
\end{definition} 
Physically interesting interactions are typically supported on an infinite subset of ${\mathcal P}(\Lambda)$, as we see in the examples below.
\begin{example}\label{ex:onsiteu1action}
An on-site action of $U(1)$ on a lattice system $(\Lambda,\{V_x\}_{x\in\Lambda})$ is a collection of continuous (and therefore smooth) homomorphisms $\rho_x:U(1)\ra U(V_x)^{even}$.
An on-site action of $U(1)$ on $(\Lambda,\{V_x\}_{x\in\Lambda})$ gives rise to a homomorphism from $U(1)$ to the automorphism group of $\SA$ via $g\mapsto\otimes_{x\in\Lambda}{\rm Ad}_{\rho_x(g)}$. The corresponding generator $\derQ$ is a derivation of the form
\begin{align}\label{eq:u1generatorex}
    \derQ: \CA\mapsto \sum_{x\in\Lambda} [\derq_x,\CA],\quad \CA\in\SA,
\end{align}
where $\derq_x\in\mfkdl(\{x\})$ is the traceless part of the generator of $\rho_x$. The domain of $\derQ$ is dense because it contains $\SAl$. It is easy to see that for an infinite $\Lambda$, $\derQ$ is unbounded even if we require the norms of $\derq_x$ to be bounded uniformly in $x$. 
\end{example}
\begin{example}\label{ex:u1generatorferm}
    Systems of fermionic particles (Example \ref{ex:algebraferm}) have a natural on-site $U(1)$ action (particle number symmetry). It is defined by saying that the weight-$\ell$ component of $V_x$ is the degree-$\ell$ component of the exterior algebra $\wedge H_x$. The generator of this $U(1)$ action has the form (\ref{eq:u1generatorex}) with $\derq_x$ being the grading operator on site $x$. More explicitly, suppose $H_x\simeq \CC^N$ for all $x\in\Lambda$. Choose an orthonormal basis for $H$, then $\SAl$ is generated as a $\ZZ_2$-graded $*$-algebra by odd local observables $c_{\alpha,x}\in\SA_x$, $\alpha\in\{1,\ldots,N\}$, $x\in\Lambda$, with the relations as in (\ref{eq:CARrelations}). 
    The charge derivation has the form (\ref{eq:u1generatorex}) with $\derq_x=\sqrt {-1}\sum_\alpha \left(c^*_{\alpha,x}c_{\alpha,x}-\frac{1}{2}\right)$. The term $-1/2$ in parentheses is needed to make $\derq_x$ traceless, while the overall $\sqrt {-1}$ makes it anti-hermitian. 
\end{example}
\subsubsection*{The brick decomposition}
The specification (\ref{eq:derint}) of a derivation in terms of an interaction $\Phi$ is many-to-one. For instance, for any $X,Y \supset Z \in \mathcal{P}(\Lambda)$ and $\dera \in \mfkdl(Z)$ we may redefine the interaction by setting $\Phi'(X) = \Phi(X) + \dera$ and $\Phi'(Y) = \Phi(Y) - \dera$, yielding the same derivation. It is useful to have a way of describing unbounded derivations that does not suffer from this ambiguity. This is the purpose of the brick decomposition \cite{LocalNoether} which we review now.

\begin{definition}
    A brick in $\RR^n$ is a non-empty subset of the form
    \begin{align}
        Y = \{(x_1,\ldots,x_n) \ \vert \ \ell_i \leq x_i \leq m_i, i=1,\ldots,n\} \label{eqn:brickdef}
    \end{align}
    where $(\ell_1,\ldots,\ell_n)$ and $(m_1,\ldots,m_n)$ are $n$-tuples of integers. We write $\BB_n$ for the set of all bricks in $\RR^n$. 
\end{definition} 
The set of bricks $\BB_n$ is a poset with respect to inclusion. In fact, it is a locally-finite distributive lattice. We denote the meet and join in $\BB_n$ by $\wedge$ and $\vee$, respectively. The set of bricks exhausts the collection of bounded subsets of $\RR^n$ in the sense that any bounded subset is contained in a brick. $\BB_n$ satisfies the following regularity property:
\begin{lemma}\label{lem:brick-sum}
    For any $z\in \RR^n$ we have
    \begin{equation}
                \sum_{Y\in\BB_n}(1+\diam(Y) + d(Y,z))^{-2n-2} \le \frac{\pi^44^n(n+1)^2}{36}
    \end{equation}
\end{lemma}
\begin{proof}
    Any pair of points $ x , y  \in \ZZ^n$ specifies a brick with $ x $ and $ y $ on opposing corners, and any brick can be specified this way (not uniquely). With $X$ the brick corresponding to $ x $ and $ y $ it is easy to see that $\max(d( x ,z),d( y ,z)) \le \diam(X)+d(X,z)$, and so $(1+ d( x ,z))(1+d( y ,z)) \le (1+\diam(X)+d(X,z))^2$. Thus we have
    \begin{align}
        \sum_{Y\in\BB_n}(1+\diam(Y) + d(Y,z))^{-2n-2} &\le \sum_{\substack{ x , y  \in \ZZ^n}}(1+d( x ,z))^{-n-1}(1+d( y ,z))^{-n-1} \notag\\
        &= \left(\sum_{ x  \in \ZZ^n}(1+d( x ,z))^{-n-1}\right)^2,
    \end{align}
    and it remains only to bound the above sum. Let $f(k):= (1+k)^{-n-1}$ and $g(k) := \#(\ZZ^n\cap B_{k}(z)) \le (1+2k)^n$. Using summation by parts we have
    \begin{align}
         \sum_{x\in\ZZ^n} (1+d( x ,z))^{-n-1} &\le \sum_{k\ge 0}f(k)(g(k+1)-g(k)) \notag\\
         &= \lim_{k\to \infty}f(k)g(k) - \sum_{k\ge 0}g(k)(f(k+1)-f(k)).
    \end{align}
    It is easy to check that $f(k)g(k) \to 0$ and that $-(f(k+1)-f(k)) \le (n+1)(1+k)^{-n-2}$, and so
    \begin{align}
         \sum_{x\in\ZZ^n} (1+d(x,z))^{-n-1} &\le 2^n(n+1)\sum_{k\ge 0}(1+k)^{-2} \notag\\
         &\le \frac{\pi^22^n(n+1)}{6}
    \end{align}
    which proves the Lemma.
\end{proof}

For any brick $Y$ we define the following subspace of $\mfkdl(Y)$:
\begin{align}
    \mfkdl^Y := \{ \CA \in \mfkdl(Y) \ \vert \ \otr_{X^c}(\CA)=0 \text{ for any brick } X \subsetneq Y\}.
\end{align}
Since each $\mfkdl(Y)$ decomposes as a direct sum $\mfkdl(Y) = \bigoplus_{X\subseteq Y}\mfkdl^X$ over bricks contained in $Y$, and since $\mfkdl = \bigcup_{Y\in\BB_n}\mfkdl(Y)$, we have a decomposition as the algebraic direct sum $\mfkdl = \bigoplus_{Y\in \BB_n}\mfkdl^Y$.
\begin{definition}
    For any $Y\in \BB_n$ define $\pi^Y: \mfkdl\to \mfkdl^Y$ as the orthogonal projection.
\end{definition}
The following lemma is often useful.
\begin{lemma}
    There is a function $\mu:\BB_n\times\BB_n\ra\ZZ$ such that for any $X\in\BB_n$ and any $\CA\in\mfkdl$ we have  
    \begin{align}
        \otr_{X^c}(\CA)&=\sum_{\substack{Y\subseteq X\\
        Y\in\BB_n}}\pi^Y(\CA),\\
        \pi^X(\CA)&=\sum_{\substack{Y\subseteq X\\
        Y\in\BB_n}}\mu(X,Y)\otr_{Y^c}(\CA).
    \end{align}
\end{lemma}
\begin{proof}
    The first statement is obvious. The second statement follows from the first one by letting $\mu$ to be the M\"{o}bius function \cite{Rota1964} of the locally-finite poset $(\BB_n,\subseteq)$. 
\end{proof}
    
Using this lemma and an explicit form of the M\"{o}bius function of $(\BB_n,\subseteq)$ one can show (see Appendix C of \cite{LocalNoether}):
\begin{lemma}\label{lma:brick-component-norm}
For any $\mathcal{A}\in \mfkdl$ and any $Y\in \BB_n$ we have $\|\pi^Y(\mathcal{A})\|\le 4^n \,\|\mathcal{A}\|$.
\end{lemma}
Let $\derF$ be a derivation of $\SA$ with domain $\SAl$. For any $X\in\BB_n$, the map $\otr_{X^c}\circ \derF|_{\SA(X)} :\SA(X)\to \SA(X)$ is a derivation (this follows from the defining property (\ref{eqn:partial-trace-defn}) of the partial trace). We therefore have $\otr_{X^c}\circ \derF|_{\SA(X)} = \ad_{\mathcal{A}_X}$ for some $\mathcal{A}_X\in \mfkdl(X)$, and we define 
\begin{align}
    \derF^X := \pi^X(\mathcal{A}_X).
\end{align}

\begin{definition}
    For any unbounded derivation with domain $\SAl$, the set of local observables $\{\derF^X\}_{X\in \BB_n}$ is the brick decomposition of $\derF$.
\end{definition}
The utility of the brick decomposition comes from its uniqueness:
\begin{lemma}\label{lem:brick-unique}
     Two unbounded derivations of $\SA$ with domain $\SAl$ and identical brick decompositions must coincide.
\end{lemma}
\begin{proof}
    An elementary calculation shows that whenever $X\subseteq Y$ we have $\otr_{X^c}(\mathcal{A}_Y) = \mathcal{A}_X$, and so we have
    \begin{align}
        \mathcal{A}_X = \bigoplus_{Z\subseteq X}\derF^Z
    \end{align}
    for every $X \in \BB_n$. Thus, for any $X\subseteq Y\in \BB_n$ and any $\mathcal{B}\in \SA(X)$, the quantity $\otr_{Y^c}(\derF(\mathcal{B})) = [\mathcal{A}_Y, \mathcal{B}]$ is determined by the brick decomposition of $\derF$. Since $\derF(\mathcal{B}) = \lim_{Y\uparrow \RR^n}\otr_{Y^c}(\derF(\mathcal{B}))$, it follows that $\derF(\mathcal{B})$ is determined by the brick decomposition of $\derF$.
\end{proof}
Lemma \ref{lem:brick-unique} allows us to identify an unbounded derivation $\derF$ with its set of brick components $\{\derF^X\}_{X\in \BB_n}$. In what follows, we will refer to a derivation and its brick decomposition interchangeably \footnote{However, note that not every element of $\prod_{X\in \BB_n}\mfkdl^X$ defines a derivation that is convergent on $\SAl$.}.
\begin{example}
    If $\derF = \derF_\Phi$ for some interaction $\Phi$ as in definition \ref{def:interaction}, then we have
    \begin{align}\label{eq:fromPhitoF}
        \derF^X = \sum_{Y\in \mathcal{P}(\Lambda)}\pi^X(\Phi(Y)),\quad \forall X\in\BB_n.
    \end{align}
    The sum is absolutely convergent by  (\ref{eqn:locally-finite}) and Lemma \ref{lma:brick-component-norm}.
\end{example}

\subsubsection*{Uniformly almost-local derivations}
As Examples \ref{ex:onsiteu1action} and \ref{ex:u1generatorferm} show, physically relevant derivations are unbounded and cannot be continuously extended to the whole quasilocal algebra $\SA$. To apply algebraic tools, one would like to define a space of derivations which is large enough to cover examples of interest and is a topological Lie algebra. 
 
We will use the brick decomposition to define such a class of derivations:

\begin{definition}
    For any element $\derF = \{\derF^Y\}_{Y\in\BB_n}$ of $\prod_{Y\in \BB_n}\mfkdl^Y$ we let 
    \begin{align}
        \|\derF\|_{k} := \sup_{Y\in\BB_n}\|\derF^Y\|(1+\diam(Y))^k. \label{eqn:bricknorm}
    \end{align}
\end{definition}
\begin{prop}
    Let $\derF \in \prod_{Y\in \BB_n}\mfkdl^Y$ and  $\|\derF\|_k < \infty$ for all $k\ge 0$. Then there exists an interaction $\Phi$ such that the brick components of $\derF_\Phi$ are given by $\derF$. 
\end{prop}
\begin{proof}
For any $X\in\CP(\Lambda)$ let \begin{align}
  \Phi(X) = \left\{\begin{array}{cc}
    \derF^{\lceil X\rceil} & \text{ if $\lceil X\rceil \cap\Lambda=X$} \\     0 & \text{ otherwise} 
\end{array}\right.
\end{align}
where $\lceil X\rceil\in\BB_n$ is the meet of all bricks containing $X$.
Clearly, $\Phi(X)\in \mfkdl(X)$. Also, for any $x\in\Lambda$ we have
\begin{align}
    \sum_{\substack{X\in\CP(\Lambda)\\X\ni x}}\|\Phi(X)\|\leq \sum_{\substack{Y\in\BB_n\\Y\ni x}}\|\derF^Y\|\leq \|\derF\|_{2n+2}\sum_{\substack{Y\in\BB_n\\Y\ni x}}(1+\diam(Y))^{-2n-2}\leq\infty,
\end{align}
where in the last step we used Lemma \ref{lem:brick-sum}. Thus, $\Phi$ is a well-defined interaction. Finally, using (\ref{eq:fromPhitoF}) it is easy to check that $\left(\derF_\Phi\right)^Y=\derF^Y$ for any $Y\in\BB_n$.
\end{proof}

\begin{definition} \label{defn:mfkDal}
    Define the space $\mfkDal$ as the set of elements $\derF\in \prod_{Y\in \BB_n}\mfkdl^Y$ with $\|\derF\|_{k} <\infty$ for all $k\ge 0$.
\end{definition}
In \cite{LocalNoether} $\mfkDal$ was called the space of {\it uniformly almost local} derivations (UAL derivations). From the physical viewpoint, UAL Hamiltonians correspond to interactions between sites of $\Lambda$ which decay faster than any power of distance. 
\begin{example}
    Let $\Lambda=\ZZ^n\subseteq\RR^n$. The derivation (\ref{eq:u1generatorex}) corresponds to $\derQ^Y$ being nonzero only for bricks of zero diameter, in which case $\derQ^{\{x\}}=\derq_x$. If $\derq_x$ is uniformly bounded, such $\derQ$ is a UAL derivation. If $\Lambda$ is an arbitrary lattice in $\RR^n$, a derivation of the form (\ref{eq:u1generatorex}) has $\derQ^Y$ nonzero only for $\diam Y\leq 1$. With our assumptions on $\Lambda$, $\derQ$ is a UAL derivation whenever $\derq_x$ is uniformly bounded. 
\end{example}
\begin{example}\label{ex:quadraticHam}
    Consider a system of fermionic particles as in Example \ref{ex:u1generatorferm} with all $H_x$ isomorphic to a fixed finite-dimensional Hilbert space $H$. If the particles are non-interacting, the corresponding Hamiltonian is described by a derivation of the form
    \begin{align}\label{eq:quadraticHam}
    \derF(h):\CA\mapsto \sum_{x,y\in\Lambda}\left[\sum_{\alpha,\beta} c^*_{\alpha,x}c_{\beta,y} h_{\alpha\beta}(x,y),\CA\right],
    \end{align}
    where $h_{\alpha\beta}(x,y)$ are complex numbers. One may regard these numbers as elements of an infinite square matrix whose rows and columns are labeled by a pair $(\alpha,x)$, or as matrix elements of an operator on $\ell^2(\Lambda,H)$.  $\derF$ is anti-hermitian iff this matrix is anti-hermitian.
    Equivalently, one can view the numbers $h_{\alpha\beta}(x,y)$ as defining a function $h:\Lambda\times\Lambda\ra \Hom_\CC(H,H)$. An anti-hermitian $\derF(h)$ is a UAL derivation if $h$ is bounded and $h(x,y)=O(|x-y|^{-\ell})$ for all $\ell\geq 0$. The charge derivation $\derQ$ from Example \ref{ex:u1generatorferm} also has the form (\ref{eq:quadraticHam}), with $h(x,y)=\delta_{xy}\in\Hom_\CC(H,H)$ for all $x,y\in\Lambda$. 
\end{example}

The functions $\|\cdot\|_k$ are norms. We show in the following that $\mfkDal$ equipped with these norms is a \Frechet\ -Lie algebra.  But first we generalize the definition of $\mfkDal$ by defining for every $U\subseteq \RR^n$ a real Lie algebra $\mfkDal(U)$ which consists of derivations approximately localized on $U$.

\begin{definition}
    For any element $\derF = \{\derF^Y\}_{Y\in\BB_n}$ of $\prod_{Y\in \BB_n}\mfkdl^Y$ and every $U\subseteq\RR^n$ we let 
    \begin{align}
        \|\derF\|_{U,k} := \sup_{Y\in\BB_n}\|\derF^Y\|(1+\diam(Y)+ d(U,Y))^k \label{eqn:U-norm}
    \end{align}
\end{definition}
\begin{definition}
    Define $\mfkDal(U)\subseteq \prod_{Y\in \BB_n}\mfkdl^Y$ as the set of elements $\derF$ with $\|\derF\|_{U,k} <\infty$ for all $k\ge 0$.
\end{definition}
If $U$ is empty, then for $k>0$ $\|\derF\|_{U,k}<\infty$ if and only if $\derF^Y=0$ for all $Y\in \BB_n$. Thus $\mfkDal(\emptyset)=\{0\}$. On the other hand we have $\mfkDal(\RR^n)=\mfkDal$.

It is easy to see that (\ref{eqn:U-norm}) is a norm on $\mfkDal(U)$ for each $k\ge 0$. We endow $\mfkDal(U)$ with the locally convex topology given by the norms (\ref{eqn:U-norm}) ranging over all $k\ge 0$. Recall that a topological vector space is called a \Frechet\ space if it is Hausdorff, and if its topology can be generated by a countable family of seminorms with respect to which it is complete.
\begin{prop}
    $\mfkDal(U)$ is a \Frechet\ space.
\end{prop}
\begin{proof}
    The Hausdorff property follows from the fact that if $\|\derF\|_{U,k} = 0$ for any $k\ge 0$ then $\derF = 0$. To show completeness, suppose $\{\derF_m\}_{m\in \mathbb{N}} \subseteq \mfkDal(U)$ is Cauchy, i.e. that for any $k\ge 0$ and any $\epsilon >0$ there is an $N\in \mathbb{N}$ such that $m,m'\ge N \implies \|\derF_m - \derF_{m'}\|_{U,k} < \epsilon$. For any fixed $Y\in \BB_n$ this implies that $\{\derF_m^Y\}$ is Cauchy in $\mfkdl^Y$ (with the operator norm) and thus converges to a limit $\derF^Y$. Let $\derF := \{\derF^Y\}_{Y\in \BB_n}$.

    Fix $k\in \mathbb{N}$ and $\epsilon>0$. For every $\ell=0,1,2,...$, choose $N_{\ell}\in \mathbb{N}$ so that $m,m'\ge N_{\ell}\implies \|\derF_m - \derF_{m'}\|_{U,k} < 2^{-\ell-1}\epsilon$. For any $Y\in\BB_n$, any $m\ge N_{1}$, and any $M>1$, we have 
    \begin{align}
        \|\derF_m^Y-\derF^Y\| &\le \|\derF_{m}^Y-\derF_{N_1}^Y\| + \sum_{i=1}^{M-1}\|\derF_{N_i}^Y-\derF_{N_{i+1}}^Y\|+ \|\derF_{N_M}^Y-\derF^Y\|\notag\\
        &\le \epsilon(1+\diam(Y) + d(U,Y))^{-k} + \|\derF_{N_M}^Y-\derF^Y\|.
    \end{align}
    Taking $M\to \infty$ shows that $\|\derF_m^Y-\derF^Y\|(1+\diam(Y) + d(U,Y))^{k} < \epsilon$. Since $Y$ was arbitrary, we have $\|\derF_m-\derF\|_{U,k} < \epsilon$. Since $k$ was arbitrary, $\{\derF_m\}$ converges to $\derF$ in the topology of $\mfkDal(U)$.
\end{proof}
\begin{definition}
    Let $U$ be a non-empty subset of $\RR^n$. For any $r\geq 0$ the $r$-thickening of $U$, denoted $U^r$, is $U^r := \{x\in \RR^n: d(x,U)\le r\}$.
\end{definition}
The norms $\|\cdot\|_{U,k}$ obey the following dominance relation.
\begin{lemma}\label{lma:seminorm-dominated}
    Let $U,V$ be subsets of $\RR^n$ and suppose that $U\subseteq V^r$. Then for any $\derF \in \mfkDal(U)$ we have
    \begin{align}
        \|\derF\|_{V,k} \le (r+1)^k\|\derF\|_{U,k}. \label{eqn:seminorm-dominated}
    \end{align}
    In particular, $\mfkDal(U)\subseteq \mfkDal(V)$ and the inclusion is continuous.
\end{lemma}
\begin{proof}
Let $Y\in\BB_n$. Then (\ref{eqn:seminorm-dominated}) follows from 
    \begin{align}
        1+ \diam(Y) + d(Y,V) &\le 1+ \diam(Y) + d(Y,V^r)+r \notag\\
        &\le  1+ \diam(Y) + d(Y,U)+r \notag\\
        &\le (r+1)(1+\diam(Y)+d(Y,U)),
    \end{align}
where in the second line we used the triangle inequality and in the third we used $U\subseteq V^r$.
\end{proof}
The above Lemma shows that the space $\mfkDal(U)$ only depends on the asymptotic geometry of the region $U$. More precisely, we have 
\begin{corollary} If $U\subseteq V^r$ and $V\subseteq U^r$ for some $r\geq 0$ then $\mfkDal(U)=\mfkDal(V)$ (as subsets of $\prod_Y \mfkdl^Y)$ and are isomorphic as \Frechet\ spaces. In particular, for any non-empty 
bounded $U \subseteq \RR^n$, the space $\mfkDal(U)$ coincides with $\mfkDal(\{0\})$. 
\end{corollary}
\subsubsection*{Behaviour under intersection}
For any two bounded sets $U,V$, a local observable is strictly localized on both $U$ and $V$ iff it is localized on their intersection, i.e. $\SA(U)\cap \SA(V) = \SA(U\cap V)$. An analogous relation for the spaces $\mfkDal(U)$ does not hold in general\footnote{Indeed, for any two bounded $U,V$ the spaces $\mfkDal(U)$ and $\mfkDal(V)$ coincide and are nontrivial but if $U$ and $V$ are disjoint then $\mfkDal(U\cap V) = 0$.}, but it does hold if we assume that $U$ and $V$ satisfy the following transversality condition:
\begin{definition}\label{defn:transverse}
    Let $U,V\subseteq \RR^n$ and $C>0$. We say $U$ and $V$ are $C$-transverse if 
    \begin{align}
        d(x,U\cap V) \le C\max(d(x,U), d(x,V))
    \end{align}
    for all $x\in \RR^n$.
\end{definition}
We will say $U$ and $V$ are transverse if they are $C$-transverse for some $C>0$.
\begin{example}
\begin{itemize}
\item[(i)] For any $U\subseteq\RR^n$, $U$ and $\RR^n$ are transverse.
\item[(ii)] If $U\cap V=\emptyset$, then $U$ and $V$ are not transverse.
\item[(iii)] If $U$ and $V$ are two intersecting straight lines in $\RR^2$, then $U$ and $V$ are transverse.
\item[(iv)] If $U=\{(x,0)\in\RR^2\}$ and $V=\{(x,x^{1/3})\in\RR^2\}$, then $U$ and $V$ are not transverse.
\end{itemize}
\end{example}

Transversality of two regions $U,V$ implies the following important compatibility between their \Frechet \ norms:
\begin{prop}\label{prop:pullback}
    If $U,V\subseteq \RR^n$ are $C$-transverse then
    \begin{align}
        \max(\|\derF\|_{U,k}, \|\derF\|_{V,k}) \le \|\derF\|_{U\cap V,k} \le (C+1)^k \max(\|\derF\|_{U,k}, \|\derF\|_{V,k}) \label{eqn:pullback}
    \end{align}
    for all $k>0$. In particular, $\mfkDal(U \cap V)$ is a topological pullback: it is the set $\mfkDal(U)\cap \mfkDal(V)$ with the topology of simultaneous convergence in $\mfkDal(U)$ and $\mfkDal(V)$.
\end{prop}
\begin{proof}
    The first inequality is true even without assuming transversality -- it follows from Lemma \ref{lma:seminorm-dominated}. For the second, note that we can assume without loss of generality that $U$ and $V$ are closed. Let $Z\in \BB_n$ and choose $x^*,y^*\in Z$ so that $d(x^*,U)=d(Z,U)$ and $d(y^*,V)=d(Z,V)$. Then we have
    \begin{align}
        d(Z,U\cap V) &= \inf_{z\in Z} d(z,U\cap V)\notag\\
        &\le C \inf_{z\in Z} \max(d(z,U),d(z,V))\notag\\
        &\le C\inf_{z\in Z} \max(d(z,x^*) + d(x^*, U), d(z,y^*) + d(y^*, V))\notag\\
        &\le C(\diam(Z) + \max(d(Z,U),d(Z,V))),
    \end{align}
    and thus
    \begin{align}
        1+ \diam(Z) + d(Z,U\cap V) \le (C+1)(1+ \diam(Z) + \max(d(Z,U),d(Z,V))), \label{eqn:stable-intersection-brick-version}
    \end{align}
    which proves (\ref{eqn:pullback}).
\end{proof}
The next proposition relates $\mfkDal(U\cup V)$ with $\mfkDal(U)$ and $\mfkDal(V)$ for any $U,V\subseteq\RR^n$.
\begin{prop}\label{prop:SES}
For any $U,V\subseteq \RR^n$ consider the sequence of vector spaces
\begin{align}
    \mfkDal(U \cap V) \xrightarrow{\alpha} \mfkDal(U)\oplus \mfkDal(V) \xrightarrow{\beta} \mfkDal(U \cup V) \to 0,\label{eqn:SAal-SES}
\end{align}
where $\alpha(\derF) = (\derF,-\derF)$ and $\beta(\derF,\derG) = \derF + \derG$.
\begin{enumerate}[i)]
    \item The sequence (\ref{eqn:SAal-SES}) admits a right splitting, i.e. a map $\gamma: \mfkDal(U \cup V) \to \mfkDal(U)\oplus \mfkDal(V)$ with $\beta\circ \gamma = \id$. In particular, it is exact on the right.
    \item If $U$ and $V$ are transverse, then (\ref{eqn:SAal-SES}) is exact on the left.
\end{enumerate}
\end{prop}
\begin{proof}
    $i)$. For any $Y\in \BB_n$, define
    \begin{align}
        \chi(Y):= &\left\{\begin{array}{cc}
            1 & \text{ if } d(Y,U) < d(Y,V) \\
            1/2 &  \text{ if }d(Y,U) = d(Y,V)\\
            0 & \text{ if } d(Y,U) > d(Y,V)
        \end{array} \right.
    \end{align}
    For any $\derF\in \mfkDal(U\cup V)$ define $\gamma_1(\derF) := \sum_{Y\in\BB_n}\chi(Y)\derF^Y$, and $\gamma_2(\derF) = \derF - \gamma_1(\derF)$. Then it is not hard to show that $\|\gamma_1(\derF)\|_{U,k}\le \|\derF\|_{U\cup V,k}$ and $\|\gamma_2(\derF)\|_{V,k}\le \|\derF\|_{U\cup V,k}$ and that $\gamma = (\gamma_1,\gamma_2)$ is a right splitting of (\ref{eqn:SAal-SES}).
    \\ \\ 
    $ii)$. Follows immediately from Proposition \ref{prop:pullback}.
\end{proof}

Next we will show how to endow $\mfkDal(U)$ with the structure of a Lie algebra.
\begin{prop}\label{prop:commutator}
    Let $U,V\subseteq \RR^n$. For any $\derF \in \mfkDal(U)$ and $\derG \in \mfkDal(V)$ the sum
    \begin{align}
        \derH^Z := \sum_{X,Y\in \BB_n}\pi^Z\left([\derF^X,\derG^Y]\right)
        \label{eqn:commutator-def-sum}
    \end{align}
    is absolutely convergent for every $Z \in \BB_n$. Here $\pi^Z$ is the orthogonal projection onto $\mfkdl^Z$ in $\mfkdl=\oplus_{Y\in\BB_n}\mfkdl^Y$. The brick components $\derH^Z$ define an element of $\prod_{Z\in \BB_n}\mfkdl^Z$ which we call $[\derF,\derG]$, which satisfies
    \begin{align}
        \|[\derF,\derG]\|_{U,k} &\le C 3^k\|\derF\|_{U,k+4n+4}\|\derG\|_{V,k+4n+4} \label{eqn:commutatorbound}
    \end{align}
    for some constant $C>0$ that depends only on $n$. The resulting bracket $[\cdot,\cdot]:\mfkDal(U)\times\mfkDal(V)\to \mfkDal(U)$ is bilinear, skew-symmetric, and satisfies the Jacobi identity.
\end{prop}
To prove Proposition \ref{prop:commutator} we will need several lemmas. 
\begin{lemma}\label{lem:brickjoin}
    For any $X,Y\in \BB_n$ with $X\cap Y \neq \varnothing$ we have 
    \begin{align}
        \diam(X\vee Y) \le \diam(X) + \diam(Y) \label{eqn:brickjoin}
    \end{align}
    and for any $z\in X\vee Y$ we have
    \begin{align}
        d(z,X) \le \diam(Y).
    \end{align}
\end{lemma}
\begin{proof}
    Let $\pi_i:\RR^n\to \RR$ be the projection onto the $i$th coordinate. The following identities hold for any bricks $X,Y\in \BB_n$:
    \begin{align}
        \pi_i(X\vee Y) &= \pi_i(X)\vee\pi_i(Y),\\
        \diam(X) &= \max_{i=1,\hdots n}\diam(\pi_i(X)),\\
        d(X,Y) &= \max_{i=1,\hdots n}d(\pi_i(X), \pi_i(Y)).
    \end{align}
    When $n=1$ the results are clear. When $n>1$ they follow from the $n=1$ case via the above identities.
\end{proof}

\begin{lemma}\label{lem:brick-commutator}
    Let $X,Y,Z \in \BB_n$ and let $\CA\in \mfkdl(X)$ and $\CB\in\mfkdl(Y)$. Let $\CA^X$ denote the component of $\CA$ in $\mfkdl^X$ and $\CB^Y$ denote the component of $\CB$ in $\mfkdl^Y$. Then the $\mfkdl^Z$-component of $[\CA^X,\CB^Y]$ vanishes unless $X\cap Y \neq \varnothing$ and $Z\subseteq X\vee Y$.
\end{lemma}
\begin{proof}
    The requirement that $X\cap Y\neq \varnothing$ is clear, since $\CA^X$ and $\CB^Y$ would commute otherwise. Suppose $Z\nsubseteq X\vee Y$. Then $Z' := (X\vee Y)\cap Z$ is a brick that is strictly contained in $Z$, so the $\mfkdl^Z$-component of $[\CA^X,\CB^Y]$ is in  $\mfkdl^Z\cap\mfkdl(Z')=\{0\}.$ 
\end{proof}
We make the following definitions for the next lemma. For any $U\subseteq \RR^n$ and any brick $X$ write $\tilde{d}(X,U) := 1+\diam(X)+ d(X,U)$ and $\tilde{d}(X):= 1+\diam(X)$.
\begin{lemma}\label{lem:doublebricksum}
    For any $U\subseteq \RR^n$, $Z\in \BB_n$, and $k \ge 0$ we have
    \begin{align}
    \sum_{\substack{X,Y\in \BB_n \\ X\cap Y\neq \varnothing \\ Z\subseteq X\vee Y}}\tilde{d}(X,U)^{-k-4n-4}\tilde{d}(Y)^{-k-4n-4} \le \frac{\pi^816^n(n+1)^43^k}{1296}\tilde{d}(Z,U)^{-k}.
\end{align}
\end{lemma}
\begin{proof}

Let $X,Y\in \BB_n$ with $X\cap Y \neq \varnothing$ and let $Z\subseteq X\vee Y$ be a brick. Pick an arbitrary point $w\in X\cap Y$. We have
\begin{align}
    d(Z,U) &\le d(Z,w) + d(w,U) \notag\\
    &\le d(Z,w) + \diam(X)  + d(X,U)\notag\\
    &\le \diam(X\vee Y) + \diam(X) + d(X,U)\notag\\
    &\le 2\diam(X) + \diam(Y) + d(X,U),
\end{align}
and so, since by Lemma \ref{lem:brickjoin} $\diam(Z) \le \diam(X\vee Y) \le \diam(X)+\diam(Y)$, we have
\begin{align}
    \tilde{d}(Z,U) &\le 1 + 3\diam(X) + 2\diam(Y) + d(X,U) \nonumber \\
    &\le 3(1+\diam(X) + d(X,U) + \diam(Y)) \nonumber \\
    &\le 3(1+\diam(X) + d(X,U))(1 + \diam(Y)) \nonumber \\
    &= 3\tilde{d}(X,U)\tilde{d}(Y)\label{eqn:claim1}.
\end{align}
It follows that for any $k\geq 0$ we have
\begin{align}
    \tilde{d}(Z,U)^{k}\tilde{d}(X,U)^{-k}\tilde{d}(Y)^{-k}&\le 3^k
\end{align}
and so
\begin{align}
    &\sum_{\substack{X,Y\in \BB_n \\ X\cap Y\neq \varnothing \\ Z\subseteq X\vee Y}}\tilde{d}(X,U)^{-k-4n-4}\tilde{d}(Y)^{-k-4n-4} \le \nonumber \\
    &\hspace{10mm} \le 3^k\tilde{d}(Z,U)^{-k}\sum_{\substack{X,Y\in \BB_n \\ X\cap Y\neq \varnothing \\ Z\subseteq X\vee Y}}\tilde{d}(X,U)^{-4n-4}\tilde{d}(Y)^{-4n-4} \label{eqn:prop-commutator1}
\end{align}
It remains to bound the sum on the right-hand side. Fix an arbitrary $z\in Z$ and let $X,Y \in \BB_n$ with $X\cap Y \neq \varnothing$ and $Z\subseteq X\vee Y$. Then by the second statement in Lemma \ref{lem:brickjoin} and the inequality $1+a+b\le (1+a)(1+b)$ for $a,b\ge 0$ we have
\begin{align}
    &(1+d(z,X) + \diam(X))(1+d(z,Y) + \diam(Y)) \le \nonumber \\
    &\hspace{15mm} \le (1+\diam(X))^2(1+\diam(Y))^2. \label{eqn:prop-commutator2}
\end{align}
Using (\ref{eqn:prop-commutator2}) we can bound the sum (\ref{eqn:prop-commutator1}) as follows:
\begin{multline}
    \sum_{\substack{X,Y\in \BB_n \\ X\cap Y\neq \varnothing \\ Z\subseteq X\vee Y}}\tilde{d}(X,U)^{-4n-4}\tilde{d}(Y)^{-4n-4} 
    \\\le \sum_{\substack{X,Y\in \BB_n \\ X\cap Y\neq \varnothing \\ Z\subseteq X\vee Y}}(1+\diam(X))^{-4n-4}(1+\diam(Y))^{-4n-4}\\
    \le \sum_{\substack{X,Y\in \BB_n}}(1+\diam(X)+d(X,z))^{-2n-2}(1+\diam(Y)+d(Y,z))^{-2n-2}\\
    \le \left( \sum_{\substack{X\in \BB_n}}(1+\diam(X)+d(X,z))^{-2n-2}\right)^2\le \frac{\pi^816^n(n+1)^4}{1296}.
\end{multline}
\end{proof}

Now we are ready to prove Proposition \ref{prop:commutator}.
\begin{proof}[Proof of Proposition \ref{prop:commutator}]
Recall $\tilde d(X,U):=1+\diam(X)+d(X,U)$ and $\tilde d(X):=1+\diam(X)$.

\medskip
Fix $Z\in\BB_n$. By Lemma \ref{lem:brick-commutator} we may restrict to pairs
$(X,Y)$ with $X\cap Y\neq\varnothing$ and $Z\subseteq X\vee Y$. Moreover, by Lemma \ref{lma:brick-component-norm}
\begin{align}
\|[\derF^X,\derG^Y]^Z\|\le 4^n\|[\derF^X,\derG^Y]\|\le 2^{2n+1}\|\derF^X\|\,\|\derG^Y\|.
\end{align}
Using the definition of $\|\cdot\|_{U,\cdot}$ and $\|\cdot\|_{V,\cdot}$ we have
\begin{align}
\|\derF^X\|&\le \|\derF\|_{U,k+4n+4}\,\tilde d(X,U)^{-k-4n-4}\\
\|\derG^Y\|&\le \|\derG\|_{V,k+4n+4}\,\tilde d(Y,V)^{-k-4n-4}.
\end{align}
Since $\tilde d(Y,V)=1+\diam(Y)+d(Y,V)\ge 1+\diam(Y)=\tilde d(Y)$, we also have
$\tilde d(Y,V)^{-k-4n-4}\le \tilde d(Y)^{-k-4n-4}$. Therefore
\begin{multline}
\sum_{X,Y\in\BB_n}\|[\derF^X,\derG^Y]^Z\|
=\sum_{\substack{X,Y\in\BB_n\\ X\cap Y\neq\varnothing\\ Z\subseteq X\vee Y}}
   \|[\derF^X,\derG^Y]^Z\| \\
\le 2^{2n+1}\|\derF\|_{U,k+4n+4}\|\derG\|_{V,k+4n+4}
\sum_{\substack{X,Y\in\BB_n\\ X\cap Y\neq\varnothing\\ Z\subseteq X\vee Y}}
\tilde d(X,U)^{-k-4n-4}\tilde d(Y)^{-k-4n-4}.
\end{multline}
By Lemma \ref{lem:doublebricksum}, the last sum is bounded by
\begin{align}
\frac{\pi^8 16^n (n+1)^4}{1296}\,3^k\,\tilde d(Z,U)^{-k}.
\end{align}
Hence
\begin{align}
\sum_{X,Y\in\BB_n}\|[\derF^X,\derG^Y]^Z\|
\le
\frac{\pi^8 2^{6n+1} (n+1)^4}{648}\,3^k\,
\|\derF\|_{U,k+4n+4}\|\derG\|_{V,k+4n+4}\,\tilde d(Z,U)^{-k}.
\end{align}
In particular, the defining series \eqref{eqn:commutator-def-sum} is absolutely convergent.
Moreover,
\begin{equation}
\|[\derF,\derG]^Z\|
\le \sum_{X,Y}\|[\derF^X,\derG^Y]^Z\|
\le
C\,3^k\,
\|\derF\|_{U,k+4n+4}\|\derG\|_{V,k+4n+4}\,\tilde d(Z,U)^{-k}
\end{equation}
with $C=\frac{\pi^8 2^{6n+1} (n+1)^4}{648}$.
Multiplying by $\tilde d(Z,U)^k$ and taking the supremum over $Z$ yields
\eqref{eqn:commutatorbound}. 

Bilinearity and antisymmetry of the bracket are immediate.
It remains to prove the Jacobi identity.
It suffices to do it on $\mfkDal(\RR^n)$, since $\mfkDal(U)\subseteq\mfkDal(\RR^n)$.
Let $\derF,\derG,\derH\in\mfkDal(\RR^n)$ and fix $W\in\BB_n$.
Expanding $[\derG,\derH]^Y$ by definition gives the formal identity
\begin{equation}\label{eqn:jacobi-expand}
[\derF,[\derG,\derH]]^W
=
\sum_{X,Y\in\BB_n}\sum_{X',Y'\in\BB_n}
\Bigl[\derF^X,\,[\derG^{X'},\derH^{Y'}]^Y\Bigr]^W.
\end{equation}
We now show that the quadruple series in \eqref{eqn:jacobi-expand} is absolutely convergent.
By Lemma \ref{lem:brick-commutator}, the term
$\bigl[\derF^X,[\derG^{X'},\derH^{Y'}]^Y\bigr]^W$ can be nonzero only if
\begin{align}
X\cap Y\neq\varnothing,\quad W\subseteq X\vee Y,
\qquad\text{and}\qquad
X'\cap Y'\neq\varnothing,\quad Y\subseteq X'\vee Y'.
\end{align}
Moreover, Lemma \ref{lma:brick-component-norm} implies
\begin{multline}
\Bigl\|\bigl[\derF^X,[\derG^{X'},\derH^{Y'}]^Y\bigr]^W\Bigr\|
\le
4^n\Bigl\|\bigl[\derF^X,[\derG^{X'},\derH^{Y'}]^Y\bigr]\Bigr\|
\le 2^{2n+1}\|\derF^X\|\,\|[\derG^{X'},\derH^{Y'}]^Y\|\\
\le 2^{4n+2}\|\derF^X\|\,\|\derG^{X'}\|\,\|\derH^{Y'}\|.
\end{multline}
Therefore
\begin{align}
&\sum_{X,Y}\sum_{X',Y'}
\Bigl\|\bigl[\derF^X,[\derG^{X'},\derH^{Y'}]^Y\bigr]^W\Bigr\|\notag\\
&\le
2^{4n+2}\sum_{\substack{X,Y\in\BB_n\\ X\cap Y\neq\varnothing\\ W\subseteq X\vee Y}}
\|\derF^X\|
\sum_{\substack{X',Y'\in\BB_n\\ X'\cap Y'\neq\varnothing\\ Y\subseteq X'\vee Y'}}
\|\derG^{X'}\|\,\|\derH^{Y'}\|.
\end{align}
Bound the inner sum using $\|\derG^{X'}\|\le \|\derG\|_{\RR^n,8n+8}\tilde d(X')^{-8n-8}$
and $\|\derH^{Y'}\|\le \|\derH\|_{\RR^n,8n+8}\tilde d(Y')^{-8n-8}$, and apply
Lemma \ref{lem:doublebricksum} with $U=\RR^n$ and $k=4n+4$ to obtain
\begin{equation}
\sum_{\substack{X',Y'\in\BB_n\\ X'\cap Y'\neq\varnothing\\ Y\subseteq X'\vee Y'}}
\|\derG^{X'}\|\,\|\derH^{Y'}\|
\le
C_1\,\|\derG\|_{\RR^n,8n+8}\|\derH\|_{\RR^n,8n+8}\,\tilde d(Y)^{-4n-4},
\end{equation}
for a constant $C_1$ depending only on $n$.
Plugging this into the outer sum and using
$\|\derF^X\|\le \|\derF\|_{\RR^n,4n+4}\tilde d(X)^{-4n-4}$ gives
\begin{multline}
\sum_{X,Y}\sum_{X',Y'}
\Bigl\|\bigl[\derF^X,[\derG^{X'},\derH^{Y'}]^Y\bigr]^W\Bigr\|
\le\\
C_2\|\derF\|_{\RR^n,4n+4}\|\derG\|_{\RR^n,8n+8}\|\derH\|_{\RR^n,8n+8}
\sum_{\substack{X,Y\in\BB_n\\ X\cap Y\neq\varnothing\\ W\subseteq X\vee Y}}
\tilde d(X)^{-4n-4}\tilde d(Y)^{-4n-4},
\end{multline}
and the remaining double sum is finite by Lemma \ref{lem:doublebricksum} with $U=\RR^n$ and $k=0$.
Hence the quadruple series \eqref{eqn:jacobi-expand} is absolutely convergent.

Absolute convergence justifies rearranging sums; in particular, since
\begin{equation}
    \sum_{Y\in\BB_n}[\derG^{X'},\derH^{Y'}]^Y=[\derG^{X'},\derH^{Y'}],
\end{equation} we obtain an absolutely
convergent triple expansion
\begin{equation}
    [\derF,[\derG,\derH]]^W=\sum_{X,Y,Z\in\BB_n}[\derF^X,[\derG^Y,\derH^Z]]^W,
\end{equation}
and similarly for the cyclic permutations.
For each fixed triple $(X,Y,Z)$, the ordinary commutator in the local algebra satisfies
the Jacobi identity:
\begin{equation} 
[\derF^X,[\derG^Y,\derH^Z]]+[\derG^Y,[\derH^Z,\derF^X]]+[\derH^Z,[\derF^X,\derG^Y]]=0.
\end{equation}
Taking the $W$-component and summing over $X,Y,Z$ 
gives the Jacobi identity for the bracket on $\mfkDal(\RR^n)$, hence on all $\mfkDal(U)$. This completes the proof.
\end{proof}

From Propositions \ref{prop:pullback} and \ref{prop:commutator} we immediately get
\begin{corollary}\label{cor:commutator}
    Suppose $U,V\subseteq \RR^n$.
    \begin{enumerate}[i)]
        \item $(\derF,\derG)\mapsto [\derF,\derG]$ is a jointly continuous bilinear map from $\mfkDal(U)\times\mfkDal(V)$ to $\mfkDal(U)\cap \mfkDal(V)$.
        \item If $U$ and $V$ are transverse, then this is a jointly continuous bilinear map from $\mfkDal(U)\times\mfkDal(V)$ to $\mfkDal(U\cap V)$.
    \end{enumerate}
\end{corollary}

\begin{example}\label{ex:bracketferm}
Consider a system of fermionic particles as in Example \ref{ex:quadraticHam}. UAL Hamiltonians of the form (\ref{eq:quadraticHam}) form a Lie subalgebra of $\mfkDal(\RR^n)$. When expressed in terms of infinite matrices, the Lie bracket is the matrix commutator: $[\derF(h),\derF(h')]=\derF([h,h'])$. The derivation $\derF(h)$ belongs to $\mfkDal(U)$ if $h(x,y)$ decays rapidly when either $x$ or $y$ are far from $U$. That is,
\begin{align}
    \sup_{x,y\in\Lambda}\|h(x,y)\|(1+d(x,y)+d(x,U)+d(y,U))^{\ell}<\infty,\quad\forall \ell\ge 0. 
\end{align}
\end{example}

\subsubsection*{Almost-local observables and anchored derivations}
Earlier in this section, UAL derivations were defined as unbounded derivations with domain the space $\SAl$ of strictly local observables. This was made necessary by the fact that in general a UAL derivation cannot be continuously extended to the entire space $\SA$ of quasilocal observables. It is natural to ask whether the use of unbounded operators can be avoided by finding a suitable subalgebra of $\SA$ such that UAL derivations act as bona-fide derivations on this algebra. It is indeed possible \cite{LocalNoether}, and below we define the space $\SAal$ of almost-local observables, equip it with a Fr\'echet topology, and show that UAL derivations act as continuous derivations of $\SAal$.

Pick any $o\in \ZZ^n$ and for each $k\ge 0$ define the following norm on $\SAl$:
\begin{align}\label{eqn:SAal-norms}
    \|\mathcal{A}\|_{k}:= \|\mathcal{A}\| + \sup_{r \in\NN\cup\{0\}}(1+r)^k\|\mathcal{A}-\otr_{B_o(r)^c}(\mathcal{A})\|.
\end{align}
An elementary argument (see \cite{LocalNoether}) shows that the norms resulting from different choices of $o\in \ZZ^n$ are equivalent.\footnote{In \cite{LocalNoether} a slightly different set of norms was used: $o$ was taken to be a lattice site, while $r$ was allowed to be an non-negative real number. The norms we use in this paper are equivalent but more convenient since $B_o(r)$ is a brick for all allowed choices of $o$ and $r$.}
\begin{definition}
    Let $\SAal$ be the completion of $\SAl$ in the topology generated by the norms $\|\cdot\|_{k}$ for $k \ge 0$.
\end{definition}
This is the space of \textit{almost-local} observables. It is shown in \cite[Appendix B]{LocalNoether} that it is a subalgebra of $\SA$. It is dense because it contains the algebra $\SAl$ of strictly local operators.
\begin{prop}\label{prop:summable}
    Suppose $U\subseteq \RR^n$ is non-empty and bounded and let $\derF \in \mfkDal(U)$. Then the sum $\sum_{X\in \BB_n}\derF^X$
    is norm-absolutely convergent and defines an embedding
    \begin{align}
        \Sigma: \mfkDal(U)\hookrightarrow \SAal
    \end{align} whose image is the subspace of traceless even anti-hermitian elements of $\SAal$.
\end{prop}
To prove Proposition \ref{prop:summable} we introduce the following Lemma
\begin{lemma}\label{lem:ball-containment}
    Let $X\in \BB_n$, $o\in \RR^n$, and $r>0$.
    \begin{enumerate}[i)]
        \item $1+r > 1 + \diam(X) + d(o,X) \implies X\subseteq B_o(r)$
        \item $1+ r \le  \frac{1}{4}(1+ \diam(X) + d(o,X)) \implies X\nsubseteq B_o(r)$
    \end{enumerate}
\end{lemma}
\begin{proof}
    The first statement follows from the inequality $d(o,x)\le d(o,X)+\diam(X)$ for every $x\in X$. For the second, let $x_1,x_2\in \RR^n$ be two points with $d(x_1,x_2) = \diam(X)$. By the triangle inequality, we have $d(o,x_1) + d(o,x_2) \ge \diam(X)$ so without loss of generality $d(o,x_1) \ge \diam(X)/2$. Since $d(o,x_1)\ge d(o,X)$, we thus have
    \begin{align}
        d(o,x_1) &\ge \frac{1}{2}\left(\frac{\diam(X)}{2} + d(o,X)\right) \notag\\
        &> \frac{1}{4}(1+ \diam(X) + d(o,X))-1,
    \end{align}
    which proves \emph{ii)}
\end{proof}
We now prove Proposition \ref{prop:summable}:
\begin{proof}[Proof of Proposition \ref{prop:summable}]
    We can assume without loss of generality that $U = \{o\}$ for some $o\in \RR^n$. We have
    \begin{align}
        \sum_{Y\in \BB_n}\|\derF^Y\| &\le \|\derF\|_{\{o\},2n+2}\sum_{Y\in\BB_n}(1+\diam(Y)+ d(Y,o))^{-2n-2},
    \end{align}
    so Lemma \ref{lem:brick-sum} the above sum is finite. To show that it is injective, notice that for any $X\in \BB_n$ we have $\pi^X(\Sigma(\derF)) = \derF^X$,
    and so $\Sigma(\derF)=0$ implies $\derF^X = 0$ for all $X\in\BB_n$.
    To show that $\Sigma$ is an embedding, we will show that for any $k\ge 0$ there are constants $C_k,C_k'>0$ such that
    \begin{align}
        \|\Sigma(\derF)\|_{k} \le C_k\|\derF\|_{\{o\},k+2n+2}\label{eqn:norm-equivalence-sigma1}
    \end{align}
    and 
    \begin{align}
        \|\derF\|_{\{o\},k} \le C_k'\|\Sigma(\derF)\|_{k}. \label{eqn:norm-equivalence-sigma2}
    \end{align}
    From the equality
    \begin{align}
        \Sigma(\derF) - \otr_{B_o(r)^c}(\Sigma(\derF)) = \sum_{X\nsubseteq B_o(r)}\derF^X, \label{eqn:partial-trace-expansion}
    \end{align}
    we get for any $r>0$ and any $k\ge 0$:
    \begin{align}
        &(1+r)^{k}\|\Sigma(\derF) - \otr_{B_o(r)^c}(\Sigma(\derF)\| \notag\\
        \le &(1+r)^k\sum_{X\nsubseteq B_o(r)}\|\derF^X\|\notag\\
        \le &\|\derF\|_{\{0\}, k+2n+2}(1+r)^k\sum_{X\nsubseteq B_o(r)}(1+\diam(X)+d(o,X))^{-k-2n-2}\notag\\
        \le &\|\derF\|_{\{o\}, k+2n+2}\sum_{X\nsubseteq B_o(r)}(1+\diam(X)+d(o,X))^{-2n-2}\notag\\
        \le &C_k\|\derF\|_{\{o\}, k+2n+2}
    \end{align}
    where we used Lemma \ref{lem:ball-containment} \emph{i)} in the second-last line and Lemma \ref{lem:brick-sum} in the last line.
    For the second inequality, using (\ref{eqn:partial-trace-expansion}) and Lemma \ref{lma:brick-component-norm} we have
    \begin{align}
        \|\derF^X\| &= \|\pi^X(\Sigma(\derF) - \otr_{B_o(r)^c}(\Sigma(\derF))\| 
        \le 4^n\|\Sigma(\derF) - \otr_{B_o(r)^c}(\Sigma(\derF)\|\label{eqn:r-ineq}
    \end{align}
    whenever $X$ is not contained in $B_o(r)$. Lemma \ref{lem:ball-containment} \emph{ii)} guarantees this for $1+ r = \frac{1}{4}(1 + \diam(X) + d(o,X))$. Thus (\ref{eqn:r-ineq}) gives
    \begin{align}
        (1+\diam(X)+d(o,X))^k\|\derF^X\| &\le 4^{n}(1+\diam(X)+d(o,X))^k\|\Sigma(\derF) - \otr_{B_o(r)^c}(\Sigma(\derF)\|\notag\\
        &= 4^{n+k}(1+r)^k\|\Sigma(\derF) - \otr_{B_o(r)^c}(\Sigma(\derF)\| \notag\\
        &\le 4^{n+k}\|\Sigma(\derF)\|_{k},
    \end{align}
    which proves (\ref{eqn:norm-equivalence-sigma2}).
    Finally, the last statement of the Proposition follows from the decomposition $\mathcal{A} = \sum_{X}\pi^X(\mathcal{A})$ for any traceless, even, anti-hermitian element $\mathcal{A}\in \SAal$.
\end{proof}

Proposition \ref{prop:summable} shows that $\mfkDal(U)$ for a bounded set $U$ coincides with the set of inner derivations of $\SAal$:
\begin{definition}\label{def:anchor}
    For $\derF \in \mfkDal(\RR^n)$, the following are equivalent:
    \begin{enumerate}[i)]
        \item $\derF\in \mfkDal(U)$ for a bounded set $U$
        \item There is a $\mathcal{A}\in \SAal$ such that $\derF(\mathcal{B}) = [\mathcal{A},\mathcal{B}]$ for every $\mathcal{B}\in \SAal$
    \end{enumerate}
    A derivation satisfying this condition is called anchored.
\end{definition}

\subsection{Automorphisms}\label{sec:LGAs}

In this section we recall certain automorphisms obtained by exponentiating elements of $\mfkDal(\RR^n)$ following \cite{LocalNoether}. One can develop the theory of such automorphisms that are almost-localized on regions in $\RR^n$ in a similar spirit to the above, but since we do not have much occasion to use them in this work, we opt instead for a more minimal development.
Let $\derF: \RR \to \mfkDal(\RR^n)$ be a smooth map. It is shown in \cite{LocalNoether} that for any $\CA \in \SAal$ the differential equation
\begin{align}
    \frac{d}{dt}\CA(t) = \derF(t)(\CA(t))
\end{align}
with the initial condition $\CA(0)=\CA$ has a unique solution $\CA(t)\in \SAal$ for all $t\in \RR$. While Ref. \cite{LocalNoether} only considered the case when all $V_x$ are even, the proof applies word for word in the $\ZZ_2$-graded case. Denote by $\alpha^\derF_t$ the map that takes $\CA$ to $\CA(t)$. It is a one-parameter family of continuous $*$-automorphisms of the algebra $\SAal$. By definition, for every $\CA\in\SAal$ it satisfies the  differential equation
\begin{align}\label{eq:Pexp}
    \frac{d}{dt}\alpha_t^\derF(\CA)=\alpha_t^\derF(\derF(t)(\CA)),
\end{align}
as well as the initial condition $\alpha_0^\derF={\rm id}$. 
\begin{definition}
    We call any automorphism of the form  $\alpha^\derF_1$ for some smooth map $\derF:[0,1]\ra\mfkDal(\RR^n)$ a locally-generated automorphism, or LGA for short.
\end{definition}
It is shown in \cite{LocalNoether}, Section 3.5, that the set of LGAs forms a group under composition. 

There is a natural notion of a smooth path between two LGAs $\alpha_0$ and $\alpha_1$. Namely, a smooth path of LGAs is a 1-parameter family of automorphisms $\alpha_t^\derF$, $t\in [0,1]$ and a smooth path of UAL derivations $\derF:[0,1]\ra\mfkDal$ which satisfy the equation (\ref{eq:Pexp}) for all $\CA\in\SAal$ and the conditions $\alpha_0^\derF=\alpha_0$, $\alpha_1^\derF=\alpha_1$. By the results of Section 3.5 of \cite{LocalNoether}, there is a smooth path between any two LGAs.

\subsection{States}
By a state $\psi$ of a quantum lattice system we will mean a state of the quasilocal algebra $\SA$ which vanishes on $\SA^{odd}$. If $\derF$ is an anchored  derivation (see Def. \ref{def:anchor}), we define $\psi(\derF) := \psi(\Sigma(\derF))$.
The group of LGAs acts on states by pre-composition, which we denote $\psi^\alpha := \psi \circ \alpha$. We say an LGA $\alpha$ preserves a state $\psi$ if $\psi^\alpha=\psi$. We say an element $\derF\in \mfkDal(\RR^n)$ preserves $\psi$ if $\psi(\derF(\CA))=0$ for any $\CA\in \SAal^{even}$, which is equivalent to the one-parameter group of automorphisms $t\mapsto \alpha_t^\derF$ corresponding to a constant map $t\mapsto\derF$ preserving $\psi$.
\begin{definition}
   For any $U\subseteq \RR^n$ define $\mfkDpal(U)$ as the set of all elements of $\mfkDal(U)$ that preserve $\psi$.
\end{definition}
It is easy to check that $\mfkDpal(U)$ is a closed subset of $\mfkDal(U)$, and that if $\derF$ and $\derG$ preserve $\psi$, then $[\derF,\derG]$ preserves $\psi$. Thus the analog of Propositions \ref{prop:pullback} and \ref{prop:commutator} and Corollary \ref{cor:commutator} hold for the spaces $\mfkDpal(U)$. Proposition \ref{prop:SES}, on the other hand, need not hold for the spaces $\mfkDpal(U)$ for a general state $\psi$. 

The following definition made in \cite{LocalNoether} is not essential for this paper, but is helpful to keep in mind.
\begin{definition}\label{def:smoothpathofstates}
    A smooth path of states is a 1-parameter family of states of the form $\psi_u=\omega\circ\alpha(u)$, where $\omega$ is a state and  $u\mapsto\alpha(u)$ is a smooth path of LGAs.
\end{definition}

In this paper, we will restrict our attention to gapped states, where quasiadiabatic evolution \cite{hastingshigher,kitaev2006anyons,osborne2007simulating} can be used to prove the analog of Proposition \ref{prop:SES}.
\begin{definition}
    A state $\psi$ is gapped if there exists $\derH\in \mfkDal(\RR^n)$ and $\Delta>0$ such that for any $\CA\in\SAl$  one has
    \begin{align}\label{eq:gapped}
        -i\psi(\CA^* \derH(\CA))\geq \Delta \left(\psi(\CA^*\CA)-\psi(\CA^*)\psi(\CA)\right).
    \end{align}
\end{definition} 
\begin{remark}
    The meaning of this condition becomes more transparent if one recalls that any $\derH\in\mfkDal$ is a generator of a one-parameter group of $*$-automorphisms of $\SA$. The condition (\ref{eq:gapped}) implies that $\psi$ is invariant under this one-parameter group of automorphisms, see \cite{BR2}, Section 5.3.3. Moreover, the corresponding one-parameter group of unitaries in the GNS representation of $\SA$ has a generator whose spectrum is contained in $\{0\}\cup[\Delta,+\infty)$. The interval $(0, \Delta)$ is called the spectral gap. The condition (\ref{eq:gapped}) also implies that $\psi$ is pure \cite{KapSopLSM}.
\end{remark}

\begin{remark}\label{rem:smoothfamilyofgapped}
In the definition of a smooth family of states, if $\omega$ is gapped, so is $\psi_u$ for every $u\in [0,1]$. Indeed, if $\omega$ is gapped with respect to $\derH\in\mfkDal$, then $\psi_u$ is gapped with respect to $\alpha(u)^{-1}(\derH)\in\mfkDal$, for the same value of $\Delta>0$. 
\end{remark}
\begin{remark}\label{rem:LGAsandphases}
    In \cite{LocalNoether} it was proposed to define a gapped phase as an orbit of a gapped state under the action of the group of LGAs. By virtue of Definition \ref{def:smoothpathofstates} and Remark \ref{rem:smoothfamilyofgapped}, two gapped states of a quantum lattice system are in the same phase iff they can be connected by a smooth path of states.  
\end{remark}
\begin{remark}
    The canonical state $\otr$ is not gapped for any $\derH\in\mfkDal(\RR^n)$. Indeed, let $\CA\in\SAl$. Using the fact that $\otr$ is a tracial state, one can easily check that the left-hand side of (\ref{eq:gapped}) is replaced with its opposite when $\CA$ is replaced with $\CA^*$. 
\end{remark}
\begin{example}\label{ex:Fockstate}
    Consider a system of fermionic particles  (Example \ref{ex:u1generatorferm}) with all $H_x$ isomorphic to $H$. Let $P:\Lambda\times\Lambda\ra\Hom_\CC(H,H)$ be a function parameterizing a UAL Hamiltonian as in Example \ref{ex:quadraticHam}. Let $\derF(P)$ be the corresponding UAL Hamiltonian. Suppose that $P$ is a projector when regarded as an operator on $\ell^2(\Lambda,H)$. For any $f\in\ell^2(\Lambda,H)$ let $c(f)=\sum_{x,\alpha} f_\alpha(x) c_{\alpha,x}$. There is a unique state $\psi$ of the CAR algebra satisfying 
    \begin{align}
        \psi(c(f)^*c(f))&=0,\quad\forall f\in {\rm im}\,P,\\
        \psi(c(f)c(f)^*)&=0,\quad\forall f\in {\rm im}\,(1-P).
    \end{align}
    This state of the CAR algebra is known as a Fock state, see \cite{BR2} for details. One can show that $\psi$ is a gapped state with gap $\Delta=1$ for the Hamiltonian  $\derH=\derF(P)$. One can show that a UAL derivation of the form (\ref{eq:quadraticHam}) belongs to $\mfkDpal$ if and only if the corresponding operator $h$ commutes with $P$.
\end{example}

In Appendix \ref{appendix:LGA} we prove the following.
\begin{prop}\label{prop:quasiadiabatic}
    Suppose $\psi$ is gapped, and the corresponding Hamiltonian is $\derH$. Then there are linear functions 
    \begin{align}
        \mathcal{J}&:\mfkDal(\RR^n) \to \mfkDpal(\RR^n)\\
        \mathcal{K}&:\mfkDal(\RR^n)\to \mfkDal(\RR^n)
    \end{align}such that
    \begin{enumerate}[i)]
        \item If $\derF$ preserves $\psi$ then $\mathcal{K}(\derF)$ preserves $\psi$.
        \item For every $k>0$, $U\subseteq \RR^n$, and $\derF\in \mfkDal(U)$ we have
        \begin{align}
            \|\mathcal{J}(\derF)\|_{U,k} &\le C_k \|\derF\|_{U,k+4n+3}\\
            \|\mathcal{K}(\derF)\|_{U,k} &\le C'_k \|\derF\|_{U,k+4n+3}
        \end{align}
        for some constants $C_k,C'_k$ depending only on $k,n, \derH$, and $\Delta$.
         \item For every $\derF\in \mfkDal(\RR^n)$ we have
         \begin{align}
             \derF = \mathcal{J}(\derF) - \mathcal{K}([\derH, \derF]).
         \end{align}
    \end{enumerate}
\end{prop}
Using the above Proposition we will prove the analog of Proposition \ref{prop:SES} for the spaces $\mfkDpal(U)$.
\begin{prop}\label{prop:SES-psi}
    For any $U,V\subseteq \RR^n$ consider the sequence
    \begin{align}
    \mfkDpal(U \cap V) \xrightarrow{\alpha} \mfkDpal(U)\oplus \mfkDpal(V) \xrightarrow{\beta} \mfkDpal(U \cup V) \to 0,\label{eqn:SAal-SES2}
    \end{align}
    where $\alpha(\derF) = (\derF,-\derF)$ and $\beta(\derF,\derG) = \derF + \derG$.
    \begin{enumerate}[i)]
    \item If $\psi$ is gapped then the sequence (\ref{eqn:SAal-SES2}) admits a right splitting, and in particular it is exact on the right
    \item If $U$ and $V$ are transverse, then (\ref{eqn:SAal-SES2}) is exact on the left.
    \end{enumerate}
\end{prop}
\begin{remark}
    This Proposition is the only place where we use the existence of the Hamiltonian with respect to which $\psi$ is gapped. In particular, the precise form of this Hamiltonian does not matter, provided it is in $\mfkDal(\RR^n)$. 
\end{remark}
To prove Proposition \ref{prop:SES-psi} we will need the following geometric result
\begin{lemma}\label{lem:voronoi-stable}
    Let $U,V \subseteq \RR^n$ be closed and define $U':= \{x\in \RR^n: d(x,U)\le d(x,V)\}$. Then $U'$ and $U\cup V$ are transverse and their intersection is $U$.
\end{lemma}
\begin{proof}
    It is easy to check that $U'\cap (U\cup V) = U$. To prove transversality we will show
    \begin{align}
        d(x,U) \le 4\max(d(x,U'),d(x,U\cup V)) \label{eqn:voronoi-stable}
    \end{align}
    for every $x\in \RR^n$. Suppose first that $d(x,U)\le 2d(x,V)$. Then
    \begin{align}
        d(x,U) &\le 2\min(d(x,U), d(x,V))\notag\\
        &= 2d(x,U\cup V)
    \end{align}
    which implies (\ref{eqn:voronoi-stable}). Suppose instead that $d(x,U) > 2d(x,V)$, and let $y\in U'$ satisfy $d(x,y) = d(x,U')$. Notice $x\notin U'$ and so $y$ lies in the boundary of $U'$, which implies $d(y,U)=d(y,V)$. Thus we have
    \begin{align}
        d(x,U) &\le d(x,y) + d(y,U)\notag\\
        &= d(x,y) + d(y,V) \notag\\
        &\le 2d(x,y) + d(x,V)
    \end{align}
    where in the first and third lines we used the triangle inequality. Using $d(x,y)=d(x,U')$ and $d(x,V)< d(x,U)/2$, this gives $d(x,U)< 4d(x,U')$, which implies (\ref{eqn:voronoi-stable}).
\end{proof}

\begin{proof}[Proof of Proposition \ref{prop:SES-psi}]
    The proof of Proposition \ref{prop:SES} goes through unmodified except for the definition of $\gamma$, which needs to be changed to ensure that the image of $\gamma$ consists of derivations that preserve $\psi$. Suppose $U,V\subseteq\RR^n$ and $\derF \in \mfkDpal(U\cup V)$. Without loss of generality we may assume that $U$ and $V$ are closed. Define
    \begin{align}
        U' &:= \{x\in \RR^n: d(x,U)\le d(x,V)\},\\
        V' &:= \{x\in \RR^n: d(x,V)\le d(x,U)\}.
    \end{align}
    Let $\gamma^{U,V}$ (resp. $\gamma^{U',V'}$) be the splitting from Proposition \ref{prop:SES} with the sets $U$ and $V$ (resp. $U'$ and $V'$). Define $\tilde{\gamma} = (\tilde{\gamma}_1, \tilde{\gamma}_2)$ as
    \begin{align}
        \tilde{\gamma}_i(\derF) &:= \mathcal{J}(\gamma_i^{U,V}(\derF)) - \mathcal{K}([\mathcal{J}(\gamma_i^{U',V'}(\derH)),\derF])
    \end{align}
    for $i=1,2$. Using Prop. \ref{prop:quasiadiabatic} and Lemma \ref{lem:voronoi-stable} and the fact that the commutator of two derivations that preserve $\psi$ preserves $\psi$, one checks that $\tilde{\gamma}$ takes $\mfkDpal(U\cup V)$ to $\mfkDpal(U)\oplus \mfkDpal(V)$.
    Using Prop. \ref{prop:quasiadiabatic} $iii)$ and the fact that $\gamma_1^{U,V}(\derF)+\gamma_2^{U,V}(\derF) = \derF$ and $\gamma_1^{U',V'}(\derH)+\gamma_2^{U',V'}(\derH) = \derH$, we get $\tilde{\gamma}_1(\derF) + \tilde{\gamma}_2(\derF) = \derF$, as desired.
\end{proof}

\begin{example}
    Consider a Fock state $\omega$ of a CAR algebra, Example \ref{ex:Fockstate}, corresponding to a projector $P$ on $\ell^2(\Lambda,H)$. Let $U_i$, $i\in I$, be a finite set of mutually transverse closed sets such that $\cup_i U_i=\RR^n$. This gives a covering of $\RR^n$. Let $\derQ$ be the generator of the canonical $U(1)$ action on the CAR algebra, Example \ref{ex:u1generatorferm}. Since $\omega$ is $U(1)$-invariant, $\derQ\in\mfkDal^\omega(\RR^n)$, hence by Prop. \ref{prop:SES-psi} there exist $\derQ_i\in\mfkDal^\omega(U_i)$ such that $\sum_i \derQ_i=\derQ$. 
    In this special case one can exhibit such $\derQ_i$ explicitly. For every $i\in I$, let $\chi_i:\Lambda\ra \{0,1\}$ be the indicator function of $U_i$. For every $x\in\Lambda$ let $n_x=\sum_{i\in I} \chi_i(x)$. We define the functions $f_i:\Lambda\ra\RR$, $i\in I$, by 
    \begin{align}\label{eq:fi}
    f_i(x)=\frac{\chi_i(x)}{n_x},\quad\forall x\in\Lambda,
    \end{align}
    so that $\sum_i f_i=1$. We can regard each $f_i$ as a bounded operator on $\ell^2(\Lambda,H)$.
    Now we let $\derQ_i$ be the derivation $\derF(h_i)$ (see Example \ref{ex:quadraticHam}) corresponding to the following operators $h_i$ on $\ell^2(\Lambda,H)$:
    \begin{align}\label{eq:hi}
        h_i=P f_i P+(1-P) f_i (1-P).
    \end{align}
    It is easy to see that $\sum_{i\in I} h_i=1,$ hence $\sum_{i\in I}\derF(h_i)=\derF(1)=\derQ$. One can also check that $\derQ_i\in\mfkDal^\omega(U_i)$. 
\end{example}

\section{Local Lie systems}\label{sec:locality}

In this section we define the central notion of this paper: local Lie systems over a site. Physically, this notion formalizes what physicists call gauge Lie algebras. Mathematically, a local Lie system is an  algebraic object that encodes the compatibility conditions between the spaces $\mfkDal(U)$ and $\mfkDpal(U)$ attached to regions $U \subseteq \mathbb{R}^n$. That is, local Lie systems are an axiomatization of Lemma \ref{lma:seminorm-dominated}, Proposition \ref{prop:SES} (resp. Proposition \ref{prop:SES-psi}), and Corollary \ref{cor:commutator} for the spaces $\mfkDal(U)$ (resp. $\mfkDpal(U)$). Our axiomatization involves a pre-cosheaf of Lie algebras over a suitable category of subsets of the space. The space can be either a discrete metric space, as in this paper, or a manifold, as in field theory. A basic construction which applies to any pre-cosheaf of vector spaces is the \Cech\ functor. We show in subsection \ref{sec:DGLA-attached-to-lls} that in the case of local Lie systems it produces a Differential Graded Lie algebra (DGLA). This DGLA will be used later on to construct topological invariants of gapped states.

\subsection{Locality and (pre-)cosheaves}

Let $M$ be a manifold and $Open(M)$ be the category whose objects are open subsets of $M$, and the set of morphisms from an open $U$ to an open $V$ is the singleton or the empty set depending on whether $U\subseteq V$ or $U\nsubseteq V$. Composition of morphisms is uniquely defined. A pre-cosheaf $\mfkF$ on $M$ with values in a category $\cC$ is a functor $\mfkF:Open(M)\ra\cC$. 
Thus for every inclusion of opens $U\subseteq V$ one is given a co-restriction morphism $e_{VU}:\mfkF(U)\ra \mfkF(V)$ such that for any three opens $U\subseteq V\subseteq W$ one has $e_{WV}\circ e_{VU}=e_{WU}$.
Pre-cosheaves (as well as pre-sheaves, which are functors from the opposite category of $Open(M)$ to $\cC$) can be used to describe local data on $M$. This form of locality is rather weak, since it does not require $\mfkF(U\cup V)$ to be expressible through $\mfkF(U)$ and $\mfkF(V)$.

As an example, consider the Lie algebra of gauge transformations, i.e. the Lie algebra $\mfkG(M):=C^\infty(M,\mfkg)$ of smooth functions on $M$ with values in a finite-dimensional Lie algebra $\mfkg$. It is a global object attached to $M$. To ``localize'' it, for any open $U\subseteq M$ we define the Lie algebra  $\mfkG(U)$ to be the space of smooth $\mfkg$-valued functions on $M$ whose closed support is contained in $U$. In particular, $\mfkG(\varnothing)=0$. For any inclusion of opens $U\subseteq V$ we have a homomorphism of Lie algebras $\iota_{VU}:\mfkG(U)\ra \mfkG(V)$ such that the Lie algebras $\mfkG(U)$ assemble into a pre-cosheaf $\mfkG$ of Lie algebras on $M$. This is a \textit{coflasque} pre-cosheaf, i.e. all its structure maps $\iota_{VU}$ are injective.\footnote{The terminology comes from sheaf theory, where a pre-sheaf $\CF:Open(M)^{opp}\ra \cC$ is called flasque if for any $U\subseteq V$ the restriction morphism $\CF(V)\ra\CF(U)$ is a surjection.}

Continuing with the example, for any two opens $U,V$ the following sequence is exact:
\begin{align}\label{eq:UVcosheaf}
    \mfkG(U\cap V)\ra \mfkG(U)\oplus\mfkG(V)\ra \mfkG(U\cup V)\ra 0.
\end{align}
Here the first arrow is $\iota_{U,U\cap V}\oplus (-\iota_{V,U\cap V})$ and the second arrow is $\iota_{U\cup V,U}\oplus\iota_{U\cup V,V}$.
Exactness follows from the existence of a partition of unity for the cover $\fU=\{U,V\}$ of $U\cup V$. In words, the exactness of the sequence (\ref{eq:UVcosheaf}) means that any element of $\mfkG(U\cup V)$ can be decomposed as a sum of elements of subalgebras attached to $U$ and $V$ modulo ambiguities which take values in the subalgebra attached to $U\cap V$.

More generally, the existence of a partition of unity implies that for any finite collection of opens $U_i$, $i\in I$ the following sequence is exact \cite{Bredon1968}:
\begin{align}\label{eq:cosheaf}
   \oplus_{i<j} \mfkG(U_i\cap U_j)\ra \oplus_i \mfkG(U_i)\ra \mfkG(\cup_i U_i)\ra 0.
\end{align}
For a compact $M$, this implies that the pre-cosheaf of vector spaces $\mfkG$ is a cosheaf of vector spaces. The cosheaf property is a compatibility of the pre-cosheaf with the notions of intersection and union of opens and expresses a stronger form of locality.

Note that the maps in the above exact sequences are not Lie algebra homomorphisms. Hence $\mfkG$ is not a cosheaf of Lie algebras. Nevertheless, the following additional property of $\mfkG$ can be regarded as a form of locality of the Lie bracket.
\begin{definition}
    A coflasque pre-cosheaf of Lie algebras $\mfkF$ is said to satisfy Property I if for any two opens $U,V$ one has $[\mfkF(U),\mfkF(V)]\subseteq \mfkF(U\cap V)$.
\end{definition}
Clearly, $\mfkG$ satisfies Property I. In particular, elements of $\mfkG(M)$ that are supported on non-intersecting opens commute.

In the previous section, for every quantum lattice system with a gapped state $\psi$ we defined Lie algebras $\mfkDal(U)$ and $\mfkDpal(U)$ attached to subsets in $\RR^n$ in such a way that $\mfkDal(U)=\mfkDal(U^r)$ and $\mfkDpal(U)=\mfkDpal(U^r)$ for any $U$ and any $r\geq 0$. Intuitively, this means that localization on the lattice is only approximate. To phrase locality of symmetries in lattice systems, one needs to replace the set $Open(M)$ of open subsets of $M$ with a more general structure which admits the notions of intersection, union, and cover.

The first thing to note is that the category $Open(M)$ is rather  special: its objects form a pre-ordered set (i.e. the set of objects carries a relation $\subseteq $ which is reflexive and transitive), and the category structure is determined by the pre-order. The relation $\subseteq$ is also anti-symmetric: $U\subseteq V$ and $V\subseteq U$ implies $U=V$. In other words, $Open(M)$ is a poset. In general, we will not require the pre-order to be anti-symmetric. From the categorical viewpoint, $U\subseteq V$ and $V\subseteq U$ means that $U$ and $V$ are isomorphic objects of the category $Open(M)$, and as a general rule, it is not advisable to identify isomorphic objects.

For any pre-ordered set $(X,\leq)$ there is a natural notion of intersection and union. The intersection of $U,V\in X$ can be defined as the greatest lower bound (or meet) of both $U$ and $V$, i.e. a $W\in X$ such that $W\leq U$, $W\leq V$, and for any $W'\leq U,V$ we have $W'\leq W$. The meet of $U$ and $V$ is denoted $U\wedge V$. Similarly, the union of $U$ and $V$ can be defined as the smallest upper bound (or join) of both $U$ and $V$. It is denoted $U\vee V$. For a general pre-ordered set, the meet and join may not exist for all pairs of objects.  If they exist, they are unique up to isomorphism. We will assume that $(X,\leq)$ is such that $U\wedge V$ and $U\vee V$ exist for all $U,V\in X$.\footnote{If we turn $X$ into a poset by identifying  isomorphic objects, then this means that the poset is a lattice in the sense of order theory.} The existence of all pairwise meets and joins  implies the existence of all finite meets and joins. In the case of the pre-ordered set $Open(M)$ arbitrary (i.e. not necessarily finite) joins make sense. 

Finally, to define covers of elements of $X$, let us assume that 
$U\wedge (V\vee W)\leq (U\wedge V)\vee (U\wedge W)$ for all $U,V,W\in X$.\footnote{The opposite relation is automatic, so this condition ensures that $U\wedge (V\vee W)\simeq (U\wedge V)\vee (U\wedge W)$ for all $U,V,W\in X$. This is equivalent to saying that the poset corresponding to $X$ is a distributive lattice.} We will say that $X$ is a distributive pre-ordered set. This condition is certainly satisfied for $Open(M)$. We say that a collection $\fU=\{U_i\}_{i\in I}$ of elements of $X$ covers $A\in X$ iff $U_i\leq A$ for all $i\in I$ and $A\leq \bigvee_{i\in I} U_i$. This definition ensures that if $\fU$ covers $A$, then for any $B\leq A$ the collection $\fU\wedge B=\{U_i\wedge B\}_{i\in I}$ covers $B$.

In the case of the pre-ordered set $Open(M)$, the standard topological definition of a cover allows  $I$ to be infinite. In general, if $X$ admits only finite joins, $I$ needs to be finite. Also, we may or may not allow $I$ to be empty. This possibility only arises when $A$ is the smallest element of $X$, i.e. $A\leq U$ for any $U\in X$. In the case of the pre-ordered set $Open(M)$, the smallest element is the empty set $\varnothing$, and the standard choice is to allow the empty cover of the empty set. For a general pre-ordered set, it is up to us whether to allow $I$ to be empty.

From a categorical perspective, this notion of a cover equips any distributive pre-ordered set $(X,\leq)$ admitting pairwise meets and joins with a Grothendieck topology, thus making it into a {\it site} \cite{MacLaneSheaves}. Apart from the option of allowing the labeling set $I$ to be empty, this Grothendieck topology is canonical. 

For any $W\in X$ we may consider the subset $X^W=\left\{U\in X\mid U\leq W\right\}$ with the pre-order inherited from $(X,\leq)$. It is a distributive pre-ordered set in its own right. When equipped with its canonical Grothendieck topology, it can be regarded as a sub-site of the site associated to $(X,\leq)$.

Given any distributive pre-ordered set $(X,\leq)$, we can define the notion of a pre-cosheaf of vector spaces, a pre-cosheaf of Lie algebras, a cosheaf of vector spaces, and a coflasque pre-cosheaf of Lie algebras with Property I exactly as before, i.e. by mechanically replacing $\cup$ with $\vee$ and $\cap$ with $\wedge$. Motivated by the above example, we introduce the following definition.
\begin{definition}
    A local Lie system over $(X,\leq)$ is a coflasque pre-cosheaf of Lie algebras with Property I which is also a cosheaf of vector spaces over the corresponding site. A morphism of local Lie systems is a morphism of the underlying pre-cosheaves of Lie algebras.
\end{definition}
\begin{remark}
    Let $\mfkF$ be a local Lie system. Then for any $U\leq V$ the Lie algebra $\mfkF(U)$ is an ideal in $\mfkF(V)$. One can equivalently define a local Lie system as a pre-cosheaf of Lie algebras over $(X,\leq)$ which is a cosheaf of vector spaces and such that all co-restriction maps are inclusions of Lie ideals.
\end{remark}
\begin{remark}
     In this paper all Lie algebras will be \Frechet-Lie algebras and all morphisms will be continuous. We define a local \Frechet-Lie system  over $(X,\leq)$ to be a coflasque pre-cosheaf of \Frechet-Lie algebras with Property I which is also a cosheaf of vector spaces.
\end{remark}
The following Lemma will be useful in future sections to check the cosheaf property:
\begin{lemma}\label{lma:2-is-sufficient}
Let $X$ be a distributive pre-ordered set with finite meets $\wedge$ and joins $\vee$.
Write $A\cong B$ to mean $A\le B$ and $B\le A$. Equip $X$ with the Grothendieck topology of finite covers: a finite family
$\{U_i\}_{i\in I}$ is a cover of $U$ if
$U_i\leq U$ for all $i\in I$ and $U\leq \bigvee_{i\in I} U_i$.
Let $\mfkF$ be a pre-cosheaf of vector spaces on $X$, with structure maps $\iota_{VU}:\mfkF(U)\to\mfkF(V)$ for $U\le V$.

Then $\mfkF$ is a cosheaf for the above topology if and only if for every $U,V\in X$
the sequence
\begin{equation}\label{eqn:n=2}
\mfkF(U\wedge V)\xrightarrow{\ \alpha\ }\mfkF(U)\oplus \mfkF(V)
\xrightarrow{\ \beta\ }\mfkF(U\vee V)\to 0
\end{equation}
is exact, where
\begin{equation}
\alpha(t)=\bigl(\iota_{U,\,U\wedge V}(t),\,-\iota_{V,\,U\wedge V}(t)\bigr),
\qquad
\beta(a,b)=\iota_{U\vee V,\,U}(a)+\iota_{U\vee V,\,V}(b).
\end{equation}
\end{lemma}
The following proof is essentially a restatement of the proof of Proposition 1.3 in \cite{Bredon1968}, adapted to the site $X$:
\begin{proof}
The ``only if'' direction is the cosheaf axiom for the $2$-cover $\{U,V\}$ of $U\vee V$.

Conversely, assume exactness for all pairs. Since $\mfkF$ sends $\cong$ to isomorphisms, for a cover
$\{U_i\le W\}$ one may replace $W$ by $\bigvee_i U_i$. Thus it suffices to treat a finite family
$U_1,\dots,U_n$ with $U:=\bigvee_{i=1}^n U_i$.

Consider the standard sequence
\begin{equation}
    \bigoplus_{i<j}\mfkF(U_i\wedge U_j)\xrightarrow{\ \alpha\ }
\bigoplus_{i=1}^n \mfkF(U_i)\xrightarrow{\ \beta\ }\mfkF(U)\to 0,
\end{equation}
with $\beta((s_i))=\sum_i \iota_{U,U_i}(s_i)$ and $\alpha$ the alternating sum on pairwise overlaps.
Surjectivity of $\beta$ follows by induction on $n$ using the exactness for the pair
$\bigl(\bigvee_{i=1}^{n-1}U_i,\;U_n\bigr)$ and the inductive hypothesis for $\bigvee_{i=1}^{n-1}U_i$.

For exactness in the middle, let $(s_i)\in\ker\beta$ and set $V:=\bigvee_{i=1}^{n-1}U_i$.
Exactness for $(V,U_n)$ gives $t\in\mfkF(V\wedge U_n)$ with
$\sum_{i=1}^{n-1}\iota_{V,U_i}(s_i)=\iota_{V,V\wedge U_n}(t)$ and $s_n=-\iota_{U_n,V\wedge U_n}(t)$.
By distributivity, $V\wedge U_n\cong\bigvee_{i=1}^{n-1}(U_i\wedge U_n)$, hence (by the already proven
surjectivity) $t$ lifts to $(t_i)\in\bigoplus_{i=1}^{n-1}\mfkF(U_i\wedge U_n)$.
Subtracting $\alpha((t_i))$ reduces to an element in $\ker\beta$ for the $(n-1)$-cover of $V$,
which lies in $\operatorname{im}\alpha$ by the inductive hypothesis. This proves
$\ker\beta=\operatorname{im}\alpha$ and hence the cosheaf axiom for all finite covers.
\end{proof}

\subsection{DGLA attached to a cover}\label{sec:DGLA-attached-to-lls}

Let $X$ be a distributive pre-ordered set. Let $W\in X$ and $\fU=\{U_i\}_{i\in I}$ be a cover of $W$. Let $\mfkF$ be a pre-cosheaf of vector spaces over $X$. 
\begin{definition}
    The \Cech\ chain complex $C_\bullet(\fU,W;\mfkF)$ is defined as 
\begin{align}
    C_n(\fU,W;\mfkF)={\rm Alt}\left[\oplus_{i_0,\ldots, i_n} \mfkF(U_{i_0}\wedge \ldots \wedge U_{i_n})\right],\quad n\geq 0,
\end{align}
where ${\rm Alt}$ means that we only keep elements that are skew-symmetric under permutations of the indices. The differential $\partial:C_n\ra C_{n-1}$ is given by $\partial = \sum_{j=0}^n(-1)^j\lambda_j$, where $\lambda_j$ is the canonical map
\begin{align}
    \mfkF(U_{i_0}\wedge \hdots \wedge U_{i_n}) \to \mfkF(U_{i_0}\wedge \hdots  \widehat{U_{i_j}}  \hdots \wedge U_{i_n}).
\end{align}
\end{definition}
\begin{lemma}\label{lma:coflasqueacyclic}
    Let $\mfkF$ be a coflasque cosheaf of vector spaces over $X$ and $\fU$ be a covering of $W\in X$. The homology of the \Cech\ complex with respect to $\fU$ is $0$ for $n>0$ and $\mfkF(W)$ for $n=0$. 
\end{lemma}
\begin{proof}
    See \cite{Bredon1968}, Corollary 4.3. Note that \cite{Bredon1968} uses the term ``flabby'' instead of ``coflasque''. 
\end{proof}
Let $C^{aug}(\fU,W;\mfkF)=\{C_\bullet(\fU,W;\mfkF)\ra \mfkF(W)\}$ be the augmented \Cech\ complex, where the augmentation $\epsilon:C_0(\fU,W;\mfkF)=\bigoplus_i\mfkF(U_i)\to \mfkF(W)$ is
$\epsilon((\derF_i))=\sum_i \iota_{W,U_i}(\derF_i)$.

\begin{prop}
    Let $\mfkF$ be a local Lie system over $X$. Then for any $W\in X$ and any cover $\fU$ of $W$ the 1-shifted augmented \Cech\ complex $C^{aug}_{\bullet+1}(\fU,W;\mfkF)$ has a natural structure of a non-negatively graded acyclic DGLA. 
\end{prop}
\begin{proof}
    Let $\fU$ be a cover of $W$ indexed by $I$. Let $\hI$ be a vector space with a basis $e_i$, $i\in I$ and let $f^i,i\in I$, be the dual basis of $\hI^*$. The 1-shifted augmented \Cech\ complex of $\mfkF$ with respect to $\fU$ is naturally identified with a sub-complex of the DGLA $(\mfkF(W)\otimes \Lambda^\bullet \hI, \partial)$, where $\partial$ is contraction with $\sum_i f^i$. It is easy to check that this sub-complex is closed with respect to the graded Lie bracket thanks to Property I. By Lemma \ref{lma:coflasqueacyclic}, the resulting DGLA is acyclic. 
\end{proof}

Covers of $W\in X$ form a category whose morphisms are refinements. A refinement of a cover $\fU=\left\{ U_i\right\}_{i\in I}$ to a cover $\fV=\left\{V_j\right\}_{j\in J}$ is a map $\phi:J\ra I$ such that $V_j\leq U_{\phi(j)}$. Refinements are composed in an obvious way. 
\begin{prop} \label{prop:DGLA_functoriality}
    For a fixed $W\in X$, the map which sends a local Lie system $\mfkF$ and a cover $\fU$ of $W$ to the DGLA $C^{aug}_{\bullet+1}(\fU,W;\mfkF)$ is functorial in both $\mfkF$ and $\fU$. 
\end{prop}
\begin{proof}
    Functoriality in $\mfkF$ is clear. To prove functoriality in $\fU$, note that the refinement map $\phi$ induces a linear map $\hat\phi:\hJ\ra\hI$, where $\hJ$ is the vector space with basis $\{e_j\}_{j\in J}$, hence a linear map $\Lambda\hat\phi:\Lambda^\bullet\hJ\ra\Lambda^\bullet\hI$, hence a linear map $\phi_*:\mfkF(W)\otimes \Lambda^\bullet\hJ\ra\mfkF(W)\otimes\Lambda^\bullet\hI$. One can check that $\phi_*$ commutes with the differential and maps the subcomplex $C^{aug}_{\bullet+1}(\fV,W;\mfkF)$  to $C^{aug}_{\bullet+1}(\fU,W;\mfkF)$. 
\end{proof}
We will call $C^{aug}_{\bullet+1}(-,W;-)$ the \Cech\ functor. We will need some variants of the \Cech\ functor. First, we can define a \textit{graded} local Lie system over a distributive pre-ordered set $(X,\leq)$ in an obvious manner. The construction of the acyclic DGLA $C^{aug}_{\bullet+1}(\fU,W;\mfkF)$ works in this case as well, except that it may have components in negative degrees. 

Second, define a {\it pointed DGLA} as a DGLA equipped with a distinguished central cycle of degree $-2$ (which we call the curvature). A morphism of pointed DGLAs is a DGLA morphism which preserves the distinguished central cycle.\footnote{The category of pointed DGLAs as defined here is a full subcategory of the category of curved DGLAs as defined in \cite{ChuangLazarevMannan}. There, for a curved DGLA with curvature $\derB$, $\derB$ is not required to be central and the derivation $\partial$ satisfies $\partial^2=\ad_\derB$.}  
\begin{definition}
    A graded local Lie system $\mfkF$ over $(X,\leq)$ with a terminal object $T\in X$ is called pointed if it is equipped with a distinguished central element $\derB\in \mfkF(T)$ of degree $-2$. Morphisms in the category of pointed graded local Lie systems are required to preserve the distinguished element.
\end{definition}
 If $\mfkF$ is a pointed graded local Lie system, then the DGLA $C^{aug}_{\bullet+1}(\fU,T;\mfkF)$ is a pointed DGLA. The distinguished central cycle is  $\derB$ which lies in the augmentation (has \Cech\ degree $-1$). Of course, since the DGLA is acyclic, this central cycle is exact. The construction of a pointed DGLA from a pointed graded local Lie system and a cover of $T$ is functorial in both arguments.  

\section{The site of fuzzy semilinear sets}\label{sec:semilinearsets}

In this section we define and study the site of fuzzy semilinear sets in $\RR^n$. We show that it is equivalent as a site to the pre-ordered set of closed  polyhedral subsets of $S^{n-1}$ equipped with a slightly unusual topology. In the following section, we will show that quantum lattice systems give rise to local Lie systems over this site.

\subsection{Semilinear sets and their thickenings}

The results of Sections \ref{sec:symmetries} and \ref{sec:locality} suggest that we should look for a pre-ordered set of subsets of $\RR^n$ with the following two properties. First, $U\leq V$ whenever there exists $r\geq 0$ such that $U\subseteq V^r$. Second, for every two elements $U,V$ there should exist $r\geq 0$ such that $U^r$ and $V^r$ are transverse. The simplest choice is the set whose elements are finite unions of polyhedra, or equivalently, closed semilinear sets in $\RR^n$.

\begin{definition}
    A semilinear set in $\RR^n$ is a subset of $\RR^n$ which can be defined by means of a finite number of linear equalities and strict linear inequalities. More precisely, a basic semilinear set in $\RR^n$ is an intersection of a finite number of hyperplanes and open half-spaces, and a semilinear set is a finite union of basic semilinear sets \cite{tametopology}.
\end{definition}
The set of semilinear sets of $\RR^n$ will be denoted $\cS_n$. Projections $\RR^m\times\RR^n\ra\RR^m$ map $\cS_{m+n}$ to $\cS_m$ \cite{tametopology}. A map $\RR^n\ra\RR^m$ is called semilinear iff its graph is a semilinear set of $\RR^{m+n}$. The  composition of two semilinear maps is a semilinear map \cite{tametopology}.  

Recall that we use the $\ell^\infty$ metric on $\RR^n$.
\begin{lemma}
    The distance function $d: \RR^n\times\RR^n\ra\RR$ is semilinear.
\end{lemma}
\begin{proof}
The function $( x , y )\mapsto x_i-y_i$ is semilinear for any $i$. The function $|\cdot|:\RR\ra\RR$ is semilinear. If $f,g:\RR^n\ra\RR$ are semilinear, then $h=\max(f,g):\RR^n\ra\RR$ is semilinear. Since the set of semilinear functions is closed under composition, this proves the lemma.
\end{proof}
Recall that for any set $U\subseteq \RR^n$, we write $U^r:=\{x\in\RR^n: \exists y\in U \text{ s.t } d(x,y)\le r\}$ and call this the $r$-\emph{thickening} of $U$. It is easy to see that if $U$ is closed, then $U^r$ is also closed for any $r\geq 0$.
\begin{lemma}
    If $U$ is semilinear, $U^r$ is semilinear for any $r$.
\end{lemma}
\begin{proof}
Consider the set 
\begin{equation}
    \Delta_r=\{( x , y )\in\RR^{2n}| d( x , y )\leq r\}.
\end{equation}
By the previous lemma, $\Delta_r$ is semilinear. On the other hand, $U^r$ is the projection to the first $\RR^n$ of $\Delta_r\cap (\RR^n\times U)\subseteq\RR^n\times\RR^n$. Since intersection and projection preserve the set of semilinear sets, the lemma is proved.
\end{proof}
\begin{lemma}
    If $U$ is convex, then $U^r$ is convex, for any $r\geq 0$.
\end{lemma}
\begin{proof}
    Suppose $ x , y \in U^r$, and suppose $ x ', y '\in U$ satisfy $d( x , x ')\leq r$ and $d( y , y ')\leq r$. Then for any $t\in [0,1]$ we have $d(t x +(1-t) y ,t x '+(1-t) y ')\leq t d( x , x ')+(1-t) d( y , y ')\leq r.$ Since $U$ is convex, $t x '+(1-t) y '\in U$, hence $tx+(1-t)y\in U^r$. 
\end{proof}

A polyhedron in $\RR^n$ is an intersection of a finite number of closed half-spaces. A polyhedron is closed, but not necessarily compact. 
A closed semilinear set is the same as a finite union of polyhedra. Conversely, according to Theorem 19.6 of  \cite{convex}, a polyhedron can be described as a closed convex semilinear set. Combining this with the above lemmas, we get
\begin{corollary}\label{cor:poly-thickening}
    If $U$ is a polyhedron, then $U^r$ is a polyhedron, for any $r\geq 0$.
\end{corollary}

\subsection{A category of fuzzy semilinear sets}

Clearly, if for $X,Y\in\cS_n$ we have $X\subseteq Y$, then for any $r\geq 0$ we have $X^r\subseteq Y^r$. Also, for any $r,s\geq 0$ and any $U\in\cS_n$ we have $(U^r)^s\subseteq U^{r+s}$. Thus we can define a pre-order $\leq$  on $\cS_n$ by saying that $U\leq V$ iff there exists $r\geq 0$ such that $U\subseteq V^r$. We will call this pre-order relation {\it fuzzy inclusion}. Equivalently, $\cS_n$ can be made into a category, with a single morphism from $U$ to $V$ iff $U\leq V$.  One can turn the pre-ordered set $(\cS_n,\leq)$ into a poset by identifying isomorphic objects of the corresponding category, but for our purposes it is more convenient not to do so. On the other hand, every semilinear set is isomorphic to its closure, and we find it convenient to work with an equivalent category (or pre-ordered set) which contains only closed semilinear subsets. We will denote it  $\CS_n$ and call it the category (or pre-ordered set) of fuzzy semilinear sets.

\begin{prop}
    $\CS_n$ has all pairwise joins: for any $U,V\in \CS_n$ the join is given by $U\cup V$.
\end{prop}
\begin{proof}
    If $U\subseteq W^r$ and $V\subseteq W^s$ for some $r,s\geq 0$, then $U\subseteq W^{\max(r,s)}$ and $V\subseteq W^{\max(r,s)}$, and thus $U\cup V\subseteq W^{\max(r,s)}$. 
\end{proof}
Let us show that $\CS_n$ has all pairwise meets, using the notion of transverse intersection from the previous section.
Recall (Definition \ref{defn:transverse}) that we say two sets $U,V\subseteq \RR^n$ are transverse if for some $C>0$ we have $d(x,U\cap V) \le C\max(d(x,U),d(x,V))$ for all $x\in \RR^n$. 

We need the following geometric result \cite{Petrov}\footnote{The proof in \cite{Petrov} is for the Euclidean distance, but since the Euclidean distance function and $d( x , y )=\| x - y \|_\infty$ are equivalent (each one is upper-bounded by a multiple of the other), the result applies to $d( x , y )$ as well. }:
\begin{lemma}\label{lma:poly-distance-bound}
    Let $P$ and $Q$ be polyhedra in $\RR^n$. If $P\cap Q\neq \varnothing$ then $P$ and $Q$ are transverse.
\end{lemma}
\begin{prop}\label{prop:semilinear-distance-bound}
    For every $U,V\in \CS_n$ there is an $r>0$ such that $U^r$ and $V^r$ are transverse.
\end{prop}
\begin{proof}
    Let $U = \cup_iP_i$ and $V = \cup_jQ_j$ be a decomposition of $U$ and $V$ into finite unions of polyhedra and choose $r>0$ so that $P_i^r\cap Q_j^r$ is nonempty for each pair $i,j$. By Corollary \ref{cor:poly-thickening} and Lemma \ref{lma:poly-distance-bound} there are constants $C_{P_i Q_j}>0$ such that 
    \begin{align}
        d( x , P_i^r\cap Q_j^r) \le C_{P_iQ_j} \max (d( x ,P_i^r),d(x ,Q_j^r))
    \end{align} for each pair $i,j$. Since $U^r = \cup_iP_i^r$ and $V^r = \cup_j Q_j^r$, for any $ x \in \RR^n$ there are indices $i^*$ and $j^*$ such that $d( x , U^r) = d( x , P_{i^*}^r)$ and $d( x , V^r) = d( x , Q_{j^*}^r)$. Then we have
    \begin{align}
        d( x , U^r\cap V^r) &= d( x , \cup_{i,j}(P_i^r\cap Q_j^r))\notag\\
        &\le d( x , P_{i^*}^r\cap Q_{j^*}^r)\notag\\
        &\le C_{P_{i^*}Q_{j^*}}\max( d( x , P_{i^*}^r), d( x , Q_{j^*}^r))\notag\\
        &= C_{P_{i^*}Q_{j^*}}\max( d( x , U^r), d( x , V^r)),
    \end{align}
    and so $U^r,V^r$ are $C$-transverse for $C := \max_{i,j}C_{P_i Q_j}$.
\end{proof}
\begin{corollary}\label{cor:semilinear-distance-bound}
    With $U,V,r$ as above, $U^r\cap V^r$ is a meet of $U$ and $V$. In particular, $\CS_n$ has all pairwise meets.
\end{corollary}
\begin{proof}
    Since $U$ and $U^r$ (resp. $V$ and $V^r$) are isomorphic in $\CS_n$, it suffices to show that $U^r\cap V^r$ is a meet of $U^r$ and $V^r$. It's clear that $U^r\cap V^r \le U^r$ and $U^r\cap V^r \le V^r$. Now suppose $W\in \CS_n$ satisfies $W\le U^r$ and $W \le V^r$. Then there is an $s>0$ such that every $ x \in W$ satisfies $\max(d( x , U^r), d( x , V^r))\le s$. Since $d( x ,U^r\cap V^r) \le C\max(d( x , U^r),d( x , V^r))$ for some $C>0$ we have $W \subseteq (U^r\cap V^r)^{Cs}$.
\end{proof}

\begin{prop}
    The pre-ordered set $\CS_n$ is distributive. 
\end{prop}
\begin{proof}
    We need to show that for any $U,V,W\in \CS_n$ we have $U\wedge (V\vee W)\leq (U\wedge V)\vee (U\wedge W)$. According to Corollary \ref{cor:semilinear-distance-bound}, there exists $r\geq 0$ such that $U\wedge (V\vee W)\simeq U^r\cap (V\cup W)^r$. Since $(V\cup W)^r=V^r\cup W^r,$ we also have $U\wedge (V\vee W)\simeq (U^r\cap V^r)\cup (U^r\cap W^r)$.  On the other hand, $U\wedge V\simeq U^r\wedge V^r\simeq (U^r)^s\cap (V^r)^s$ for some $s\geq 0$, and $U\wedge W\simeq U^r\wedge W^r\simeq (U^r)^t\cap (W^r)^t$ for some $t\geq 0$. Thus, $(U\wedge V)\vee (U\wedge W)\simeq \left((U^r)^s\cap (V^r)^s\right)\cup \left((U^r)^t\cap (W^r)^t\right)$ for some $s,t\geq 0$. Since we have inclusions $U^r\cap V^r\subseteq (U^r)^s\cap (V^r)^s$ and $U^r\cap W^r\subseteq (U^r)^t\cap (W^r)^t$, the proposition is proved.
\end{proof}

We can now equip $\CS_n$ with a Grothendieck topology of Section \ref{sec:locality}. There are two versions of it which differ in whether we allow empty covers of an initial object or not. In the case of the pre-ordered set $\CS_n$, every bounded closed semilinear set is an initial object (they are all isomorphic objects of the category $\CS_n$). Since such sets are not empty, we will disallow empty covers. This choice is also forced on us if we want certain pre-cosheaves to be cosheaves (see below). Note that we only consider non-empty closed semilinear sets.

\begin{definition}\label{def:coherent}
    The site $\CS_n$ has as its objects nonempty closed semilinear sets of $\RR^n$. For any two objects $U,V$, $Mor(U,V)$  contains a single element if $U\subseteq V^r$ for some $r\geq 0$ and is empty otherwise. In the former case we write $U\leq V.$ A cover of an object $V$ is a finite collection of objects $U_i$, $i\in I$, with $I$ nonempty, such that $U_i\leq V$ for all $i\in I$ and $V\leq \bigvee_{i\in I} U_i$.
\end{definition}

More generally, for any $W\in\CS_n$ we may consider a full sub-category $\CSnW$ whose objects are $U\in\CS_n$ such that $U\leq W$. This is a distributive pre-ordered set, and we will also have occasion to consider local Lie systems on the associated site.

\subsection{Spherical CS sets}

Every two bounded elements of $\CS_n$ are isomorphic objects of the corresponding category. More generally, any two elements of $\CS_n$ which coincide outside some ball in $\RR^n$ are isomorphic objects. Thus $\CS_n$ encodes the large-scale structure of $\RR^n$. To make this explicit, we will show that the pre-ordered set $\CS_n$ is equivalent as a category to a certain poset of subsets of the ``sphere at infinity'' $S^{n-1}$. 

A cone in $\RR^n$ is a non-empty subset of $\RR^n$ which is invariant under $ x \mapsto \lambda x $, where $\lambda\geq 0$. Every cone contains the origin $0$. Cones in $\RR^n$ are in bijection with subsets of $S^{n-1}=(\RR^n\backslash \{0\})/\RR^*_+$ where $\RR^*_+$ is the group of positive real numbers under multiplication. If $A\subseteq S^{n-1}$, we denote the corresponding cone $c(A)$ over $A$. In particular, $c(\varnothing)=\{0\}\in \RR^n$. Note also that $c(A\cup B)=c(A)\cup c(B)$ for any  $A,B\subseteq S^{n-1}$. If $K$ is a cone in $\RR^n$, we will denote the corresponding subset of $S^{n-1}$ by $\widehat K$. For any cone $K\subseteq \RR^n$, we have $c\left(\widehat K\right)=K$, for any $A\subseteq S^{n-1}$ we have $\widehat {c(A)}=A$.

\begin{definition}
    $A\subseteq S^{n-1}$ is a spherical polyhedron iff $c(A)$ is a polyhedron and $A$ is contained in some open hemisphere of $S^{n-1}$. $A\subseteq S^{n-1}$ is a spherical CS set iff it is a union of a finite number of spherical polyhedra. The set of spherical CS sets in $S^{n-1}$ is denoted $\SCS_n$.
\end{definition}
\begin{remark}\label{rem:cone-is-SCS}
    For $K$ a polyhedral cone, $\widehat K$ need not be contained in an open hemisphere and thus may not be a spherical polyhedron. However, by intersecting $K$ with orthants, one can present every such $K$ as a finite  union of polyhedral cones $K_i$ such that $\widehat K_i$ is contained in an open hemisphere. Therefore $\widehat K\in \SCS_n$ for every polyhedral cone $K$.
\end{remark} 
Every nonempty polyhedron is convex and thus contractible. This implies:
\begin{prop}
    Any nonempty spherical polyhedron is contractible.
\end{prop}
\begin{proof}
    Without loss of generality, we may assume that the spherical polyhedron is contained in the hemisphere $S^{n-1}_+=S^{n-1}\cap \{x_n>0\}$. The map $\RR^{n-1}\ra S^{n-1}_+$ which sends $(x_1,\ldots,x_{n-1})$ to the equivalence class of $(x_1,\ldots,x_{n-1},1)$ is a homeomorphism which establishes a bijection between nonempty bounded polyhedra in $\RR^{n-1}$ and nonempty spherical polyhedra contained in $S^{n-1}_+$.
\end{proof}
\begin{prop}
    The intersection of two spherical polyhedra is a spherical polyhedron. The union and intersection of two spherical CS sets is a spherical CS set. 
\end{prop}
\begin{proof}
    Clear from the definitions.  
\end{proof}
Spherical CS sets form a poset $\SCS_n$ under inclusion. This poset has pairwise joins and meets given by unions and intersections, respectively.

\begin{prop}\label{prop:CS-is-SCS}
    The category $\SCS_n$ is equivalent to the category $\CS_n$. 
\end{prop}
To prove this proposition, we need two definitions and three lemmas.
\begin{definition}
    Let $P$ be a nonempty polyhedron in $\RR^n$. Its recession cone is the set
    \begin{align}
    \rec(P):=\{v\in \RR^n : P+t v \subseteq P \textit{ for all }t\ge 0\}.
    \end{align}
\end{definition}
\begin{lemma}\label{lem:rec}
    If $P=\bigcap_{i=1}^k\{ n_i \cdot  x  \leq b_i\},$ then
    $\rec(P)=\bigcap_{i=1}^k\{ n_i \cdot  x  \leq 0\}.$
    In particular, for any non-empty polyhedron $P$, $\rec(P)$ is a polyhedral cone, and $\widehat{\rec(P)}\in\SCS_n$.
\end{lemma}
\begin{proof}
    Straightforward check.
\end{proof}
If $K$ is a polyhedral cone, then clearly $\rec(K)=K$. More generally, we have
\begin{lemma}\label{lem:MW}
    For any non-empty polyhedron $P$, $P\simeq \rec(P)$ as objects in $\CS_n$. 
\end{lemma}
\begin{proof}
    By the Minkowski-Weyl theorem (see e.g. \cite{convex}), $P=Q+\rec(P)$ where $Q$ is a bounded polyhedron. Clearly, $P\subseteq \rec(P)^r$ where $r=\sup_{q\in Q}\|q\|_\infty$. On the other hand, by the definition of $\rec(P)$ we have $x_0+\rec(P)\subseteq P$ for any $x_0\in P$, hence  $\rec(P)\subseteq P^s$ where $s=\|x_0\|_\infty$.
\end{proof}
\begin{definition}
    For $U\in\CS_n$, let $\widehat U\subseteq S^{n-1}$ be the set of all $u\in S^{n-1}$ such that $c(\{u\})\leq U$.
\end{definition}
If $U$ is a polyhedral cone $K$, this definition agrees with the previous definition of $\widehat K$. This will be shown in the course of the proof of the next lemma.
\begin{lemma}\label{lem:hat}
    For any $U\in\CS_n$, $\widehat U\in\SCS_n$. More precisely, if $U=\cup_{i=1}^m P_i$ where each $P_i$ is a polyhedron, then 
    \begin{align}\label{eq:hat}
    \widehat U=\bigcup_{i=1}^m \widehat{\rec(P_i)}
    \end{align}
\end{lemma}
\begin{proof}
 First, let us prove (\ref{eq:hat}). If $u\in \widehat{\rec(P_i)}$ for some $i$, then $c(\{u\})\leq \rec(P_i)$. Since $\rec(P_i)\leq P_i$ by Lemma \ref{lem:MW}, this implies $c(\{u\})\leq P_i\leq U$, hence $u\in \widehat U$. 

In the opposite direction, suppose $u\in \widehat U$, i.e. $c(\{u\})\subseteq U^r$ for some $r\geq 0$. We claim there exists $i_*\in\{1,\ldots,m\}$ such that $c(\{u\})\leq P_{i_*}$. Indeed, let 
\begin{align}
    T_i=\left\{t\in [0,\infty)|d(tu,P_i)\leq r\right\}.
\end{align}
Here we identified $u\in S^{n-1}$ with the corresponding unit vector in $\RR^n$. We have $\cup_i T_i=[0,\infty)$, hence there exists 
$i_*\in\{1,\ldots,m\}$ such that $T_{i_*}$ is unbounded. Suppose $P_{i_*}=\cap_{\alpha=1}^M 
\{n_\alpha\cdot x\leq b_\alpha\}$ for some vectors $n_\alpha\in\RR^n$ and some $b_\alpha\in\RR$. Choose a sequence $\{t_k\}_{k\in\NN}\subseteq T_{i_*}$ going to infinity and for each $k\in\NN$ a point $x_k\in P_{i_*}$ such that $d(t_ku,x_k)\leq r$. 
Then for every $\alpha\in\{1,\ldots,M\}$
\begin{align}
t_k n_\alpha\cdot u\leq b_\alpha+n_\alpha\cdot (t_k u-x_k)\leq b_\alpha+r\|n_\alpha\|_1. 
\end{align}
Taking the limit $k\ra\infty$ we conclude that $n_\alpha\cdot u\leq 0$ for all $\alpha$, hence by Lemma \ref{lem:rec} $c(\{u\})\subseteq \rec(P_{i_*})$ and by Lemma \ref{lem:MW} $c(\{u\})\leq P_{i_*}$.

The same argument applied to the case when $U$ is a polyhedral cone $K$ shows that the new definition of $\widehat K$ agrees with the earlier one.

Finally, since each $\rec(P_i)$ is a polyhedral cone, Remark \ref{rem:cone-is-SCS} implies that $\widehat U\in\SCS_n$.
\end{proof}

\begin{proof}[Proof of Proposition \ref{prop:CS-is-SCS}]
    The map $U\mapsto \widehat U$ is order-preserving. Indeed, if $U\leq V$, then $\forall u\in \widehat U$ we have $c(\{u\})\leq U\leq V$, hence $u\in \widehat V$. Thus $U\mapsto \widehat U$ defines a functor from $\CS_n$ to $\SCS_n$ which we call $F$. The map $A\mapsto c(A)$ is obviously order-preserving and defines a functor from $\SCS_n$ to $\CS_n$ which we call $G$. 

    We claim that these two functors establish an equivalence between $\CS_n$ and $\SCS_n$. Indeed, $F\circ G={\rm id}_{\SCS_n}$ since $\widehat {c(A)}=A$ for any $A\in\SCS_n$. In the opposite direction, let $U=\cup_{i=1}^m P_i$ where all $P_i$ are polyhedra. Then by Lemma \ref{lem:hat},
    \begin{align}
        G\circ F(U)=c\left(\cup_{i=1}^m \widehat{\rec(P_i)}\right)=
        \cup_{i=1}^m c\left(\widehat{\rec(P_i)}\right)=\cup_{i=1}^m \rec(P_i).
    \end{align}
Since by Lemma \ref{lem:MW} $\rec(P_i)\simeq P_i$, this implies $G\circ F(U)\simeq U$.
\end{proof}
Note that under this equivalence all bounded elements of $\CS_n$ correspond to $\varnothing\in\SCS_n$. The canonical Grothendieck topology on $\CS_n$ corresponds to a slightly unusual Grothendieck topology on the poset of spherical CS subsets of $S^{n-1}$: the one where empty covers of $\varnothing$ are not allowed. Consequently, a cosheaf of vector spaces on $\SCS_n$ equipped with this topology need not map $\varnothing$ to the zero vector space.

\subsection{Spherical CS cohomology}

Let $\fU$ be a spherical CS cover of a spherical CS set $A$.  Spherical CS cohomology $\check{H}^\bullet_{CS}(\fU,A;\RR)$ is defined to be the simplicial cohomology of the \v{C}ech nerve $N(\fU)$.
\begin{definition}
    Let $A$ be a spherical CS set. The spherical CS cohomology $\check{H}^\bullet_{CS}(A,\RR)$ is defined as $\varinjlim \check{H}^\bullet_{CS}(\fU,A;\RR),$ where the colimit is taken over the directed set of all spherical CS covers.
\end{definition}

The following proposition connects spherical CS cohomology with singular cohomology using a functorial version of nerve theorems \cite{Borsuk1948OnTI, leray1945forme}.

\begin{prop} \label{cech_to_singular}
For any spherical CS set $A$ the graded vector space $\check{H}^\bullet_{CS}(A,\RR)$ is isomorphic to the singular cohomology $H^\bullet(A,\RR)$.
\end{prop}

\begin{proof}
    Every spherical CS cover can be refined to a cover by spherical polyhedra, so in the computation of $\varinjlim \check{H}^\bullet_{CS}(\fU,A;\RR)$ it suffices to take colimit over such covers. All intersections of elements of a spherical polyhedral cover $\fU$ are either empty or contractible. Also, since every polyhedron is a geometric realization of a simplicial complex, $\fU$ is a cover of a simplicial complex by subcomplexes.  By Theorem C of \cite{bauer2023unified} for such $\fU$ the direct system of groups $\fU\mapsto \check{H}^\bullet_{CS}(\fU,A;\RR)$ is constant and its limit is $H^\bullet(A,\RR)$.
\end{proof}

\section{Local Lie systems over fuzzy semilinear sets}\label{sec:LLA}

In this section, we define several examples of local Lie systems\footnote{In an earlier version, we used the term “local Lie algebras,” which conflicts with a different notion bearing the same name in~\cite{costellofactorization}.} over the site $\CS_n$. The motivating examples arise from the Lie algebras $\mfkDal(U)$ and $\mfkDpal(U)$ studied in Section \ref{sec:symmetries}. When $\psi$ is a gapped state invariant under a compact Lie group $G$, we define a pointed graded local Lie system which can be regarded as an equivariant version of $\mfkDpal$. Applying the \Cech\ functor gives a pointed DGLA which encodes the properties of the state. Obstructions for the existence of morphisms of pointed DGLAs give rise to topological invariants of $G$-invariant gapped states.

\subsection{Basic examples}
Let $\psi$ be a state of a quantum lattice system on $\RR^n$. By Lemma \ref{lma:seminorm-dominated} the maps sending $U \in \CS_n$ to $\mfkDal(U)$ and $\mfkDpal(U)$ are pre-cosheaves of \Frechet\ spaces on $\CS_n$. We denote these pre-cosheaves by $\mfkDal$ and $\mfkDpal$. Putting together Lemma \ref{lma:2-is-sufficient}, Proposition \ref{prop:SES}, Corollary \ref{cor:commutator}, and Proposition \ref{prop:SES-psi}, we have
\begin{theorem}\label{thm:D-is-LLA}
   $\mfkDal$ is a local Lie system over $\CS_n$. If $\psi$ is gapped, then $\mfkDpal$ is a local Lie system over $\CS_n$.
\end{theorem}

We may also consider quantum lattice systems on subsets of $\RR^n$.
\begin{definition}
    Let $W\in\CS_n$. A quantum lattice system on $W$ is a quantum lattice system on $\RR^n$ such that the lattice $\Lambda$ is contained in $W^r$ for some $r\geq 0$. 
\end{definition}
Obviously, for a quantum lattice system on $W$, $\mfkDal$ is a local Lie system over $\CSnW$, and for any gapped state $\psi$ of such a system, $\mfkDpal$ is a local Lie system on $\CSnW$.

A much more elementary example of a local Lie system arises from a finite-dimensional Lie algebra $\mfkg$ and any subset $\Lambda\subseteq\RR^n$. For any $U\in\CS_n$ let $\mfkgal(U)$ be the space of bounded functions $\Lambda\ra\mfkg$ which decay superpolynomially away from $U$. The subscript ``al'' stands for ``almost localized''. More precisely,  $\mfkgal(\RR^n)$ is the space of bounded functions $\Lambda\ra\mfkg$, while $\mfkgal(U)$ is defined as a subspace of $\mfkgal(\RR^n)$ consisting of functions $f:\Lambda\ra\mfkg$ such that the following seminorms are finite:
\begin{align}
p_{k,U}(f)=\sup_{x\in \Lambda} |f(x)|(1+d(U,x))^k,\quad k\in\NN.
\end{align}
\begin{prop}
    The assignment $U\mapsto \mfkgal(U)$ is a local Lie system.
\end{prop}
\begin{proof}
It is easy to check that the assignment $U\mapsto\mfkgal(U)$ is a coflasque pre-cosheaf of \Frechet-Lie algebras satisfying Property I. The only thing left to check is that it is a cosheaf of vector spaces.
By Lemma \ref{lma:2-is-sufficient}, it is sufficient to show that for any $U,V\in\CS_n$ the sequence
\begin{align}
    \mfkgal(U\wedge V)\ra \mfkgal(U)\oplus \mfkgal(V)\ra \mfkgal(U\vee V)\ra 0
\end{align}
is exact. To show exactness at $\mfkgal(U\vee V)=\mfkgal(U\cup V)$, we note that every $f\in \mfkgal(U\cup V)$ can be written as a sum $f_U+f_V$, where
\begin{align}
f_U( x )=\begin{cases} f( x ), & d( x ,U)< d( x ,V),\\ \frac{1}{2}f( x ), & d( x ,U)=d( x ,V),\\ 0, & d( x ,U)>d( x ,V), \end{cases}
\end{align}
and $f_V( x )$ is defined by a similar expression with $U$ and $V$ exchanged. Using $d( x ,U\cup V)=\min(d( x ,U),d( x ,V))$ it is easy to check that $p_{k,U}(f_U)\leq p_{k,U\cup V}(f)$ and $p_{k,V}(f_V)\leq p_{k,U\cup V}(f)$ for all $k$, and thus $f_U\in\mfkgal(U)$ and $f_V\in\mfkgal(V)$.

To show exactness at $\mfkgal(U)\oplus\mfkgal(V)$, suppose $f\in \mfkgal(U)\cap\mfkgal(V)$. From the proof of Prop. \ref{prop:semilinear-distance-bound}, there exist $r\geq 0,C_{UV}>0$ such that $d( x ,U^r\cap V^r)\leq C_{UV}\max(d( x ,U),d( x ,V))$. We may assume that $C_{UV}\geq 1$, in which case for any $ x \in\RR^n$ and any $k\in \NN$
\begin{align}
    (1+d(x,U^r\cap V^r))^k\leq C^k_{UV}\left(1+\max(d(x ,U),d(x,V))\right)^k.
\end{align}
Therefore for any $f\in\mfkgal(U)\cap\mfkgal(V)$ we have  
\begin{align}
    p_{k,U^r\cap V^r}(f)\leq C^k_{UV} \max(p_{k,U}(f),p_{k,V}(f)).
\end{align}
Since by Corollary \ref{cor:semilinear-distance-bound} one can represent $U\wedge V$ by $U^r\cap V^r$, we conclude that $f\in\mfkgal(U\wedge V)$.
\end{proof}

\subsection{Symmetries of lattice systems}

If $\Lambda\subseteq\RR^n$ is countable, it can be viewed as a lattice in the physical sense, and the local Lie system $\mfkgal$ over $\CS_n$ models infinitesimal gauge transformations of a lattice system on $\RR^n$. First, we generalize Example \ref{ex:onsiteu1action} by replacing $U(1)$ with an arbitrary compact Lie group. 
\begin{definition}
    An on-site action of a compact Lie group $G$ on a lattice system $(\Lambda,\{V_x\}_{x\in\Lambda})$ is a collection of continuous (and therefore smooth) homomorphisms $\rho_x:G\ra U(V_x)^{even}$ such that the norms of the corresponding Lie algebra homomorphisms $\mfkg\ra B(V_x)$ are bounded uniformly in $x$. 
\end{definition}
An on-site action of $G$ on $(\Lambda,\{V_x\}_{x\in\Lambda})$ gives rise to a homomorphism from $G$ to the automorphism group of $\SAl$ via $g\mapsto\otimes_{x\in\Lambda}{\rm Ad}_{\rho_x(g)}$. 
\begin{lemma}
    An on-site action of $G$ on $\SAl$ extends to an action on $\SAal$. This action is smooth in the sense that the corresponding map $G\times\SAal\ra\SAal$ is smooth with respect to the \Frechet norm. The generator of the action is a homomorphism $\frg\ra\mfkDal(\RR^n)$ which sends $a\in\frg$ to a function o
    \begin{align}\label{eq:QgeneralG}
     Y\mapsto\sum_{x\in Y\cap\Lambda}\derq_x(a)^Y.
\end{align}
where $Y\in\BB_n$ and $\derq_x$ is the traceless part of the generator of $\rho_x$. 
\end{lemma}
\begin{proof}
Since $\SAal$ is the completion of $\SAl$ with respect to the norms (\ref{eqn:SAal-norms}) and any on-site $G$ action preserves these norms, the action clearly extends to $\SAal$.
Let us denote by $\rho$ the corresponding map from $G$ to the space of linear maps $L(\SAal,\SAal)$. To prove the second statement, note that the map $\rho$ is continuous in the strong operator topology. Then it is a standard fact that to prove joint smoothness, it is enough to prove smoothness of the function $g\mapsto\rho(g)(\CA)$ for a fixed $\CA\in\SAal$ \cite{KrieglMichor}. The latter is straightforward, since the derivatives are given by explicit formulas. For example, the 1st derivative is
\begin{equation}
    d\rho(g)(\CA)(a)=\rho(g)(\derQ(a)(\CA)),\quad g\in G,a\in\mfkg,\CA\in\SAal,
\end{equation}
where 
\begin{align}
    \derQ(a)(\CA)=\sum_{x\in\Lambda}[\derq_x(a),\CA].
\end{align}
Since $\rho(g)$ is continuous in strong operator topology, so is its first derivative. The continuity of higher derivatives is proved similarly. Setting $g$ to be the identity of $G$, we find that the generator of the action is given by $\derQ$, and it is easy to see that the corresponding function $\BB_n\ra \mfkdl$ is given by (\ref{eq:QgeneralG}).
\end{proof}

This morphism of \Frechet-Lie algebras can be lifted to a morphism of local Lie systems $\mfkgal\ra\mfkDal$ over $\CS_n$. Indeed, for any $U\in\CS_n$ and any $f\in\mfkgal(U)$ we let $\derQ(f)$ be a derivation of $\SAal$ given by
\begin{align}
    \derQ(f)(\CA)=\sum_{x\in\Lambda}[\derq_x(f(x)),\CA],\quad\CA\in\SAal.
\end{align}
It is easy to check that this derivation belongs to $\mfkDal(U)$
and that the above map is a continuous homomorphism $\mfkgal(U)\ra\mfkDal(U)$. The physical interpretation is that an on-site action of a compact Lie group on a quantum lattice system  can be gauged on the infinitesimal level. 

\begin{definition}
    A state $\psi$ of $\SA$ is said to be invariant under an on-site action of a compact Lie group $G$ if it is invariant under the corresponding automorphisms of $\SA$.
\end{definition}
Let $\psi$ be a gapped state of $\SA$ invariant under an on-site action of a compact Lie group $G$. In that case the image of $\derQ:\mfkg\ra \mfkDal(\RR^n)$ lands in $\mfkDpal(\RR^n)$. One may ask if this morphism of \Frechet-Lie algebras can be lifted to a morphism of local Lie systems $\mfkgal\ra\mfkDpal$. If this is the case, then the symmetry $G$ of $\psi$ can be gauged on the infinitesimal level. In the next section we construct obstructions for the existence of such a morphism of local Lie systems and show that zero-temperature Hall conductance is an example of such an obstruction.

\subsection{Equivariantization}\label{sec:equiv}

As a preliminary step, for any $G$-invariant gapped state $\psi$ we are going to define a graded local Lie system over $\CS_n$ which is a $G$-equivariant version of the local Lie system $\mfkDpal$. Recall that a graded local Lie system is a cosheaf of graded vector spaces that is a pre-cosheaf of graded Lie algebras satisfying the graded analogue of Property I. For example, if $\mfkF$ is a local Lie system and $A=\prod_{k\in\ZZ} A_k$ is a locally finite supercommutative graded algebra with finite-dimensional graded factors $A_k$,\footnote{A graded vector space is locally finite iff its graded components are finite-dimensional.} then $U\mapsto \mfkF(U)\otimes A$ is a graded local Lie system. We denote it $\mfkF\otimes A$.

Fix a compact Lie group $G$ and a distributive pre-ordered set $X$ and consider the category of  graded local Lie systems over $X$ equipped with a $G$-action. An object of this category is a graded local Lie system $\mfkF$ on which $G$ acts by automorphisms; morphisms are defined in an obvious manner. The first step is to define a functor $\mfkF\mapsto \mfkF^G$ from this category to the category of graded local Lie systems over $X$ such that $\mfkF^G(U)$ is the Lie algebra of $G$-invariant elements of $\mfkF(U)$. It is clear how to define such a functor for coflasque pre-cosheaves of Lie algebras with Property I, but the pre-cosheaf $\mfkF^G$ will not be a cosheaf of vector spaces without further assumptions about $\mfkF$ and the $G$-action. 
\begin{definition}
    An action of $G$ on a pre-cosheaf of \Frechet\ spaces  $\mfkF$ is smooth if for each $U\in X$ the seminorms defining the topology of $\mfkF(U)$ can be chosen to be $G$-invariant and the map $G\times\mfkF(U)\ra \mfkF(U)$ defining the action is smooth.
     An action of $G$ on a pre-cosheaf of graded \Frechet\ spaces is smooth if the $G$-action on every graded component is smooth.
\end{definition}

\begin{prop}
    Let $\mfkF$ be a cosheaf of graded \Frechet\  spaces over $X$  equipped with a smooth action of a compact Lie group $G$. The assignment $U\mapsto \mfkF^G(U)= \left(\mfkF(U)\right)^G$ is a cosheaf of graded \Frechet\ spaces. 
\end{prop}
\begin{proof}
     We need to show that for any $U,V\in X$ the sequence 
    \begin{align}
    \mfkF^G(U \wedge V) \xrightarrow{\alpha} \mfkF^G(U)\oplus \mfkF^G(V) \xrightarrow{\beta} \mfkF^G(U \vee V) \to 0,    
    \end{align}
    is exact. To show exactness at the rightmost term, let $\derF$ be a $G$-invariant element of $\mfkF(U\vee V)$ and let $\derF_U\in\mfkF(U)$ and $\derF_V\in\mfkF(V)$ be such that $\iota_{U\vee V,U}\derF_U+\iota_{U\vee V,V}\derF_V=\derF$. Averaging the action map $G\times\mfkF(U)$ over $G$ with the Haar measure gives a linear map $h_U:\mfkF(U)\ra\mfkF^G(U)$ which is identity when restricted to  $\mfkF^G(U)$. The co-restriction morphisms intertwine these maps. Thus $\iota_{U\vee V,U}\circ h_U(\derF_U)+\iota_{U\vee V,V}\circ h_V(\derF_V)=\derF$ which proves that $\beta$ is surjective. Exactness in the middle term is proved similarly.
\end{proof}
\begin{remark}\label{rem:smooth}
    For the proof to go through, it suffices to require the map $G\times\mfkF(U)\ra \mfkF(U)$ to be continuous. However, if it is smooth, $\mfkF(U)$ becomes a $\frg$-module and all elements in  $\mfkF^G(U)\subseteq \mfkF(U)$ are annihilated by the $\frg$-action. We will use this later on. 
\end{remark}

\begin{example} \label{ex:mfkDpalG}
If $G$ acts on-site on a lattice system, the action of $G$ on the local Lie system $\mfkDal$ is smooth. If $\psi$ is a $G$-invariant gapped state of such a lattice system, the action of $G$ on $\mfkDpal$ is smooth. 
\end{example}
\begin{example}\label{ex:mfkgalG}
Consider the local Lie system $\mfkgal$ over $\CS_n$  associated to a finite-dimensional Lie algebra $\mfkg$ (see Section \ref{sec:semilinearsets}). Assume that $\mfkg$ is the Lie algebra of a compact Lie group $G$, then there is an obvious $G$-action on $\mfkgal$:
\begin{align}
    (g\cdot f)(x)={\rm Ad}_g f(x),\quad g\in G,x\in\Lambda.
\end{align}
This $G$-action is smooth. 
\end{example}
\begin{example} If $\mfkF$ is a local \Frechet-Lie system with a smooth $G$-action and $A=\prod_{k\in\ZZ} A_k$ is a locally-finite supercommutative graded algebra on which $G$ acts by automorphisms, then the $G$-action on $\mfkF\otimes A$ is smooth. 
\end{example}

\begin{corollary}\label{cor:Ginv}
    Let $\mfkF$ be a graded local \Frechet-Lie system over $X$ with a smooth $G$-action. The functor of $G$-invariant elements maps $\mfkF$ to a graded local \Frechet\-Lie system $\mfkF^G$ over $X$.
\end{corollary}
\begin{definition}
    Let $\mfkF$ be a local \Frechet-Lie system over $X$ with a smooth $G$-action. The $G$-equivariantization functor sends $\mfkF$ to the negatively-graded local \Frechet-Lie system $\mfkF^{\bf G}$ defined by 
    \begin{align}
        U\mapsto \left(\mfkF(U)\otimes\prod_{k=1}^\infty \Sym^k(\frg^*[-2])\right)^G.
    \end{align}
\end{definition}
In the cases of interest to us, the $G$-action on a local Lie system over $\CS_n$ is infinitesimally inner, in the sense that the $\frg$-module structure mentioned in Remark \ref{rem:smooth} arises from a homomorphism $\rho:\frg\ra \mfkF(\RR^n)$. In such a case, the graded local Lie system $\mfkF^{\bf G}$ has an extra bit of structure: a central element in $\mfkF^{\bf G}(\RR^n)$ of degree $-2$. This element is simply $\rho$ re-interpreted as an element of $\mfkF(\RR^n)\otimes\frg^*[-2]$. In the terminology of Section \ref{sec:locality}, $\mfkF^{\bf G}$ is a pointed graded local \Frechet-Lie system over $\CS_n$. 
It is easy to see that the $G$-equivariantization functor respects this extra structure. That is, if $f:\mfkF\ra \mfkF'$ is a morphism of local \Frechet-Lie system over $\CS_n$ commuting with infinitesimally inner smooth $G$-actions on $\mfkF$ and $\mfkF'$, then $f^{\bf G}$ maps the central element $\rho\in \mfkF^{\bf G}(\RR^n)_{-2}$ to the central element $\rho'\in{\mfkF'}^{\bf G}(\RR^n)_{-2}$.

\begin{example}
    Let $\psi$ be a $G$-invariant gapped state of a quantum lattice system with an on-site $G$-action which on the infinitesimal level is described by $\derQ:\mfkg\ra \mfkDpal(\RR^n)$. Consider the local Lie system $\mfkDpal$ with its smooth $G$-action (Example \ref{ex:mfkDpalG}) and its $G$-equivariantization $\mfkDpalG$. The distinguished central element of $\mfkDpalG(\RR^n)$ is $\derQ$ regarded as an element of $\mfkDpal(\RR^n)\otimes \frg^*[-2]$.
\end{example}
\begin{example}\label{ex:mfkgalGpointed}
    Consider the local Lie system $\mfkgal$ of Example \ref{ex:mfkgalG}. The degree $-2$ component of $\mfkgalG(\RR^n)$ is the space of $G$-invariant bounded functions on $\Lambda$ with values in $\frg\otimes\frg^*$. The distinguished central element is the constant function on $\Lambda$ which takes the value ${\rm id}_\frg$.
\end{example}

Armed with the equivariantization functor, we can now explain our strategy for constructing obstructions for the existence of a local Lie system morphism $\mfkgal\ra \mfkDpal$ which lifts the \Frechet-Lie algebra morphism $\derQ:\mfkg\ra\mfkDpal(\RR^n)$. Suppose such a morphism $\rho$ exists. Applying the $G$-equivariantization functor, we get a morphism of pointed negatively-graded local \Frechet-Lie systems $\rho^{\bf G}:\mfkgalG\ra \mfkDpalG$. For any $\CS_n$-cover $\fU$ of $\RR^n$ an application of the \Cech\ functor gives a morphism of acyclic pointed DGLAs 
\begin{equation}
C^{aug}_{\bullet+1}(\fU,\RR^n;\mfkgalG)\ra C^{aug}_{\bullet+1}(\fU,\RR^n;\mfkDpalG).
\end{equation}
Consequently, an  obstruction for the existence of such a pointed DGLA morphism is an obstruction for the existence of $\rho$. In the next section we use the twisted Maurer-Cartan equation for pointed DGLAs to construct such obstructions and identify them as topological invariants of gapped states defined in \cite{LocalNoether}.

\section{Topological invariants of $G$-invariant gapped states}\label{sec:invariants}

In this section we use the twisted Maurer-Cartan equation for a pointed DGLA to construct topological invariants of gapped $G$-invariant states. For states of quantum lattice systems on $\RR^n$, these invariants take values in $G$-invariant polynomials on the Lie algebra of $G$. We show that the Hall conductance of a $U(1)$-invariant gapped lattice system on $\RR^2$ is a special case of this construction. We also discuss how to extend the theory to quantum lattice systems on a polyhedral subset of $\RR^n$. In that case, topological invariants take values in $G$-invariant polynomials on the Lie algebra of $G$ valued in the cohomology of the corresponding spherical polyhedral subset of $S^{n-1}$ (``the sphere at infinity'').

\subsection{The commutator class}\label{sec:derivationvaluedinvariants}

Let $(\mfkM,\derB)$ be a pointed DGLA, i.e. a DGLA with a distinguished central cycle $\derB\in\mfkM_{-2}$. Assume it is a limit of an inverse system of nilpotent pointed DGLAs $(\mfkM_N,\derB_N)$, $N\in\NN$. 
\begin{definition}
    The commutator DGLA, denoted by $\clcomm$, is defined to be the closure of the commutator subalgebra of $\mfkM$. Namely, $\derq \in \clcomm$ if and only if for any $N$ its projection to $\mfkM_N$ is a finite linear combination of commutators in $\mfkM_N$. 
\end{definition}

Even if $\mfkM$ is acyclic, the DGLA $\clcomm$ is not necessarily acyclic. We would like to construct an obstruction to finding an element $\derp\in\mfkM_{-1}$ which satisfies $\partial\derp=\derB$ and $[\derp,\derp]=0$. 

\begin{definition}
    Let $(\mfkM,\derB)$ be a pointed DGLA. A $\derB$-twisted Maurer-Cartan element in $\mfkM$ is $\derp\in\mfkM_{-1}$ which satisfies
    \begin{align}
        \partial \derp +\frac12[\derp,\derp]=\derB.
    \end{align}
\end{definition}
We will denote the set of $\derB$-twisted MC elements of $\mfkM$ by $\MC(\mfkM,\derB)$. The map $(\mfkM,\derB)\mapsto \MC(\mfkM,\derB)$ can be upgraded to a functor from the category of pointed DGLAs to the category of sets in an obvious way. 

Let $(\mfkM,\derB)$ be a pronilpotent pointed DGLA and $\derp\in \MC(\mfkM,\derB)$. Then  $[\derp,\derp]$ is a cycle of the DGLA $\clcomm$. 
  
\begin{prop}\label{prop:MCmain}
    Let $(\mfkM,\derB)$ be an acyclic pointed DGLA which is a limit of an inverse system of nilpotent acyclic pointed DGLAs $(\mfkM_N,\derB_N)$, $N\in\NN$. Assume further that the structure morphisms $r_{N,N-1}:\mfkM_N\ra\mfkM_{N-1}$ are surjective and $\fj_N=\ker r_{N,N-1}$ is central in $\mfkM_N$. Finally, assume $[\mfkM_1, \mfkM_1]=0$. Then $\MC(\mfkM,\derB)$ is non-empty. Furthermore, the homology class of $[\derp,\derp]$ in $\clcomm$ for $\derp\in\MC(\mfkM,\derB)$ is independent of the choice of $\derp$. 
\end{prop}
A proof of this result can be found in Appendix \ref{sec:deformation}. It uses some results from deformation theory. We will call the homology class of $[\derp,\derp]$ the commutator class of the acyclic pointed DGLA $(\mfkM,\derB)$, for lack of a better name, and denote it $\com(\mfkM,\derB)$.
The assignment $(\mfkM,\derB)\mapsto (H_\bullet(\clcomm),\com(\mfkM,\derB))$ is a functor from the full sub-category of the category of acyclic pointed DGLAs satisfying the conditions of Proposition \ref{prop:MCmain}, to the category $\pVectZ$ of pointed graded vector spaces.

The conditions of Prop. \ref{prop:MCmain} apply to any pointed DGLA which is the value of the \Cech\ functor on a graded local Lie system $\mfkF^{\bf G}$ over some $(X,\leq)$, where $\mfkF$ is a local Lie system equipped with a smooth infinitesimally inner $G$-action. Indeed, along with the graded algebra $A=\prod_{k=1}^\infty \Sym^k(\frg^*[-2])$ used in the construction of $\mfkF^{\bf G}$ one can consider its quotient by the ideal $J_N=\prod_{k=N+1}^\infty \Sym^k(\frg^*[-2])$. Replacing $A$ with the nilpotent graded algebra $A/J_N$ throughout, for any cover $\fU$ of the terminal object $T$ one gets a sequence of nilpotent acyclic pointed DGLAs labeled by $N\in\NN$. They assemble into an inverse system in an obvious manner, and its limit is the acyclic pointed DGLA $C^{aug}_{\bullet+1}(\fU,T;\mfkF^{\bf G}).$ It is easy to see that the remaining conditions of Prop. \ref{prop:MCmain} are also satisfied. In particular, Prop. \ref{prop:MCmain} applies to the pointed DGLAs associated to the graded local Lie systems $\mfkDpalG$ and $\mfkgalG$ and any CS cover of $\RR^n$. 

\begin{prop}\label{prop:MC_functoriality}
    For any pointed DGLA $(\mfkM(\fU),\derB)$ obtained,  as above, from the \Cech\ functor with respect to a cover $\fU$, $\MC$ is functorial in $\fU$.
\end{prop}
\begin{proof}
    Let $\fV=\{V_j\}_{j\in J}$ be a refinement of $\fU$. By definition, there exists a map $\phi: J\rightarrow I$ with $V_j\leq U_{\phi(j)}.$ According to Prop. \ref{prop:DGLA_functoriality}, there is a map from $\mfkM(\fV)$ to $\mfkM(\fU)$. As $\derB$ is unaffected by $\phi_*$, we deduce from $\derp^\fV \in \MC(\mfkM(\fV),\derB)$ that $\phi_*\derp^\fV \in \MC(\mfkM(\fU),\derB)$.
\end{proof}
\begin{prop}\label{prop:trivialcomclass}
    Let $G$ be a compact Lie group, $\mfkg$ be its Lie algebra, and $\mfkgalG$ be the pointed graded local Lie system over $\CS_n$ of Example   (\ref{ex:mfkgalGpointed}). 
    For any cover $\fU=\{U_i\}_{i\in I}$ the commutator class vanishes.
\end{prop}
\begin{proof}
    The distinguished central cycle in $C^{aug}_{\bullet+1}(\fU,\RR^n,\mfkgalG)$ is the constant function on $\Lambda$ with value  $ {\rm id}_\frg$ . Here ${\rm id}_\frg$ is regarded as a $G$-invariant element of $\mfkgal(\RR^n)\otimes\frg^*[-2]$. For any cover $\fU=\{U_i\}_{i\in I}$ one can construct a twisted MC element $\derq$ as follows. Pick $r>0$ large enough so that the interiors of $U_i^r$, $i\in I$, cover $\RR^n$ in the usual sense and pick a partition of unity $\chi_i$, $i\in I$, subordinate to this open cover. For any $x\in \Lambda$ and any $i\in I$ let $\derq_i(x)=\chi_i(x) {\rm id}_\frg $. It is easy to verify that this is a twisted MC element satisfying $[\derq,\derq]=0$. Thus, the commutator class vanishes.
\end{proof}
\begin{corollary}
    Let $G$ be a compact Lie group, $\psi$ be a gapped $G$-invariant state of a lattice system on $\RR^n$, and $\fU$ be a $\CS_n$-cover of $\RR^n$. The commutator class of the acyclic pointed DGLA $(C^{aug}_{\bullet+1}(\fU,\RR^n;\mfkDpalG),\derQ)$ is an obstruction for the existence of a morphism of local Lie systems $\rho:\mfkgal\ra \mfkDpal$ which is a lift of the Lie algebra morphism $\derQ:\mfkg\ra \mfkDpal(\RR^n)$.
\end{corollary}
\begin{proof}
If $\rho$ exists, it induces a morphism of pointed DGLAs $C^{aug}_{\bullet+1}(\fU,\RR^n;\mfkgalG)\ra C^{aug}_{\bullet+1}(\fU,\RR^n;\mfkDpalG)$ which in turn induces a morphism in the category $\pVectZ$ which maps the commutator class of $\mfkgal$ to the commutator class of $\mfkDpalG$. Since the former class vanishes, so must the latter.
\end{proof}

\subsection{Construction of topological invariants}\label{sec:construction}

The commutator class defined in the previous section is not a useful invariant of a gapped $G$-invariant state $\psi$ because it takes values in a set which itself depends on $\psi$ and the choice of the cover. It is also not invariant under LGAs (defined in Section \ref{sec:LGAs}) and thus is not an invariant of a gapped phase (see Remark \ref{rem:LGAsandphases}). In this section we define a pairing between the commutator classes of $\mfkDpalG$ and spherical CS cohomology classes of the sphere at infinity $S^{n-1}$, which takes values in the algebra of $G$-invariant symmetric polynomials on $\frg$. This gives a useful invariant of a gapped phase which is also an obstruction to promoting the global symmetry $G$ of the state  to a local symmetry. Additionally, we will show that the invariant  does not depend on the choice of the cover $\fU$ and is controlled by the large-scale topology. 

Keeping in view a generalization to quantum lattice systems on general CS sets in $\RR^n$, we work over a sub-site $\CSnW$\footnote{This is the so-called overcategory whose objects are equipped with a morphism to $W$. Due to coflasqueness, this amounts to a restriction to objects fuzzily included in $W$.} which depends on an arbitrary $W\in \CS_n$. Let ${\widehat W}\in \SCS_n$ be the image of $W$ under the equivalence between $\CS_n$ and $\SCS_n$. For any cover $\fU$ of $W$ we denote by $\hat\fU$ the corresponding cover of $\widehat W$.

\begin{definition}
    Let $\mfkF$ be a local Lie system over $\CSnW$ and let $p\geq 0$. For any $\fU=\{U_i\}_{i\in I}$  covering $W\in \CS_n$, any $\derf\in C_p(\fU,W;\mfkF)$, and any $\beta\in \check{C}^p(\hat\fU,\widehat W;\RR)$ define an evaluation 
    \begin{align}
        \langle \derf,\beta\rangle=\sum_{s\in I^{p+1}}\beta_s \derf_s\in \mfkF(W),
    \end{align}
 We adopt the convention that $\beta_s=0$ for $\bigwedge_{j\in s}U_j$ bounded. The definition implicitly uses co-restriction of the coflasque cosheaf $\mfkF$.
\end{definition}
Let $\mfkF$ be a local Lie system over $\CSnW$. By analogy with Def. \ref{def:anchor}, we will say that $\dera\in\mfkF(W)$ is {\it anchored} if $\dera$ belongs to the sub-algebra $\mfkF(\{0\})$ (or equivalently, to $\mfkF(U)$ for some bounded $U$). 

\begin{lemma} \label{lma:cocycle_inner}
    Let $\mfkF$ be a local Lie system over $\CSnW$ and $p\geq 0$. For any $\fU$  covering $W\in \CS_n$, any cycle  $\derf\in C^{aug}_p(\fU,W;\mfkF)$, and any cocycle $\beta\in \check{C}^p(\hat\fU,\widehat W;\RR)$ the element $\langle \derf,\beta\rangle\in\mfkF(W)$ is anchored. 
\end{lemma}
\begin{proof}
The augmented chain complex 
\begin{equation}\label{eq:complex}
\ldots \ra {\rm Alt}\left[\bigoplus_{i,j\in I} \mfkF(U_i\wedge U_j)\right]\ra \bigoplus_{i\in I}\mfkF(U_i)\ra \mfkF(W) \ra 0 
\end{equation} is acyclic. Since $\derf$ is a cycle, $\derf = \partial \derg$ for some $\derg\in {\rm Alt}\left[\bigoplus_{s\in I^{p+2}} \mfkF(\bigwedge_{j\in s} U_j)\right].$ Thus
\begin{align}
    \langle \derf,\beta\rangle=&\langle \partial\derg,\beta\rangle\notag\\
    =&\langle \derg, \partial^* \beta\rangle\notag\\
    =&\sum_{s\in I^{p+2}} \derg_s (\partial^* \beta)_s\notag\\
    =&\sum_{\substack{s\in I^{p+2},\ \mathrm{with}\\ \bigwedge_{j\in s}U_j\ \mathrm{ bounded}}} \derg_s (\partial^* \beta)_s,
\end{align}
where $\partial^*$ is the adjoint of $\partial$ defined by
\begin{align}(\partial^* \beta)_{i_0\dots i_{p+1}}=\sum_{k=0}^{p+1} (-1)^{k} \beta_{i_0\dots\widehat{i_k}\dots i_{p+1}}.
\end{align}
The last expression is clearly anchored because each $\derg_s$ is almost localized near some bounded region. The last equality depends on the vanishing of 
\begin{align}
\sum_{\substack{s\in I^{p+2},\ \mathrm{with}\\ \bigwedge_{j\in s}U_j\ \mathrm{ unbounded}}} \derg_s (\partial^* \beta)_s,
\end{align} 
because when $\bigwedge_{j\in s}U_j$ is unbounded $(\partial^* \beta)_s=(\check{\delta} \beta)_s = 0$ where $\check{\delta}$ is the coboundary of $\check{C}^p(\hat\fU,\widehat W;\RR)$. 
\end{proof}
\begin{remark}
The relation between $\partial^*$ and $\check{\delta}$ is as follows. They are equal for $s\in I^{p+2}$ with $\bigwedge_{j\in s}U_j$ unbounded. For $s$ with $\bigwedge_{j\in s}U_j$ bounded, $\bigcap_{j\in s}\hat U_j=\varnothing$ and $(\partial^* \beta)_s$ is, in general, nonzero whereas $(\check{\delta} \beta)_s=0$. This discrepancy arises because the Grothendieck topology on the poset of spherical CS sets corresponding to Definition \ref{def:coherent} does not allow the empty cover of the empty set. It is this discrepancy that  makes the evaluation of DGLA cycles on spherical CS cocycles non-trivial. 
\end{remark}

We are going to apply this lemma to the pointed graded local Lie system $\mfkDpalG$ arising from a $G$-invariant gapped state of a quantum lattice system on $W\subseteq\RR^n$. We take $\derf$ to be the commutator class $[\derp,\derp]$ defined via the twisted Maurer-Cartan equation. In this case $\Sigma\langle \derf,\beta\rangle$ takes values in the $G$-invariant part of $\mfkdpal\otimes A$ where $A=\prod_{k=1}^\infty \Sym^k(\frg^*[-2])$. (Recall that $\Sigma$ is the isomorphism of the space of anchored elements of $\mfkDpal$ and $\mfkdpal$).
The final step in the construction of the topological invariant is to evaluate the expectation value $\psi(\Sigma\langle [\derp,\derp],\beta\rangle)$. Since $\psi$ is $G$-invariant, it takes values in the complexification of $A^G$. The main result of this section is
\begin{theorem}\label{thm:main}
    Let $\psi$ be a gapped state of a quantum lattice system on $W\in\CS_n$ invariant under an on-site action of a compact Lie group $G$. There is a function $\nu_\psi:\check{H}^\bullet_{CS}(\widehat W,\RR)\ra A^G$ which depends only on $\psi$, and whose value on a cocycle $\beta\in \check{C}^\bullet(\hat\fU,{\widehat W};\RR)$ is given by
    \begin{align}
        \nu_\psi(\beta)= {\sqrt {-1}}\psi(\Sigma\langle [\derp,\derp],\beta\rangle).
    \end{align}
    The function $\nu_\psi$ vanishes on even-degree cohomology.
\end{theorem}

\begin{lemma}\label{lma:psicommutator}
    Let $U$ and $V$ be closed semilinear sets such that $U\wedge V$ is bounded. Let $A$ be a locally-finite supercommutative graded algebra and $\psi$ be a state. Then $\psi(\Sigma[\mfkDpal(U)\otimes A,\mfkDal(V)\otimes A])=0.$
\end{lemma}
\begin{proof}

It suffices to consider the case $A=\RR$. We begin by showing that the sum
\begin{align}
    \sum_{X,Y,Z\in \BB_n}\pi^Z([\derF^X,\derG^Y])
\end{align}
is absolutely convergent. Using Lemmas \ref{lma:brick-component-norm} and \ref{lem:brickjoin}
\begin{align}
    \sum_{X,Y\in \BB_n}\sum_{Z\in\BB_n}\|\pi^Z([\derF^X,\derG^Y])\| &\le 
    \sum_{X,Y\in \BB_n}\sum_{Z\subseteq X\wedge Y}4^d2\|\derF^X\|\|\derG^Y\|\notag\\
    &\le \sum_{X,Y\in \BB_n}(1+\diam(X\wedge Y))^{2d} 4^d2\|\derF^X\|\|\derG^Y\|\notag\\
    &\le \sum_{X,Y\in \BB_n}(1+\diam(X)+\diam(Y))^{2d} 4^d2\|\derF^X\|\|\derG^Y\|\notag\\
    &\le 4^d2\left(\sum_{X\in \BB_n}(1+\diam(X))^{2d} \|\derF^X\|\right)\left(\sum_{Y\in \BB_n}(1+\diam(Y))^{2d} \|\derG^Y\|\right)\notag\\
    &<\infty
\end{align}
Recall that for an anchored derivation $\derH$ we define $\psi(\derH) = \psi(\Sigma(\derH))$. We have
\begin{align}
    \psi(\Sigma([\derF,\derG])) &= \sum_{X,Y,Z \in \BB_n }\psi(\pi^Z([\derF^X,\derG^Y])\notag\\
    &= \sum_{X,Y \in \BB_n}\psi([\derF^X,\derG^Y])\notag\\
    &= \sum_{Y\in\BB_n}\psi(\derF(\derG^Y))\notag\\
    &= 0.
\end{align}
\end{proof}

\begin{lemma}\label{lma:commutatorDGLA}
    Let $A$ be a locally-finite supercommutative graded algebra and $\psi$ be a gapped state of a quantum system on $W\in\CS_n$. Fix a cover $\fU$ of $W$. For any cycle $\derf$ in the commutator DGLA of $C^{aug}_{\bullet+1}(\fU,W;\mfkDpal\otimes A)$, the expectation value $\psi(\Sigma\langle\derf,\beta\rangle)\in A$ depends only on the cohomology class of $\beta\in \check{C}^\bullet(\hat\fU,\widehat W;\RR)$ and the homology class of $\derf$. 
\end{lemma}
\begin{proof}
Suppose $\derf\in C^{aug}_p(\fU,W;\mfkDpal\otimes A)$ and $\derf=\partial\derg$ for some $\derg$ in the commutator DGLA. Then the vanishing of $\psi(\Sigma\langle\derf,\beta\rangle)\in A$ follows from the proof of Lemma \ref{lma:cocycle_inner} and Lemma \ref{lma:psicommutator}. Thus, it remains to show that $\psi(\Sigma\langle\derf,\check{\delta} 
 b\rangle)=0$ for any $b\in \check{C}^{p-1}(\hat\fU,\widehat W)$. Indeed, since $\langle \derf,\partial^* b\rangle=\langle \partial\derf, b \rangle=0$, we have 
\begin{align}
    \langle \derf,\check{\delta} b\rangle
    =\sum_{\substack{s\in I^{p+1},\ \mathrm{with}\\ \bigwedge_{j\in s}U_j\ \mathrm{ unbounded}}} \derf_s (\partial^* b)_s=-
    \sum_{\substack{s\in I^{p+1},\ \mathrm{with}\\ \bigwedge_{j\in s}U_j\ \mathrm{ bounded}}} \derf_s (\partial^* b)_s,
\end{align}
and the $\psi$-average of each term in the latter sum vanishes by Lemma \ref{lma:psicommutator}.
\end{proof}

\begin{proof}[Proof of Theorem \ref{thm:main}]
    Note first that the anchored derivation $[\derp,\derp]$ is anti-hermitian, hence its expectation value is purely imaginary, and $\nu_\psi(\beta)$ is real. Also, since the commutator class has even degree and thus odd \Cech\ degree, $\nu_\psi(\beta)=0$ if $\beta$ has even degree. For a fixed covering $\fU$, Prop. \ref{prop:MCmain} and Lemma \ref{lma:commutatorDGLA} imply $\psi(\Sigma\langle [\derp,\derp],\beta\rangle)$ is independent of the choice of the solution $\derp$ of the twisted MC equation. By Lemma \ref{lma:commutatorDGLA}, it depends on the cocycle $\beta$ solely through its cohomology class. It remains to show invariance under refinement of the covering  $\fU$.
    
    Let $(\fV, \phi)$ be a refinement of $\fU$ with $V_j\leq U_{\phi(j)}.$ Cocycles $(\fU, \beta)$ and $(\fV, \phi^* \beta)$, where $(\phi^*\beta)_{j_0, \dots, j_k}= \beta_{\phi(j_0), \dots, \phi(j_k)}$, are in the same CS cohomology class. From Prop.\ref{prop:MCmain} there exists $\derQ$-twisted MC element $\derp$ for the cover $\fV$. Prop.\ref{prop:MC_functoriality} then implies that $\phi_*\derp$ is a $\derQ$-twisted MC element for the cover $\fU$. An easy expansion shows 
    \begin{equation}
    \langle [\phi_*\derp, \phi_*\derp], \beta \rangle= \langle [\derp, \derp], \phi^*\beta \rangle.
    \end{equation}
    Since any $\derQ$-twisted MC element gives the same answer, we have shown that this contraction of interest depends only on the CS cohomology class of $(\fU, \beta)$.
\end{proof}

Theorem \ref{thm:main} defines an invariant of $G$-invariant gapped states of quantum lattice systems on arbitrary CS sets in $\RR^n$. In view of Prop. \ref{cech_to_singular}, one can view this invariant as a function from $H^\bullet_{odd}(\widehat W,\RR)$ to $G$-invariant formal power series on $\frg$. In fact, since the cohomology of $\widehat W$ vanishes in degree larger than $n-1$, the values of the invariant are polynomials on $\frg$ whose degrees do not exceed $(n+2)/2$. 

In the case $W=\RR^n$, the odd cohomology is nontrivial only for even $n$ and is a  one-dimensional vector space of degree $n-1$. Once the orientation of $\RR^n$ has been chosen, there is a canonical basis element, the dual of the fundamental class of $S^{n-1}$. Thus, for $G$-invariant gapped states of quantum lattice systems on $\RR^n$ with $n$ even, Theorem \ref{thm:main} supplies a unique invariant taking values in $G$-invariant polynomials on $\frg$ of degree $(n+2)/2$. This invariant changes sign when the orientation of $\RR^n$ is changed. Note that the same datum classifies characteristic classes of $G$-bundles with real coefficients in degree $n+2$, or equivalently, classical Chern-Simons field theories with gauge group $G$ in space-time dimension $n+1$.

In Section \ref{sec:symmetries} we defined the notion of a smooth path of gapped states. One can extend it to $G$-invariant gapped states in an obvious manner.
\begin{definition}\label{def:smoothfamiliesofstates}
    Fix an on-site action of a compact Lie group $G$ on a quantum lattice system. A smooth path of $G$-invariant LGAs is a smooth path of LGAs $u\mapsto\alpha(u)$ such that $\alpha(u)$ commutes with the $G$-action for all $u\in [0,1]$. A smooth path of $G$-invariant gapped states is a smooth path of states of the form $u\mapsto \psi(u)=\omega\circ\alpha(u)$, where $\omega$ is a $G$-invariant gapped state and $u\mapsto\alpha(u)$ is a smooth path of $G$-invariant LGAs. 
\end{definition}
The following proposition shows that the function $\nu_\psi$ does not vary in smooth families of states.  This  justifies its interpretation as a topological invariant of the state $\psi$.  

\begin{prop}
    Let $u\mapsto \psi(u)$ be a smooth path of $G$-invariant gapped states of a quantum lattice system on $W\in\CS_n$. Then  $\nu_{\psi(u)}=\nu_{\psi(0)}$ for all $u\in [0,1]$.
\end{prop}
\begin{proof}
    By definition of a smooth path, for any $u\in [0,1]$ there exists a smooth path of LGAs $u\mapsto \alpha(u)$ commuting with the $G$-action such that $\psi(u)=\psi(0)\circ\alpha(u)$. Therefore, $\mfkDal^{\psi(u),{\mathbf G}}(U)=\alpha(u)^{-1}(\mfkDal^{\psi(0),{\mathbf G}}(U))$ for every $U\in\CS_n$. Hence for every $\CS$ cover $\fU$ the graded DGLAs corresponding to these two local Lie systems are related by conjugation with $\alpha$ inside $C^{aug}_{\bullet+1}(\fU,W;\mfkDal^{\mathbf G})$. Furthermore, these DGLAs are pointed and their curvatures coincide. Indeed, the curvature is the generator $\derQ$ of the $G$-action regarded as the distinguished degree $-2$ element of $\mfkDal^{\mathbf G}(W)$ which preserves $\psi(u)$ for all $u\in[0,1]$. Finally, if $\derp(0)$ is a twisted MC element corresponding to $\psi(0)$, then $\alpha(u)^{-1}(\derp(0))$ is a twisted MC element for $\psi(u)$. This implies $\nu_{\psi(u)}=\nu_{\psi(0)}$.
\end{proof}

\subsection{The Hall conductance}
For physics applications, the most important case is $G=U(1)$ and $n=2$. Let us specialize the construction of topological invariants to this case. 

Let $\psi$ be a $U(1)$-invariant gapped state of a lattice system on $\RR^n$. The generator of the $U(1)$ action is $\derQ\in\mfkDpal(\RR^n)$. Let $\mfkDpQal$ be the sub-algebra of $U(1)$-invariant elements of $\mfkDpal$. Since $U(1)$ is connected, this is the same as the sub-algebra of elements of $\mfkDal$ which commute with $\derQ$. More generally, for any $U\in\CS_n$ $\mfkDpQal(U)=\mfkDal(U)\cap\mfkDpQal$. This is a local Lie system over $\CS_n$. The graded local Lie system $\mfkDpalG$ reduces in this case to $\mfkDpQal\otimes\RR[[t]]$ where $t$ is a variable of degree $-2$.

Pick a cover $\fU=\{U_i\}_{i\in I}$ of $\RR^n$. To find a solution $\derp$ of the inhomogeneous Maurer-Cartan equation with $\derB=\derQ\otimes t$, we write $\derp=\sum_{k=1}^\infty \derp_k\otimes t^k$, where $\derp_k\in C_{k-1}(\fU,\RR^n;\mfkDpQal)$. 
To compute the topological invariant of a state on $\RR^2$ it is sufficient to solve for $\derp_1$. 

$\derp_1$ is a solution of $\partial\derp_1=\derQ$. Explicitly, $\derp_1=\{\derQ_i\in\mfkDpQal(U_i)\}_{i\in I}$ such that $\sum_{i\in I}\derQ_i=\derQ$. Such $\derQ_i$  exist because $\mfkDpQal$ is a cosheaf. Then the component of the commutator class in $C_1(\fU,\RR^n;\mfkDpQal)$ is $\oplus_{i,j}[\derQ_i,\derQ_j]$. The topological invariant of the state $\psi$ is obtained by evaluating it on a \Cech\ 1-cocycle $\beta$ on $S^1$ corresponding to the cover $\hat \fU$ and then averaging the resulting anchored derivation:
\begin{align}
\nu_\psi(\beta)={\sqrt {-1}}\psi\left(\sum_{i,j}\beta_{ij}[\derQ_i,\derQ_j]\right).
\end{align}
Note that averaging over $\psi$ must be performed after the summation over $i,j$ because $[\derQ_i,\derQ_j]$ is not an anchored  derivation, in general.

The simplest cover of $\RR^2$ which can represent a nontrivial class in $H^1_{CS}(S^1,\RR)$ involves three cones $U_i$, $i\in\{1,2,3\},$ with a common vertex. A cocycle representing the canonical generator of  $H^1(S^1,\ZZ)$ is $\beta_{12}=-\beta_{21}=1$, $\beta_{13}=\beta_{23}=0$. Then the invariant is given by
\begin{align}\label{eq:sigma2d}
\nu_\psi(\beta)=2{\sqrt {-1}}\psi(\Sigma[\derQ_1,\derQ_2])\in\RR.
\end{align}
It was shown in Section 4.1 of \cite{ThoulessHall} (see also Section 4.4.2 of \cite{LocalNoether}) that $\nu_\psi(\beta)$ is the zero-temperature Hall conductance of the system in units of $e^2/\hbar$, where $e$ is the electron charge and $\hbar$ is the Planck constant. 

\begin{example}
    Consider a gapped Fock state of the CAR algebra associated with a lattice $\Lambda\subseteq\RR^2$ and a single-site Hilbert space $H$ (Example \ref{ex:Fockstate}). It depends on a projector $P$ acting on $\ell^2(\Lambda,H)$ . If the cover  $U_i$, $i=1,2,3$, is chosen as above, one can take $\derQ_i$ to have the form $\derF(h_i)$ where $h_i$ are operators on $\ell^2(\Lambda,H)$ given by (\ref{eq:hi}). Evaluating the expectation value (\ref{eq:sigma2d}), one finds
    \begin{align}\label{eq:sigma2d3cones}
        -{\sqrt {-1}}\nu_\psi(\beta)=2{\rm Tr}\, P[h_1,h_2]=2{\rm Tr}\,(Pf_1Pf_2P-Pf_2Pf_1P)=2{\rm Tr}[Pf_1P,Pf_2P],
    \end{align}
    where the functions $f_i:\Lambda\ra \RR$ are given by (\ref{eq:fi}). Note that the operators $Pf_iP$ are not trace class, so one cannot use the cyclic property of the trace to conclude that  $\nu_\psi(\beta)$ is zero. It is shown in Appendix C.3 of \cite{kitaev2006anyons} that for any gapped Fock state $2\pi\nu_\psi(\beta)$ is an integer. Examples are known where this integer is nonzero, see e.g. \cite{TKNN}. 
    \begin{remark}
        It is a long-standing conjecture that for an arbitrary $U(1)$-invariant gapped state of a quantum lattice system the topological invariant $2\pi\nu_\psi(\beta)$ is a rational number.
    \end{remark}
\end{example}

\subsection{Topological invariants of gapped states on subsets of $\RR^n$}

Consider a $U(1)$-invariant gapped state of a quantum lattice system on  $W=\RR^2$ affinely embedded in $\RR^3$. If we regard it as a quantum lattice system on $\RR^3$, then the  topological invariant defined in Section \ref{sec:construction} vanishes for dimensional reasons. On the other hand, by contracting with the 1-cocycle of $\widehat W=S^1$ one obtains the Hall conductance of this system.

For a more non-trivial example, for any $n>2$ consider a finite graph $\Gamma\subseteq S^{n-1}$ whose edges are geodesics and take $W$ to be the cone with base $\Gamma$ and apex at an arbitrary point of $\RR^n$. The invariants of $U(1)$-invariant gapped states of quantum lattice systems on $W\subseteq\RR^n$ are labeled by generators of the free abelian group $H^1(\Gamma,\ZZ)$. This example goes beyond Chern-Simons field theory, since $W$ need not be smooth or even locally Euclidean.

\appendix

\section{0-chains}\label{appendix:LGA}
In this section we use the results of \cite{LocalNoether} to derive the properties of LGAs which we used in Section \ref{sec:symmetries}.

\subsection{0-chains on $\ZZ^n$}
First let us characterize $\mfkDal(U)$ in terms of 0-chains. A 0-chain on $\ZZ^n$ is an element $\dera = \{\dera_x\}_{x\in \ZZ^n} \in \prod_{x\in \ZZ^n}\mfkDal(\{x\})$ such that
\begin{align}
    \|\dera \|_{k} := \sup_{x\in \ZZ^n}\|\dera_x\|_{\{x\},k} < \infty. \label{eqn:c0frechetnorm}
\end{align}
We say a 0-chain $\dera$ is supported on $U$ if $\dera_x = 0$ whenever $x\notin U$, and write $C_0(U)$ for the set of $U$-supported 0-chains, endowed with the norms (\ref{eqn:c0frechetnorm}) for $k\ge 0$.
\begin{prop}\label{prop:confined-supported}
    Let $U\subseteq \RR^n$ be nonempty and let $U^1:= \{x\in \RR^n: d(x,U)\le 1\}$ be its 1-thickening. 
    \begin{enumerate}[i)]
        \item If $\derF\in \mfkDal(U)$ then $\derF = \partial \derf$ for a $U^1$-supported 0-chain $\derf$ with $\|\derf\|_k \le 2^k\|\derF\|_{U,k}$
        \item If $\derf \in C_0(U)$, then for any $Y\in\BB_n$ the sum
        \begin{align}
            (\partial \derf)^Y := \sum_{x\in \ZZ^n\cap U}\derf_x^Y
        \end{align}
        is absolutely convergent and defines a map $\partial : C_0(U)\to \mfkDal(U)$ with $\| \partial \derf\|_{U,k} \le C\|\derf\|_{k+2n+1}$, where the constant $C>0$ depends only on $n$.
    \end{enumerate}
\end{prop}
\begin{lemma}\label{lem:confined-supported}
    For any nonempty $U\subseteq \RR^n$ and any $Y\in\BB_n$ we have 
    \begin{align}
        1 + d(Y,U^1\cap \ZZ^n) + \diam(Y) \le 2(1+ d(Y,U) + \diam(Y))
    \end{align}
\end{lemma}
\begin{proof}
    Since the distances and thickenings involving $U$ depend only on the closure of $U$, we may assume that $U$ is closed. Choose $y\in Y$ and $u\in U$ with $d(y,u) = d(Y,U)$, and choose $z\in \ZZ^n$ with $d(u,z)\le 1$. Then since $z\in U^1\cap \ZZ^n$ we have
    \begin{align}
        d(Y,U^1\cap \ZZ^n) &\le \diam(Y) + d(y,z)\notag\\
        &\le \diam(Y) + d(y,u) + d(u,z)\notag\\
        &\le \diam(Y) + d(Y,U) + 1,
    \end{align}
    and the Lemma follows.
\end{proof}
\begin{proof}[Proof of Proposition \ref{prop:confined-supported}]
    $i)$. Suppose $\derF \in \mfkDal(U)$. Choose any total order on $U^1\cap \ZZ^n$ and for every $Y\in\BB_n$ let $j^*(Y)$ be the closest point to $Y$ in $U^1\cap \ZZ^n$, using the total order as a tiebreaker. For every $i\in \Lambda$, define
    \begin{align}
        \derf_i := \sum_{\substack{Y\in \BB_n \\ j^*(Y)=i}}\derF^Y. \label{eqn:0chain-remainder-expansion}
    \end{align}
    Then either $\derf_i^Y = 0$ or $d(Y,U^1\cap \ZZ^n) = d(Y,i)$ and $\|\derf_i^Y\|= \|\derF^Y\|$, and so 
    \begin{align}
        \|\derf_i\|_{k} &= \sup_{Y\in \BB_n\backslash \{\varnothing\}} (1+\diam(Y)+d(Y,\{i\}))^k \|\derf_i^Y\|\notag\\
        &= \sup_{Y\in \BB_n\backslash \{\varnothing\}}(1+\diam(Y)+d(Y,U^1\cap \ZZ^n))^k \|\derF^Y\|,
    \end{align}
    which by Lemma \ref{lem:confined-supported} is bounded by $2^k\|\derF\|_{U,k}$.
    \\\\
    $ii)$. Suppose that $\derf$ is a $U$-supported $0$-chain. For any $k\ge0$ we have
    \begin{align}
       \|(\partial \derf)^Y\| &\le \sum_{x\in \ZZ^n\cap U}\|\derf_x^Y\|\notag\\
       &\le \|\derf\|_{k+2n+1} \sum_{x\in \ZZ^n\cap U}(1+\diam(Y) + d(x,Y))^{-k-2n-1}\notag\\
       &\le \|\derf\|_{k+2n+1} (1+\diam(Y)+d(U,Y))^{-k}\sum_{x\in \ZZ^n\cap U}(1+\diam(Y) + d(x,Y))^{-2n-1}\notag\\
       &\le \|\derf\|_{k+2n+1} (1+\diam(Y)+d(U,Y))^{-k}(1+\diam(Y))^{-n}\sum_{x\in \ZZ^n}(1+ d(x,Y))^{-n-1}.
    \end{align}
    It is not hard to show that for any brick $Y$ the quantity $(1+\diam(Y))^{-n}\sum_{x\in \ZZ^n}(1+ d(x,Y))^{-n-1}$ is bounded by a constant $C$ depending only on $n$, which shows $\|\partial \derf\|\le C\|\derf\|_{k+2n+1}$.
\end{proof}
Proposition \ref{prop:confined-supported} will allow us to apply the results of \cite{LocalNoether} on 0-chains. The results in \cite{LocalNoether} are phrased in terms of the norms
\begin{align}
    \|\dera\|^{KS}_{x,k} := \sup_{r> 0}(1+r)^k\inf_{\derb \in \mfkdl(B_r(x))}\|\dera-\derb\|
\end{align}
where $\mfkdl(B_r(x))$ is the set of traceless anti-hermitian operators strictly localized on the ball of radius $r$ around $x\in \RR^n$. To apply their results we prove the equivalence of these norms:
\begin{lemma}\label{lem:norm-equivalence}
    For any $x\in \ZZ^n$ and $k>0$, the norms $\|\cdot \|^{KS}_{x,k}$ and $\|\cdot \|_{\{x\},k}$ obey
    \begin{align}
        \|\dera\|^{KS}_{x,k} &\le C \|\dera\|_{\{x\},k+2n+2} \label{eqn:norm-equivalence-1}\\
        \|\dera\|_{\{x\},k} &\le 4^k C' \|\dera\|^{KS}_{x,k} \label{eqn:norm-equivalence-2}
    \end{align}
    where $C,C'$ are constants depending only on $n$.
\end{lemma}
\begin{proof}
    Suppose $\|\dera\|_{\{x\},k+2n+2} <\infty$ and let $r>0$. Define $\derb := \sum_{X\subseteq B_r(x)}\dera^X$. Then
    \begin{align}
        \|\derb -\dera\| \le \|\dera\|_{\{x\},k+2n+2}\sum_{X\nsubseteq B_r(x)}(1+\diam(X)+d(x,X))^{-k-2n-2}
    \end{align}
    Since $X\nsubseteq B_r(x)$ means $\diam(X)+d(x,X)\ge r$, we continue this as follows
    \begin{align}
        &\le (1+r)^{-k}\|\dera\|_{\{x\},k+2n+2}\sum_{X\nsubseteq B_r(x)}(1+\diam(X)+d(x,X))^{-2n-2}\notag\\
        &\le C(1+r)^{-k}\|\dera\|_{\{x\},k+2n+2}
    \end{align}
    where in the last line we used Lemma \ref{lem:brick-sum}. This proves (\ref{eqn:norm-equivalence-1}). To prove (\ref{eqn:norm-equivalence-2}) suppose $\|\dera\|^{KS}_{x,k} <\infty$ and let $X$ be any brick. Set $r:= \lfloor(\diam(X)+d(x,X))/4\rfloor$. Then $X\nsubseteq B_r(x)$. Indeed, if $X\subseteq B_r(x)$ then $d(x,X)\le r$ and $\diam(X)\le 2r$, implying $\diam(X)+d(x,X)\le 3r$, which is impossible. Choose $\derb \in \mfkdl(B_r(x))$ with $\|\dera-\derb\|\le (1+r)^{-k}\|\dera\|_{x,k}^{KS}$. Since $X\nsubseteq B_r(x)$ we have $\derb^X=0$, and so
    \begin{align}
        \|\dera^X\| &= \|(\dera-\derb)^X\|\notag\\
        &\le 4^n\|\dera - \derb\|\notag\\
        &\le 4^n\|\dera\|_{x,k}^{KS}(1+r)^{-k}\notag\\
        &\le 4^{n+k}\|\dera\|_{x,k}^{KS}(1+\diam(X) + d(x,X))^{-k},
    \end{align}
    where in the second line we used \cite[Proposition C.1]{LocalNoether}.
\end{proof}

\subsection{Proof of Proposition \ref{prop:quasiadiabatic}}
We begin by describing the construction of $\mathcal{J}$ and $\mathcal{K}$. Suppose $\psi$ is gapped with Hamiltonian $\derH$ and gap $\Delta$, and write $\tau_t$ for the one-parameter family of LGAs obtained by exponentiating $\derH$. There exists\footnote{See for instance Lemma 2.3 in \cite{Bachmannetal}.} a function $w_\Delta:\RR\to \RR$ such that $w_\Delta(t) = O(|t|^{-\infty})$, and the Fourier transform\footnote{We use the convention $\widehat{f}(\xi) = \int f(t)e^{-i\xi t}dt$.} $\widehat{w_\Delta}$ is supported in the interval $[-\Delta/2,\Delta/2]$ and satisfies $\widehat{w_\Delta}(0)=1$. Define $W_\Delta(t)$ for $t>0$ as $W_\Delta(t) = -\int_{t}^\infty w_\Delta(u) du$ and extend to the whole real line as an odd function. Then we define $\mathcal{J}$ and $\mathcal{K}$ as the following integral transforms:
    \begin{align}
        \mathcal{J}(\derF) := \int w_\Delta(t) \tau_t(\derF) dt,\\
        \mathcal{K}(\derF) := \int W_\Delta(t) \tau_t(\derF) dt.
    \end{align}
\begin{proof}[Proof of Proposition \ref{prop:quasiadiabatic}]
    $\mathcal{J}$ and $\mathcal{K}$ correspond to $\mathscr{J}_{\derH,w_{\Delta}}$ and $\mathscr{J}_{\derH,W_{\Delta}}$ in Section 4.1 of \cite{LocalNoether}. Part $i)$ follows from the definition of $\mathcal{K}$ and the fact that $\derH$ preserves $\psi$. Part $ii)$ follows from Proposition \ref{prop:confined-supported} and Lemma \ref{lem:norm-equivalence}, together with \cite[Lemma F.1]{LocalNoether} (specifically line (177) therein). Part $iii)$ is \cite[line (72)]{LocalNoether}.
\end{proof}
\section{Inhomogeneous Maurer-Cartan equation} \label{sec:deformation}

Let $(\mfkM, \derB)$ be a pointed pronilpotent DGLA which is a limit of an inverse system of nilpotent pointed DGLAs $(\mfkM_N,\derB_N)$, $N\in\NN$. The set $\MC(\mfkM, \derB)$ of $\derB$-twisted MC elements has an additional equivalence relation. This section revolves around this additional structure leading eventually to a proof of Prop. \ref{prop:MCmain}.

We have morphisms $r_{N,K}:\mfkM_N\ra \mfkM_K$ for all $N>K$ and the DGLA $\mfkM$ is the inverse limit of the corresponding system of DGLAs. For any $N\in\NN$ let $r_N:\mfkM\ra\mfkM_N$ be the natural projection. Let $\fj_{N}=\ker r_{N,N-1}$. The sets of $\derB_N$-twisted MC elements of $\mfkM_N$ will be denoted $\MC(\mfkM_N,\derB_N)$.

A $\derB$-twisted MC element $\derp$ gives rise to a degree $-1$ derivation $\partial_\derp=\partial+\ad_\derp$ of $\mfkM$ which squares to zero (twisted differential).
\begin{lemma}\label{lma:MCcontinuity}
    If $\derp\in \MC(\mfkM_N,\derB_N)$, then $r_{N,K}(\derp)\in\MC(\mfkM_K,\derB_K)$ for all $K<N$. Further, $\derp\in \MC(\mfkM, \derB)$ iff  $\derp^N=r_N(\derp)\in \MC(\mfkM_N,\derB_N)\ \forall N\in\NN$. 
\end{lemma}
\begin{proof}
    Straightforward.
\end{proof}
Thus $\MC(\mfkM, \derB)$ is the inverse limit of the system of sets $\MC(\mfkM_N,\derB_N)$, $N\in\NN$.

We are going to define an equivalence relation on $\MC(\mfkM, \derB)$ and $\MC(\mfkM_N,\derB_N)$ for all $N$. This is done in the same way as for the ordinary (homogeneous) Maurer-Cartan equation \cite{GM, Manetti}. 

First, $\mfkM_{N,0}$ is a nilpotent Lie algebra, so there is a well-defined nilpotent Lie group $\exp(\mfkM_{N,0})$ with the group law given by the Campbell-Baker-Hausdorff formula. Similarly, $\mfkM_0$ is pronilpotent (i.e. is an inverse limit of a system of nilpotent Lie algebras), so the CBH formula defines a group $\exp(\mfkM_0)$. 

Second, there is a Lie algebra homomorphism from $\mfkM_0$ (resp. $\mfkM_{N,0}$) to the Lie algebras of affine-linear vector fields on $\mfkM_{-1}$ (resp. $\mfkM_{N,-1}$). This homomorphism maps $\dera\in \mfkM_0$ or $\mfkM_{N,0}$ to the affine-linear vector field
\begin{align}
    \xi_\dera(\derp)=[\dera,\derp] -\partial\dera,
\end{align}
where $\derp\in \mfkM_{-1}$ or $\mfkM_{N,-1}$.
Here we used the identification of the space of affine-linear vector fields on a vector space $V$ with the space of affine-linear maps $V\ra V$. These homomorphisms exponentiate to actions of the groups $\exp(\mfkM_0)$ and $\exp(\mfkM_{N,0})$ on $\mfkM_{-1}$ and $\mfkM_{N,-1}$ by affine-linear transformations. Explicitly, the actions are given by \cite{GM,Manetti}:
\begin{align}\label{eq:gaugeaction}
    \derp\mapsto \exp(\dera)* \derp=\exp(\ad_\dera)(\derp)+\frac{1-\exp(\ad_\dera)}{\ad_\dera}(\partial\dera).
\end{align} 
\begin{lemma}
    The actions of $\exp(\mfkM_0)$ on $\mfkM_{-1}$ (resp.  $\exp(\mfkM_{N,0})$ on $\mfkM_{N,-1}$) preserve the sets $\MC(\mfkM, \derB)$ (resp. $\MC(\mfkM_N,\derB_N)$).
\end{lemma}
\begin{proof}
    The proof in \cite{GM}, Section 1.3, applies just as well in the inhomogeneous case.
\end{proof}
We say that elements $\derp_1,\derp_2$ of $\MC(\mfkM, \derB)$ or $\MC(\mfkM_N,\derB_N)$ are equivalent if they belong to the same orbit of these actions. 
\begin{remark}
    By analogy with the homogeneous case, one can define a $\derB$-twisted Deligne groupoid as the transformation groupoid for the action of $\exp(\mfkM_0)$ on $\MC(\mfkM, \derB)$. Similarly, one can define "reduced" Deligne groupoids for every $N\in\NN$.
\end{remark}
We observe an easy but useful lemma. 
\begin{lemma}
    If $\derp_i\in \MC(\mfkM_N,\derB_N)$, $i=1,2$, are equivalent, then $r_{N,K}(\derp_1)$ is equivalent to $r_{N,K}(\derp_2)$ for all $K<N$.
\end{lemma}

Now come the interesting statements. Assume from now on that the DGLAs $\mfkM_N$ and $\mfkM$ are acyclic, that $\fj_N\subseteq\mfkM_N$ is central for all $N>1$, and that the morphisms $r_{N,N-1}$ are surjective for all $N>1$.
\begin{lemma}
With the above assumptions, the set $\MC(\mfkM, \derB)$ is non-empty if and only if $\MC(\mfkM_1, \derB_1)$ is non-empty.
\end{lemma}
\begin{proof}
    The only if statement follows from Lemma \ref{lma:MCcontinuity}. To prove the if direction, we use induction on $N$. Assume $\MC(\mfkM_{N-1},\derB_{N-1})$ is non-empty. Let $\derp_{N-1}\in \MC(\mfkM_{N-1},\derB_{N-1})$. Pick $\tilde\derp_N\in r_{N,N-1}^{-1}(\derp_{N-1})$. Since $\derp_{N-1}$ is a $\derB_{N-1}$-twisted MC element and $r_{N,N-1}(\derB_N)=\derB_{N-1}$, we must have  
    \begin{align}
\partial\tilde\derp_N+\frac12[\tilde\derp_N,\tilde\derp_N]=\derB_N+\derq_N
    \end{align}
    for some $\derq_N\in\fj_N$. We look for a solution of the  $\derB_N$-twisted MC equation of the form $\derp_N=\tilde\derp_N+\derb_N$ where $\derb_N\in\fj_N$. Taking into account that $\fj_N$ is central, the inhomogeneous MC equation reduces to $\partial\derb_N=-\derq_N$.
    Since by assumption $\fj_N,\mfkM_N$, and $\mfkM_{N-1}$ form a short exact sequence and the latter two are acyclic, so is $\fj_N$. Hence such a $\derb_N$ exists. This completes the inductive step proving that  $\MC(\mfkM_N,\derB_N)\neq \varnothing$ for all $N$. Moreover, we also proved that the morphisms $\MC(\mfkM_N,\derB_N)\ra\MC(\mfkM_{N-1},\derB_{N-1})$ are surjective. Therefore by Lemma \ref{lma:MCcontinuity} $\MC(\mfkM,\derB)$ is non-empty.
\end{proof}
The above lemma proves the first part of Prop. \ref{prop:MCmain}. Indeed when $[\mfkM_1,\mfkM_1]=0$, $\MC(\mfkM_1, \derB_1)\ne \varnothing$ is implied by acyclicity. Next we show that with the above assumption on $\mfkM_N$ and $\mfkM$ all $\derB$-twisted MC elements are equivalent. 
\begin{lemma}\label{lma:GMlemma28}
    Suppose $\dera,\derb\in \mfkM_{N,0}$ and $\dera\in \fj_{N}$. Then for any $\derp\in\mfkM_{N,-1}$ one has
    \begin{align}
        \exp(\derb+\dera)*(\derp)=\exp(\derb)*(\derp)-\partial\dera.
    \end{align}
\end{lemma}
\begin{proof}
    See \cite{GM}, Lemma 2.8.
\end{proof}

\begin{lemma}\label{lma:MCequalimpliesequiv}
Let $\derp_i\in\MC(\mfkM_N,\derB_N)$, $i=1,2$ such that $r_{N,N-1}(\derp_1)=r_{N,N-1}(\derp_2)$. Then $\derp_1$ and $\derp_2$ are equivalent. 
\end{lemma}
\begin{proof}
    Let $\derq=\derp_2-\derp_1\in \mfkM_{N,-1}$. By assumption, $\derq\in \fj_{N}$. Moreover, $\partial\derq=0$. Indeed:
    \begin{align}
    \partial\derq=-\frac12[\derp_2,\derp_2]+\frac12[\derp_1,\derp_1]=-[\derq,\derp_1]-\frac12[\derq,\derq]=0.
    \end{align}
    The last equality is because $\derq\in\fj_N$ is central.
    By acyclicity of $\fj_N$, we have $\derq=\partial\dera$ for some $\dera\in \fj_{N,0}$. Then Lemma \ref{lma:GMlemma28} implies that $\exp(\dera)\in \exp(\mfkM_{N,0})$ maps $\derp_2$ to $\derp_1$.
\end{proof}

\begin{lemma}\label{lma:uniqueorbitMCN}
Let  
$\tilde{\derp}_i\in\MC(\mfkM_N,\derB_N)$, $i=1,2,$ be such that $\derp_1=r_{N,N-1}(\tilde{\derp}_1)$ and $\derp_2=r_{N,N-1}(\tilde{\derp}_2)$ are equivalent. Then $\tilde{\derp}_1$ and $\tilde{\derp}_2$ are equivalent.
\end{lemma}
\begin{proof}
    Let $\dera\in\mfkM_{N-1,0}$ be an equivalence between $\derp_1$ and $\derp_2$, i.e. $\exp(\dera)*\derp_2=\derp_1.$ Let $\tilde{\dera}\in \mfkM_{N,0}$ be any lift of $\dera$. Then
    $r_{N,N-1}(\exp(\tilde{\dera})*\tilde{\derp}_2)=r_{N,N-1}(\tilde{\derp}_1)$. By Lemma \ref{lma:MCequalimpliesequiv}, 
    $\exp(\tilde{\dera})*\tilde{\derp}_2$ is equivalent to $\tilde{\derp}_1$, therefore $\tilde{\derp}_2$ is equivalent to $\tilde{\derp}_1$.
\end{proof}
\begin{prop}
     For any $N\in\NN$ all elements of $\MC(\mfkM_N,\derB_N)$ are equivalent. All elements of $\MC(\mfkM, \derB)$ are equivalent.
\end{prop}
\begin{proof}
    The first statement is proved by  induction on $N$, where the inductive step is Lemma \ref{lma:uniqueorbitMCN}.
    
    To prove the second statement, let $\derp,\derp'\in \MC(\mfkM,\derB)$ and write
$\derp_N:=r_N(\derp)$, $\derp'_N:=r_N(\derp')\in \MC(\mfkM_N,\derB_N)$.
By the first part, for each $N$ there exists $g_N\in \exp(\mfkM_{N,0})$ with
$g_N*\derp'_N=\derp_N$.

We may choose $g_N$ compatibly: pick $g_1$ arbitrarily, and assume $g_{N-1}$ chosen.
Choose any lift $\tilde g_N\in \exp(\mfkM_{N,0})$ with $r_{N,N-1}(\tilde g_N)=g_{N-1}$, and set
$\tilde\derq_N:=\tilde g_N*\derp'_N$. Then
$r_{N,N-1}(\tilde\derq_N)=g_{N-1}*\derp'_{N-1}=\derp_{N-1}=r_{N,N-1}(\derp_N)$,
so Lemma~\ref{lma:MCequalimpliesequiv} yields $h_N\in \exp(\ker(r_{N,N-1})_0)$ with
$h_N*\tilde\derq_N=\derp_N$. Putting $g_N:=h_N\tilde g_N$ gives
$r_{N,N-1}(g_N)=g_{N-1}$ and $g_N*\derp'_N=\derp_N$.

Thus $(g_N)_N$ defines $g\in \varprojlim_N \exp(\mfkM_{N,0})\cong \exp(\mfkM_0)$, and since the
gauge action is compatible with the projections, we have $r_N(g*\derp')=g_N*\derp'_N=\derp_N=r_N(\derp)$
for all $N$, hence $g*\derp'=\derp$ in $\mfkM$. Therefore $\MC(\mfkM,\derB)$ has a single gauge orbit.

\end{proof}

\begin{theorem}
    Let $\derp\in \MC(\mfkM, \derB)$. Then the homology class of $[\derp,\derp]$ in $H_\bullet(\clcomm)$ is independent of the choice of $\derp$. 
\end{theorem}
\begin{proof}
    Let $\derp,\derp'\in\MC(\mfkM, \derB)$ and let $\dera\in \mfkM_0$ be an equivalence between $\derp$ and $\derp'$. Since $\derp,\derp'$ satisfy the inhomogeneous Maurer-Cartan equation and $\derB$ is a cycle, $[\derp,\derp]$ and $[\derp',\derp']$ are cycles as well. From (\ref{eq:gaugeaction}) we have
    \begin{align}
        \derp'-\derp=-\partial\dera+\derf,
    \end{align}
    where 
\begin{align}
    \derf:=\sum_{k=1}^\infty \frac{\ad^k_\dera (\derp)}{k!}-\sum_{k=1}^\infty \frac{\ad^k_\dera (\partial\dera)}{(k+1)!}\in \clcomm.
\end{align}
    Thus $[\derp,\derp]-[\derp',\derp']=2\partial(\derp'-\derp)= 2 \partial \derf$ is $\partial$-exact in $\clcomm$.
\end{proof}

\bibliographystyle{unsrt} 
\bibliography{Bibliography.bib} 

 \end{document}